\documentclass[pra,twocolumn,superscriptaddress,
showpacs,preprintnumbers,amsmath,amssymb]{revtex4}

\usepackage{graphicx}
\usepackage[caption=false]{subfig}
\usepackage{amsthm}
\usepackage{tensor}
\usepackage{color}
\usepackage[all]{xy}
\usepackage{tikz}
\usepackage{dsfont}
\usepackage{times,txfonts}
\usetikzlibrary{positioning}
\usepackage{braket}

\newtheorem{Thm}{Theorem}
\newtheorem{Lem}[Thm]{Lemma}
\newtheorem{Prop}[Thm]{Proposition}

\theoremstyle{definition}

\newtheorem{Rem}[Thm]{Remark}

\newcommand{\ketbra}[2]{\left|{#1}\right\rangle\left\langle{#2}\right|}
\newcommand{\sbigotimes}
{%
\mathop{\mathchoice{\textstyle\bigotimes}{\bigotimes}{\bigotimes}{\bigotimes}}%
}
\newcommand{\Tr}{\mathop{\mathrm{Tr}}\nolimits}

\begin{document}

\title{Quantum state rotation: Circularly transferring quantum states of multiple users}

\author{Yonghae Lee} \email{yonghaelee@khu.ac.kr}
\affiliation{
Department of Mathematics and Research Institute for Basic Sciences,
Kyung Hee University, Seoul 02447, Korea}

\author{Hayata Yamasaki} \email{hayata.yamasaki@gmail.com}
\affiliation{
Photon Science Center, Graduate School of Engineering, The University of Tokyo, 7-3-1 Hongo, Bunkyo-ku, Tokyo 113-8656, Japan}
\affiliation{Institute for Quantum Optics and Quantum Information --- IQOQI Vienna, Austrian Academy of Sciences, Boltzmanngasse 3, 1090 Vienna, Austria}

\author{Soojoon Lee} \email{level@khu.ac.kr}
\affiliation{
Department of Mathematics and Research Institute for Basic Sciences,
Kyung Hee University, Seoul 02447, Korea}

\date{\today}

\begin{abstract}
Quantum state exchange is a quantum communication task for two users
in which the users faithfully exchange their respective parts of an initial state
under the asymptotic scenario.
In this work,
we generalize the quantum state exchange task
to a quantum communication task for $M$ users
in which the users circularly transfer their respective parts of an initial state.
We assume that
every pair of users may share entanglement resources,
and they use local operations and classical communication
in order to perform the task.
We call this generalized task the (asymptotic) quantum state rotation.
First of all,
we formally define the quantum state rotation task and its optimal entanglement cost,
which means the least amount of total entanglement
required to carry out the task.
We then present lower and upper bounds
on the optimal entanglement cost,
and provide conditions for
zero optimal entanglement cost.
Based on these results,
we find out a difference between the quantum state rotation task for three or more users
and the quantum state exchange task.
\end{abstract}

\maketitle

\section{Introduction}

In quantum information theory,
some quantum communication tasks~\cite{BBCJPW93,HOW05,HOW06,D06,O08,DY08,YD09,ADHW09},
such as quantum teleportation~\cite{BBCJPW93}
and quantum state merging~\cite{HOW05,HOW06},
commonly deal with a two-user setting in which a quantum state is transmitted 
from one user to the other.
In these quantum communication tasks,
the users are determined as either a sender or a receiver, as depicted in Fig.~\ref{fig:SendingQSE}(a),
and it is assumed that the users are in each other's laboratories far apart.
So, in order to successfully perform the tasks, 
it is required to consume non-local resources, such as ebits and bit channels.

One of the research topics related to quantum communication tasks is 
to find out the minimal amounts of the non-local resources consumed during the tasks.
Such research is considered to be important in quantum information theory, since the minimal amounts can often be represented as entropic quantities, such as the von Neumann entropy and the quantum conditional entropy~\cite{W13}, and hence provides a way to interpret these quantities from an operational viewpoint.
For example, the quantum conditional entropy $H(A|B)$ of a quantum state $\rho_{AB}$ can be operationally interpreted as
the minimal amount of entanglement needed in the quantum state merging~\cite{HOW05,HOW06}
in which Alice and Bob share parts $A$ and $B$ of the quantum state $\rho_{AB}$,
respectively, and Alice's part $A$ is merged to Bob
via entanglement-assisted local operations and classical communication (LOCC).
Note that, in the quantum state merging task, Bob does not transmit his part $B$ to Alice, but he can use it as quantum side information .

\begin{figure}
\includegraphics[clip,width=\columnwidth]{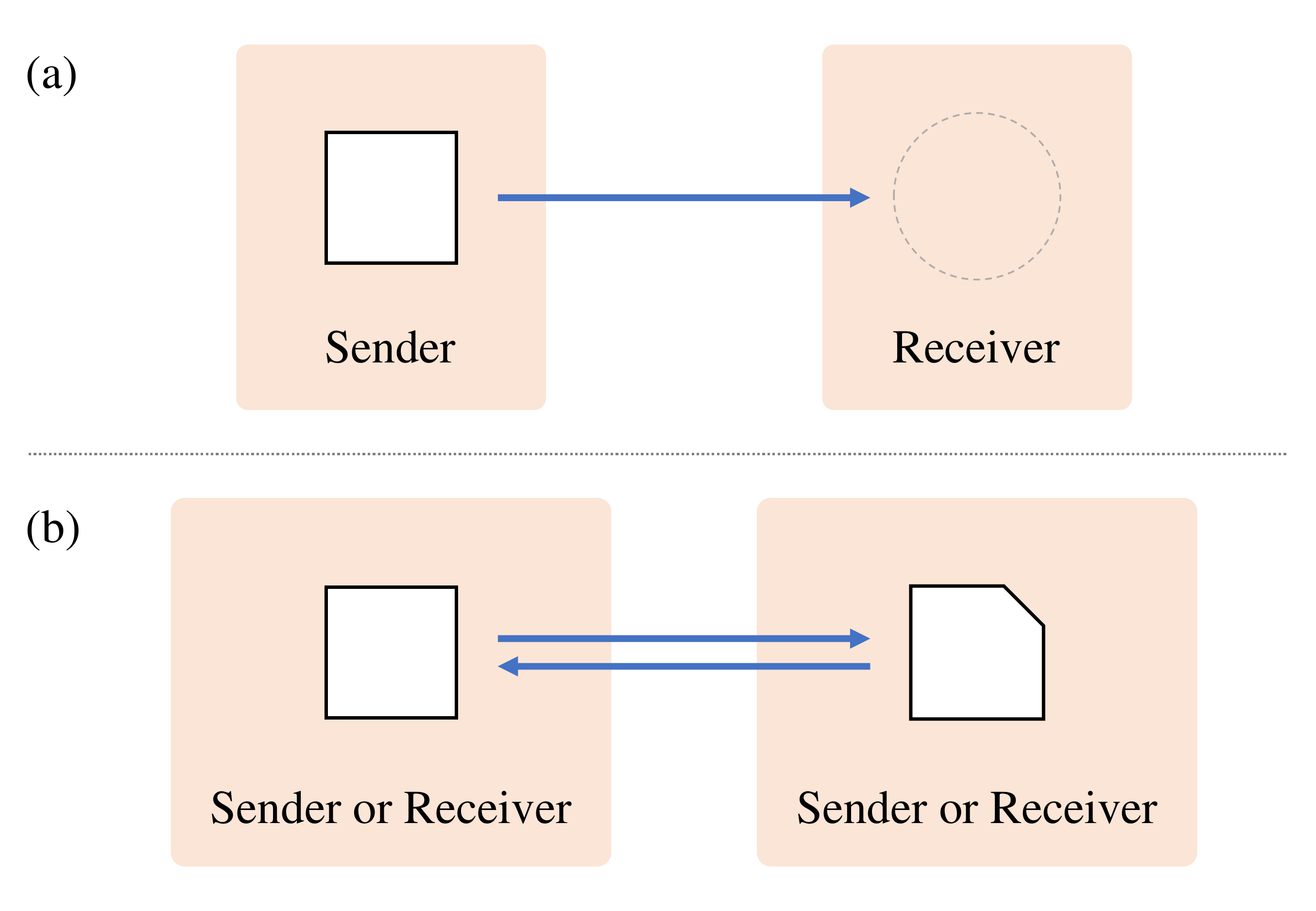}
\caption{
In these illustrations, quantum states are represented as polygons,
and the users transmit their quantum states in the directions of arrows;
(a) Illustration of transmitting a quantum state from a sender to a receiver;
(b) Illustration of exchanging two quantum states of two users:
In this case, each user is not only a sender but also a receiver.
}
\label{fig:SendingQSE}
\end{figure}

Quantum state exchange (QSE)~\cite{OW08,LTYAL19,LYAL19},
on the other hand,
is a more complex quantum communication task
in which two users share parts $A$ and $B$ of a quantum state $\rho_{AB}$,
and they exchange their respective parts with each other
by means of entanglement-assisted LOCC.
Thus, 
the users of the QSE task do not take one of the roles of a sender and a receiver, but both, as depicted in Fig.~\ref{fig:SendingQSE}(b).
The main concern of the QSE is to figure out the minimal amount of entanglement between the users
under the assumption that classical communication is free.
The authors of the original QSE task~\cite{OW08} named the minimal amount of entanglement ``uncommon information.''
Unlike other quantum communication tasks~\cite{BBCJPW93,HOW05,HOW06,D06,O08,DY08,YD09,ADHW09},
an exact value of the uncommon information is unknown to date.

In this work,
we introduce a new quantum communication task involving three or more users,
which is similar to the rotation in volleyball.
In a volleyball game, players rotate on the court when their team makes a serve.
Similarly to this rotation,
one may think of users of the new task and their quantum states
as locations of the court and the players, respectively.
More specifically,
$M$ users of the new task transmit their respective quantum states
from the $i^{\mathrm{th}}$ user to the $(i+1)^{\mathrm{th}}$ user
via entanglement-assisted LOCC, while keeping entanglement with an environment system.
We call this task quantum state rotation (QSR).
We provide a simple illustration of the QSR task for three users in Fig.~\ref{fig:ThreeQSR}(a).
Note that the QSR task for two users,
i.e., $M=2$, is nothing but the QSE task described in Fig.~\ref{fig:SendingQSE}(b),
since the $1^{\mathrm{st}}$ user transmits his/her quantum state to the $2^{\mathrm{nd}}$ user,
and the $1^{\mathrm{st}}$ user also receives the $2^{\mathrm{nd}}$ user's quantum state.
So the QSR task can be regarded as one possible generalization of the QSE task.

\begin{figure}
\subfloat[]{%
  \includegraphics[clip,width=.735\columnwidth]{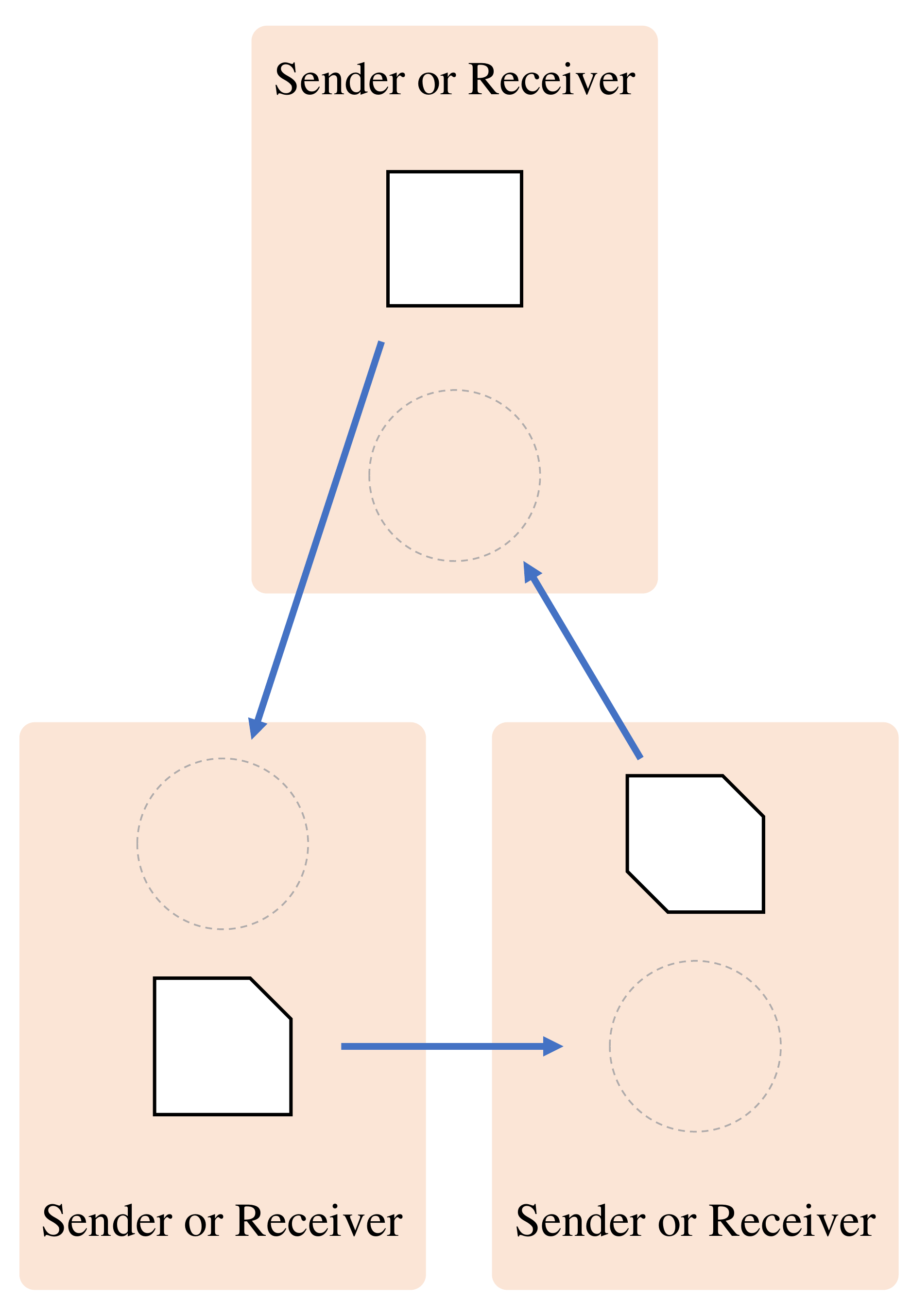}%
}

\subfloat[]{%
  \includegraphics[clip,width=.735\columnwidth]{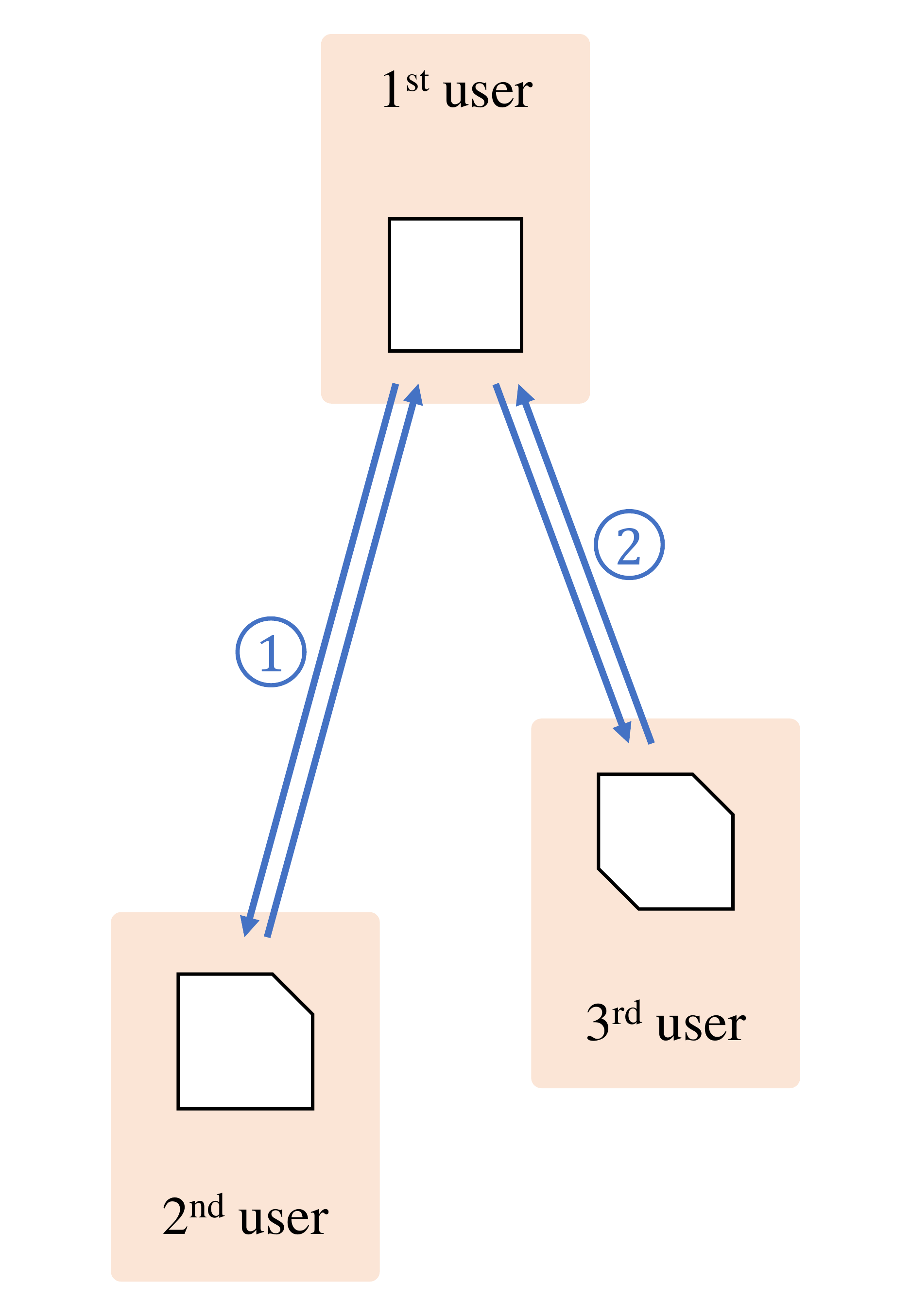}%
}
\caption{In each illustration, quantum states are represented as polygons,
users transmit their quantum states in the directions of arrows;
(a) Concept of the quantum state rotation task for three users:
The users are not only a sender but also a receiver;
(b) Performing the quantum state rotation task for three users with a combination of the quantum state exchange tasks:
The $1^{\mathrm{st}}$ user and the $2^{\mathrm{nd}}$ user firstly exchange their quantum states,
and then the $1^{\mathrm{st}}$ user and the $3^{\mathrm{rd}}$ user perform the quantum state exchange task again.
Consequently, these two quantum state exchange tasks carry out the quantum state rotation task
for three users.
}
\label{fig:ThreeQSR}
\end{figure}

Intuitively,
the QSR task for three or more users may be achievable
by sequentially performing the QSE tasks for several pairs of the users.
For example,
as depicted in Fig.~\ref{fig:ThreeQSR}(b),
two QSE tasks among three users can carry out the QSR task for three users.
So one may think it is not necessary to conduct a study on the QSR task.
However,
from the perspective of entanglement resources, it is unclear whether the combination of QSE tasks gives the minimal amount of total entanglement needed in the QSR task,
since the uncommon information of the QSE task is unknown.
Moreover,
the QSR for three or more users might exhibit intrinsic properties of multi-partite entanglement, which cannot be understood only by a straightforward generalization of the analysis of the QSE for two users.
On this account,
the main parts of our work focus on analyzing the minimal amount of total entanglement
consumed among the $M$ users for achieving the QSR task.

This paper is organized as follows.
In Sec.~\ref{sec:DEFs},
we provide formal descriptions of the QSR task,
the achievable total entanglement rate,
and the optimal entanglement cost.
In Sec.~\ref{sec:LB} and Sec.~\ref{sec:AUB},
we present lower and upper bounds,
respectively,
on the optimal entanglement cost of the QSR task.
In Sec.~\ref{sec:CZATER},
we present conditions obtained by zero optimal entanglement costs
and zero achievable total entanglement rates.
Based on these results, in Sec.~\ref{sec:DbRaE},
we show that a property of the QSE task does not hold
in the QSR task for three or more users.
In Sec.~\ref{sec:Exmpls},
we consider two settings of the QSR tasks and investigate what the users should do to reduce the optimal entanglement cost in each setting.
In the first setting,
some users do not have to participate in the QSR task.
In the second,
some users can cooperate to achieve the task by performing global operations over the users.
In Sec.~\ref{sec:conclusion},
we summarize and discuss our results.

\section{Definitions} \label{sec:DEFs}

In this section,
we explain notations used throughout this paper,
and we describe formal definitions of the QSR task
and its optimal entanglement cost.

\subsection{Notations: Systems, states, channels, and entropies} \label{subsec:Ns}

We assume that all Hilbert spaces $\mathcal{H}$ in this paper are finite-dimensional,
and let $d_X$ denote the dimension of the Hilbert space $\mathcal{H}_X$ representing a quantum system $X$.
A composite quantum system of two quantum systems $X$ and $Y$
is described by the tensor product $\mathcal{H}_X\otimes \mathcal{H}_Y$ of
the Hilbert spaces $\mathcal{H}_X$ and $\mathcal{H}_Y$.
For the sake of convenience,
the composite quantum system is denoted by $X\otimes Y$ or $XY$,
and $d_X$ is called the dimension of the quantum system $X$.

Let $\mathcal{D}(\mathcal{H})$ be the set of density operators on a Hilbert space $\mathcal{H}$,
i.e.,
$\mathcal{D}(\mathcal{H})
=\{\rho\in\mathcal{L}(\mathcal{H}) : \rho\ge0, \Tr\rho=1\}$,
where $\mathcal{L}(\mathcal{H})$ is the set of linear operators on $\mathcal{H}$.
For a Hilbert space $\mathcal{H}_X$ representing a quantum system $X$,
we use notations $\mathcal{D}(X)$ and $\mathcal{L}(X)$ instead of $\mathcal{D}(\mathcal{H}_X)$ and $\mathcal{L}(\mathcal{H}_X)$, respectively, in order to emphasize the quantum system $X$.
The elements of the set $\mathcal{D}(\mathcal{H})$ are called quantum states.
If a quantum state $\rho$ is a rank-1 projector, i.e., it is represented as
\begin{equation}
\psi \coloneqq\ket{\psi}\bra{\psi},
\end{equation}
where $\ket{\psi}$ is a normalized vector on the Hilbert space $\mathcal{H}$, the quantum state is said to be pure.
For pure quantum states $\ket{\phi}\bra{\phi}$, we also call the unit vector $\ket{\phi}$ a pure quantum state.

For quantum systems $X$ and $Y$,
a map $\mathcal{N}\colon\mathcal{L}(X)\to\mathcal{L}(Y)$ is called a quantum channel
if it is linear, completely positive, and trace-preserving~\cite{W13}.
As a special case of quantum channels, $\mathrm{id}_{\mathcal{L}(X)}$ is the identity map on $\mathcal{L}(X)$.
When $\mathcal{L}(X)=\mathcal{L}(Y)$,
$\mathrm{id}_{\mathcal{L}(X) \to \mathcal{L}(Y)}$ means the identity map from $\mathcal{L}(X)$ to $\mathcal{L}(Y)$.
For reference, $\mathds{1}_{X}$ is the identity matrix on the quantum system $X$.

The von Neumann entropy $H(\rho)$ of a (pure) quantum state $\rho$ on a quantum system $X$
is defined as $H(\rho)=H(X)_{\rho}=-\Tr\rho\log\rho$.
For a (pure) quantum state $\sigma$ on a bipartite quantum system $XY$,
the von Neumann entropy $H(X)_{\sigma}$ of $\sigma$ on the subsystem $X$ is calculated as
$H(X)_{\sigma}=H(\Tr_{Y}\sigma)$.
Then the quantum conditional entropy $H(X|Y)_{\sigma}$ and the quantum mutual information $I(X;Y)_{\sigma}$ 
of the bipartite quantum state $\sigma$ are given by
\begin{eqnarray}
H(X|Y)_{\sigma}
&=& H(XY)_{\sigma}-H(Y)_{\sigma}, \\
I(X;Y)_{\sigma}
&=& H(X)_{\sigma}+H(Y)_{\sigma}-H(XY)_{\sigma}.
\end{eqnarray}

Finally, the number of users of the QSR task is denoted
by a natural number $M \ge 2$.
If the $i^{\mathrm{th}}$ user has a quantum system $X_i$ for each $i$
in the set $[M]=\{1,2,\ldots,M\}$,
then the addition of two indices is defined modulo $M$, with offset $1$.

\subsection{Formal description of quantum state rotation} \label{subsec:FDQSR}

Before describing definitions of the QSR task,
we briefly explain a conception of the QSR task.
The QSR is a quantum communication task for $M$ users.
The users initially share an $M$-partite quantum state,
and circularly transfer their respective quantum states
from the $i^{\mathrm{th}}$ user to the $(i+1)^{\mathrm{th}}$ user
via entanglement-assisted LOCC.

\begin{figure}
\includegraphics[clip,width=\columnwidth]{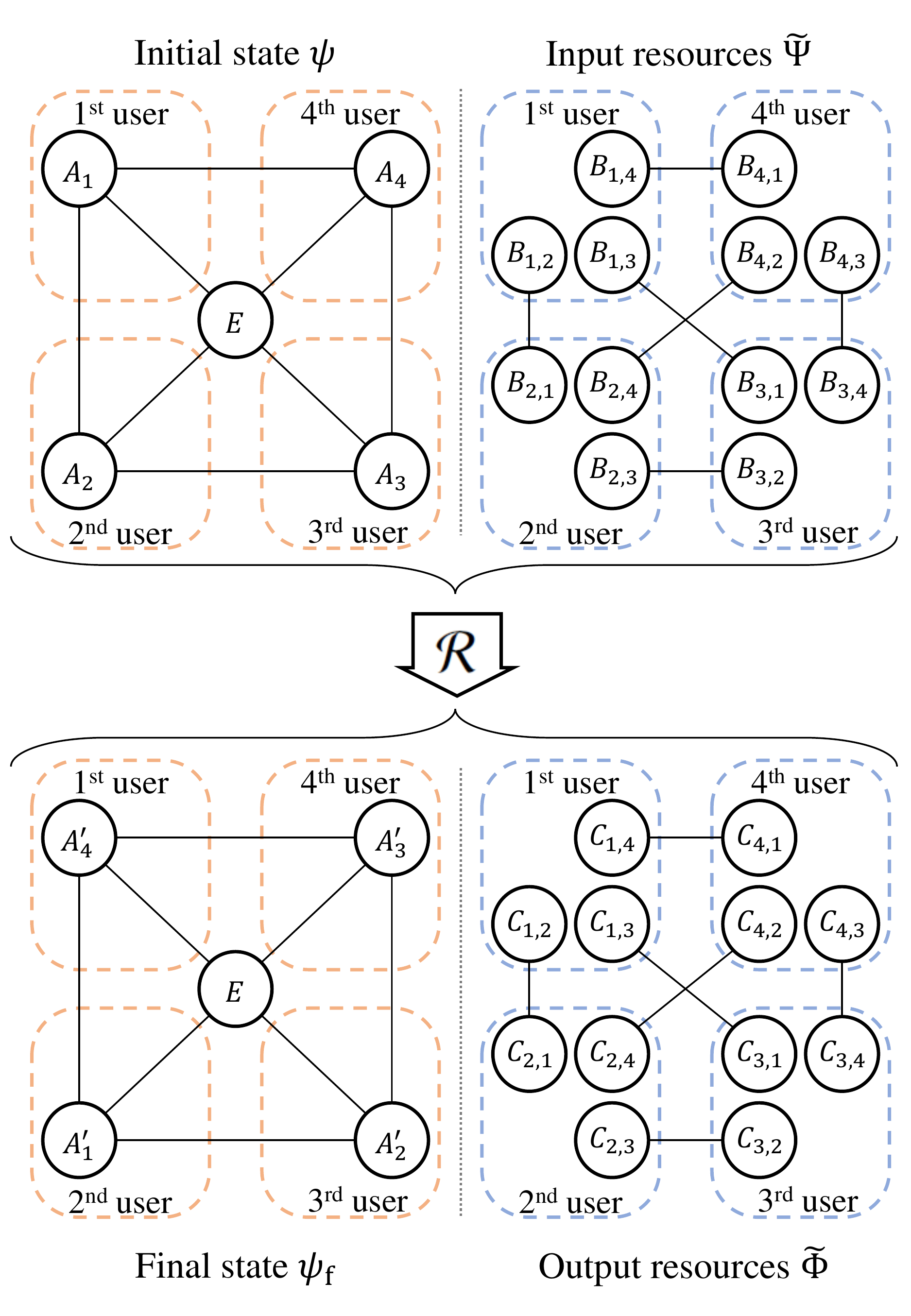}
\caption{
Illustration of the quantum state rotation task for four users:
Circles indicate quantum systems for the task, and correlations among the quantum systems are represented 
by lines connecting them.
In this task,
four users want to transform an initial state $\ket{\psi}_{A_1A_2A_3A_4E}$
into a final state $\ket{\psi_{\mathrm{f}}}_{A'_1A'_2A'_3A'_4E}$.
To carry out this task,
they apply LOCC $\mathcal{R}$
to the initial state $\psi$ and input entanglement resources $\tilde{\Psi}$,
while they cannot apply any operations on the environment system $E$.
After the task,
the $i^{\mathrm{th}}$ user's quantum state is transmitted to the $(i+1)^{\mathrm{th}}$ user's quantum system $A'_i$,
and they can gain output entanglement resources $\tilde{\Phi}$ from the task.
The allocation of entanglement resources, such as $\tilde{\Psi}$ and $\tilde{\Phi}$, is called
the complete entanglement allocation.
}
\label{fig:QSR_task}
\end{figure}

To be specific,
let $\ket{\psi}_{AE}$ be the \emph{initial state}
of the QSR task,
where $A$ is an $M$-partite quantum system with $A = A_1A_2\cdots A_M$,
and $E$ is the environment system.
Assume that the $i^{\mathrm{th}}$ user holds a quantum subsystem $A_i$ of $A$,
so that the $i^{\mathrm{th}}$ user has the $i^{\mathrm{th}}$ part of the initial state $\ket{\psi}_{AE}$.
Let us now consider an $M$-partite quantum system $A'$ with $A' = A'_1A'_2\cdots A'_M$ and $\mathcal{H}_{A'_i}=\mathcal{H}_{A_i}$,
and assume that the $i^{\mathrm{th}}$ user holds a quantum subsystem $A'_{i-1}$ of $A'$.
Then
the \emph{final state} $\ket{\psi_{\mathrm{f}}}_{A'E}$ of the QSR task is defined by using $\ket{\psi}_{AE}$
as follows:
\begin{equation}
\psi_{\mathrm{f}}
=
\left(
\sbigotimes_{i=1}^{M} \mathrm{id}_{\mathcal{L}(A_i)\to \mathcal{L}(A'_i)} \otimes \mathrm{id}_{\mathcal{L}(E)}
\right)
(\psi),
\end{equation}
which means that
the $i^{\mathrm{th}}$ user's quantum state on the quantum system $A_i$ is transferred to the $(i+1)^{\mathrm{th}}$ user's quantum system $A'_i$.
The initial state and the final state for four users,
i.e.,
$M=4$, are presented in Fig.~\ref{fig:QSR_task}.

In the QSR task,
the users make use of LOCC assisted by shared entanglement,
in order to transform the initial state $\ket{\psi}_{AE}$
into the final state $\ket{\psi_{\mathrm{f}}}_{A'E}$.
In this work,
we assume that every two of the $M$ users of the QSR task
may share an entanglement resource of varying dimensions.
More specifically, for each $i\neq j$,
the $i^{\mathrm{th}}$ user and the $j^{\mathrm{th}}$ user have additional quantum systems $B_{i,j}$ and $B_{j,i}$,
respectively, whose dimensions are the same, and the two users share a bipartite maximally entangled state $\ket{\Psi_{i,j}}$ on the quantum systems $B_{i,j}B_{j,i}$ given by
\begin{equation}
\ket{\Psi_{i,j}}
=\frac{1}{\sqrt{d_{B_{i,j}}}}\sum_{k=0}^{d_{B_{i,j}}-1}
\ket{k}_{B_{i,j}}
\otimes
\ket{k}_{B_{j,i}}.
\end{equation}

As in other quantum communication tasks~\cite{HOW05,HOW06,D06,O08,DY08,YD09,ADHW09,OW08,LTYAL19,LYAL19},
the users of the QSR task may gain extra entanglement resources from the QSR task.
To describe these entanglement resources,
we also assume that, for each $i\neq j$,
the $i^{\mathrm{th}}$ user and the $j^{\mathrm{th}}$ user have quantum systems $C_{i,j}$ and $C_{j,i}$,
respectively, with $d_{C_{i,j}}=d_{C_{j,i}}$,
and they share a bipartite maximally entangled state $\ket{\Phi_{i,j}}$ on the quantum systems $C_{i,j}C_{j,i}$
as an entanglement resource after the QSR task,
i.e.,
\begin{equation}
\ket{\Phi_{i,j}}
=\frac{1}{\sqrt{d_{C_{i,j}}}}\sum_{k=0}^{d_{C_{i,j}}-1}
\ket{k}_{C_{i,j}}
\otimes
\ket{k}_{C_{j,i}}.
\end{equation}

Let $\tilde{\Psi}$ and $\tilde{\Phi}$ be global quantum states representing all entanglement resources
shared among the $M$ users before and after the QSR task,
respectively,
which are defined as
\begin{equation} \label{Eq:ERs} 
\tilde{\Psi} = \sbigotimes_{\substack{ i,j\in[M] \\ i<j }} {\Psi}_{i,j}
\quad \mathrm{and} \quad
\tilde{\Phi} = \sbigotimes_{\substack{ i,j\in[M] \\ i<j }} {\Phi}_{i,j}.
\end{equation}
The shapes of entanglement resources $\tilde{\Psi}$ and $\tilde{\Phi}$ 
correspond to a complete graph,
if we regard the $M$ users and their entanglement resources as vertices and edges of a graph,
respectively.
In this work,
we call such a resource allocation of entangled states
the \emph{complete entanglement allocation},
and the complete entanglement allocation for four users is described in Fig.~\ref{fig:QSR_task}.

In the QSR task, a quantum channel
\begin{equation} \label{eq:DomainRange}
\mathcal{R}\colon
\mathcal{L}\left(\sbigotimes_{i=1}^M A_i B_i\right)
\longrightarrow
\mathcal{L}\left(\sbigotimes_{i=1}^M A'_i C_i\right)
\end{equation}
is called the \emph{QSR protocol} of the initial state $\ket{\psi}_{AE}$ with error $\varepsilon$,
if it is performed by LOCC among the $M$ users and satisfies 
\begin{equation}
\left\|
\left( \mathcal{R} \otimes \mathrm{id}_{\mathcal{L}(E)} \right)
\left( \psi \otimes \tilde{\Psi} \right)
-
\psi_{\mathrm{f}} \otimes \tilde{\Phi}
\right\|_{1}
\le
\varepsilon,
\end{equation}
where quantum systems $B_i$ and $C_i$ are defined by
\begin{equation} \label{eq:ERSystems}
B_i = \sbigotimes_{j\in[M]\setminus\{i\}} B_{i,j}
\quad \mathrm{and} \quad
C_i = \sbigotimes_{j\in[M]\setminus\{i\}} C_{i,j},
\end{equation}
and $\|\cdot\|_1$ is the trace norm~\cite{W13}.

\subsection{Optimal entanglement cost of quantum state rotation} \label{subsec:QEC}

To investigate asymptotic limits for the total amount of entanglement,
we consider a sequence $\{\mathcal{R}_n\}_{n\in\mathbb{N}}$ 
of QSR protocols $\mathcal{R}_n$ of $\psi^{\otimes n}$
with error $\varepsilon_n$,
where $\psi^{\otimes n}$ indicates the $n$ copies of the initial state $\psi$.
We call the case dealing with 
sequences 
of QSR protocols an \emph{asymptotic scenario}.

According to the number of the initial state and the users' strategies in the asymptotic scenario, the total amount of entanglement consumed/gained among the users can differ. 
To reflect this, it is assumed that, for each $n$, 
the $i^{\mathrm{th}}$ user and the $j^{\mathrm{th}}$ user have additional quantum systems $B^{(n)}_{i,j}C^{(n)}_{i,j}$ and $B^{(n)}_{j,i}C^{(n)}_{j,i}$,
respectively, where $d_{B^{(n)}_{i,j}}=d_{B^{(n)}_{j,i}}$ and $d_{C^{(n)}_{i,j}}=d_{C^{(n)}_{j,i}}$,
and the two users share bipartite maximally entangled states $\ket{\Psi^{(n)}_{i,j}}$ and $\ket{\Phi^{(n)}_{i,j}}$ on the quantum systems $B^{(n)}_{i,j}B^{(n)}_{j,i}$ and $C^{(n)}_{i,j}C^{(n)}_{j,i}$, respectively.
In this case, for each $n$, the complete entanglement allocations before and after the QSR protocol $\mathcal{R}_n$ of $\psi^{\otimes n}$ with error $\varepsilon_n$
are represented as $\tilde{\Psi}_n$ and $\tilde{\Phi}_n$, respectively.

For the initial state $\psi$ and the sequence $\{\mathcal{R}_n\}$, 
we define the \emph{segment entanglement rate} $e_{i,j}(\psi,\{\mathcal{R}_n\})$
between the $i^{\mathrm{th}}$ user and the $j^{\mathrm{th}}$ user as
\begin{equation}  \label{eq:SER} 
e_{i,j}(\psi,\{\mathcal{R}_n\})
=
\lim_{n\to\infty}
\frac{1}{n}
\left( \log d_{B^{(n)}_{i,j}}-\log d_{C^{(n)}_{i,j}}
\right),
\end{equation}
where $i\neq j$, and logarithms are taken to base two throughout this paper.
For convenience, we define $e_{i,i}(\psi,\{\mathcal{R}_n\})$ as zero for each $i$.
Note that $e_{j,i}(\psi,\{\mathcal{R}_n\})=e_{i,j}(\psi,\{\mathcal{R}_n\})$
holds for each $i, j$.
If the segment entanglement rate $e_{i,j}(\psi,\{\mathcal{R}_n\})$
converges for each $i\neq j$,
then we can define the \emph{total entanglement rate} $e_{\mathrm{tot}}(\psi,\{\mathcal{R}_n\})$
as
\begin{equation} \label{eq:tot}
e_{\mathrm{tot}}(\psi,\{\mathcal{R}_n\})
=\sum_{\substack{ i,j\in[M] \\ i< j }} e_{i,j}(\psi,\{\mathcal{R}_n\}).
\end{equation}
A real number $r$ is said to be an (asymptotically) \emph{achievable} total entanglement rate,
if there is a sequence $\{\mathcal{R}_n\}_{n\in\mathbb{N}}$
of QSR protocols $\mathcal{R}_n$ of $\psi^{\otimes n}$
with error $\varepsilon_n$ such that
(i) for any $i,j$, $e_{i,j}(\psi,\{\mathcal{R}_n\})$ converges;
(ii) $e_{\mathrm{tot}}(\psi,\{\mathcal{R}_n\})=r$;
(iii) $\lim_{n\to\infty} \varepsilon_n = 0$.
The \emph{optimal entanglement cost} (OEC) $e_{\mathrm{opt}}(\psi)$ 
of the QSR task for $\psi$ is defined as
the infimum of the achievable total entanglement rates.

\begin{figure}
\includegraphics[clip,width=\columnwidth]{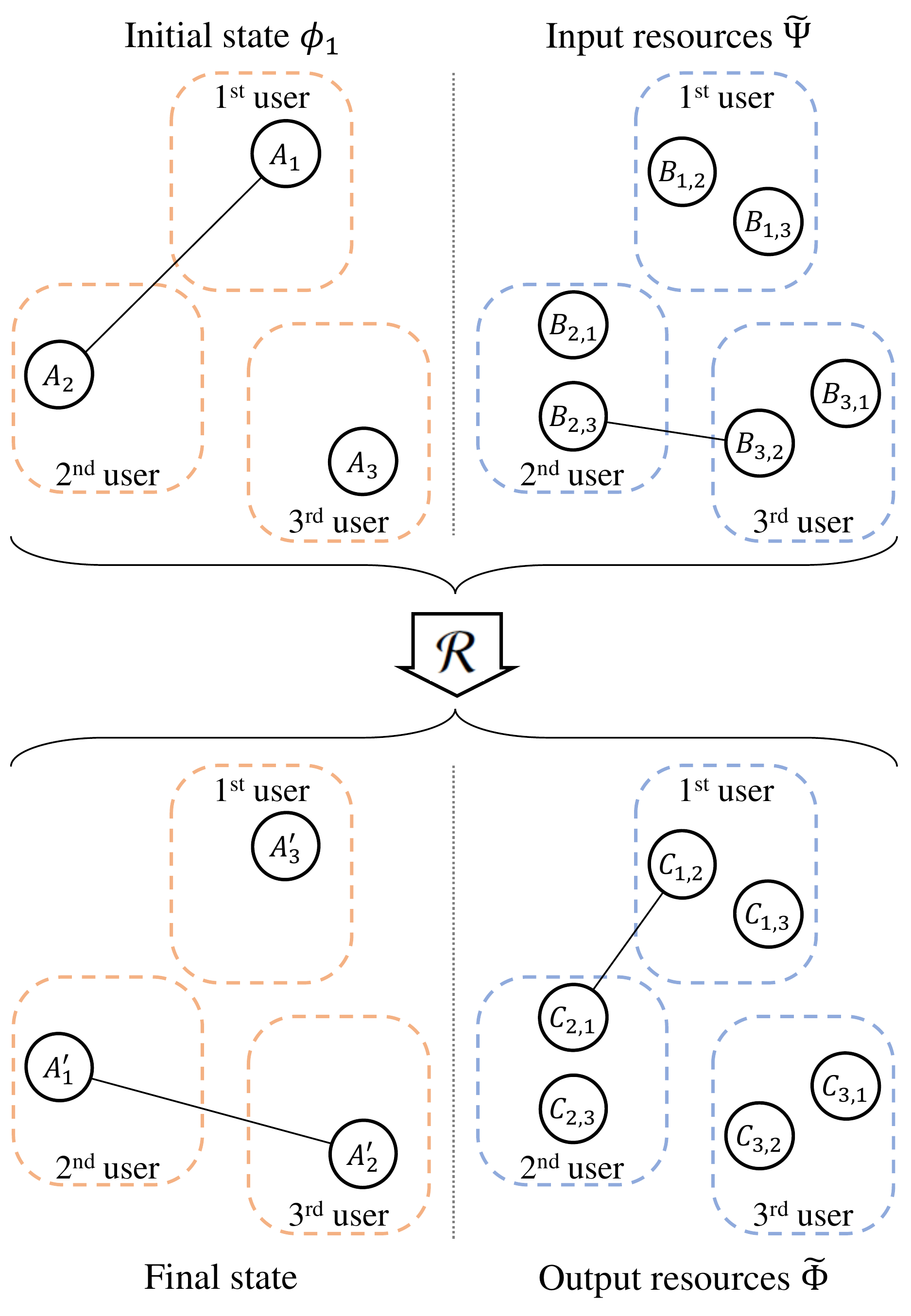}
\caption{
Initial state $\phi_1$ in Eq.~(\ref{eq:whytotal}), its final state, and input/output entanglement resources:
Circles indicate quantum systems, and correlations among the quantum systems are represented 
by lines connecting them.
In order to rotate $\phi_1$,
the $2^{\mathrm{nd}}$ user locally prepares the two-qubit entangled state $\varphi_1$,
and transfers one qubit of $\varphi_1$ to the $3^{\mathrm{rd}}$ user.
So they can share $\varphi_1$ on quantum systems $A'_1A'_2$.
For this, the amount of entanglement consumed by them is $H(A_1)_{\varphi_1}$.
Then the $1^{\mathrm{st}}$ user and the $2^{\mathrm{nd}}$ user asymptotically generate
the same amount of entanglement 
by applying entanglement distillation~\cite{BBPS96,BDSW96,BBPSSW97}
to $\varphi_1$ on quantum systems $A_1A_2$.
Finally,
the $1^{\mathrm{st}}$ user locally prepares the pure quantum state $\varphi_2$ on the system $A'_3$.
In an illustration for input (output) entanglement resources, 
two circles connected by a line indicate consumed (gained) entanglement whose amount is $H(A_1)_{\varphi_1}$.
}
\label{fig:Easy_exam}
\end{figure}

\begin{Rem} \label{rem:whytotal}
In the QSR task,
the users obtain the final state $\psi_{\mathrm{f}}$ and the output entanglement resources $\tilde{\Phi}$
by applying the QSR protocol to the initial state $\psi$ and the input entanglement resources $\tilde{\Psi}$.
Note that the QSR protocol is LOCC,
and the initial state $\psi$ and the final state $\psi_{\mathrm{f}}$ have the same amount of entanglement,
since $\psi_{\mathrm{f}}$ is defined by using $\psi$ and identity maps.
So, it is obvious that the total amount of entanglement among the users does not increase on average
via the QSR protocol.
As a measure that fulfills this condition, we use the total entanglement rate $e_{\mathrm{tot}}$ in this work.
The total entanglement rate $e_{\mathrm{tot}}$ measures the total amount of entanglement between pairs of the users.
In addition,
we will see the non-negativity of the total entanglement rate in Remark~\ref{rem:NNOEC}.
On this account, the total entanglement rate is a valid measure for analyzing
the total amount of entanglement required for the QSR task.

For example, let us consider a simple initial state
\begin{equation} \label{eq:whytotal}
\ket{\phi_1}_{A} = \ket{\varphi_1}_{A_1A_2}\otimes \ket{\varphi_2}_{A_3},
\end{equation}
where $E$ is regarded as a one-dimensional system,
$\ket{\varphi_1}$ is any two-qubit entangled state,
and $\ket{\varphi_2}$ is any quantum state.
We provide illustrations of the initial state $\phi_1$ and its final state in Fig.~\ref{fig:Easy_exam}.
For the initial state $\phi_1$,
we can calculate the segment entanglement rates $e_{i,j}$ and the total entanglement rate $e_{\mathrm{tot}}$ through the following strategy.
(i)
In order to rotate $\phi_1$,
the $2^{\mathrm{nd}}$ user locally prepares a two-qubit state $\ket{\varphi_1}$,
which is not the original state $\ket{\varphi_1}$ on the quantum systems $A_1A_2$,
and asymptotically transfers one qubit of the new state to the $3^{\mathrm{rd}}$ user
by using Schumacher compression~\cite{S95,W13} and the quantum teleportation~\cite{BBCJPW93}.
From this,
the $2^{\mathrm{nd}}$ user and the $3^{\mathrm{rd}}$ user can share the state $\ket{\varphi_1}$ on the systems $A'_1A'_2$,
and the amount of entanglement consumed by them is $H(A_1)_{\varphi_1}$.
(ii) Since the quantum state $\ket{\varphi_1}$ is already distributed to the $2^{\mathrm{nd}}$ user and the $3^{\mathrm{rd}}$ user,
 the original state $\ket{\varphi_1}$ of the $1^{\mathrm{st}}$ user and the $2^{\mathrm{nd}}$ user is now superfluous, but it can be transformed into an output entanglement resource of the QSR task.
In other words, the $1^{\mathrm{st}}$ user and the $2^{\mathrm{nd}}$ user asymptotically generate
$H(A_1)_{\varphi_1}$ amount of entanglement 
by applying entanglement distillation~\cite{BBPS96,BDSW96,BBPSSW97}
to their state $\ket{\varphi_1}_{A_1A_2}$.
(iii)
Finally,
the $1^{\mathrm{st}}$ user locally prepares the pure quantum state $\ket{\varphi_2}$.
This preparation neither requires nor generates any entanglement resources.
To be specific, this strategy can be represented as a sequence $\{\mathcal{R}_n\}$ of QSR
protocols of $\phi_1$
whose segment entanglement rates are
\begin{eqnarray}
e_{1,2}(\phi_1,\{\mathcal{R}_n\})&=&-H(A_1)_{\varphi_1}, \label{eq:ExamNSER} \\
e_{2,3}(\phi_1,\{\mathcal{R}_n\})&=&H(A_1)_{\varphi_1}, \label{eq:ExamPSER} \\
e_{3,1}(\phi_1,\{\mathcal{R}_n\})&=&0,
\end{eqnarray} 
and so the total entanglement rate is zero,
i.e.,
$e_{\mathrm{tot}}(\phi_1,\{\mathcal{R}_n\})=0$.
The positive (negative) segment entanglement rate is described in Fig.~\ref{fig:Easy_exam}.
\end{Rem}

\section{Lower bound} \label{sec:LB}

In this section,
we present a lower bound on the OEC of the QSR task.

For a non-empty proper subset $P$ of the set $[M]$ and the initial state $\ket{\psi}_{AE}$
with $A = A_1A_2\cdots A_M$,
we consider a quantity $l_P(\psi)$ defined as
\begin{equation} \label{eq:lP}
l_P(\psi)
=
\max_{U}
\left\{
 H\left(\sbigotimes_{i\in P} A_{i-1}V\right)_{U\ket{\psi}}
-H\left(\sbigotimes_{i\in P} A_{i }V\right)_{U\ket{\psi}}
\right\},
\end{equation}
where the maximum is taken over all isometries $U$ from $E$ to $V\otimes W$~\cite{W13}, $V$ and $W$ are any quantum systems,
and $U\ket{\psi}$ is an abbreviation for $\mathds{1}_{A}\otimes U\ket{\psi}$.
Note that, for any partition $\{P,P^{\mathsf{c}}\}$ of the set $[M]$, $l_P(\psi)=l_{P^{\mathsf{c}}}(\psi)$ holds.
The quantity $l_P(\psi)$ is a lower bound on the sum of the segment entanglement rate as follows:  
\begin{Lem} \label{lem:LBOBER}
For the initial state $\ket{\psi}_{AE}$ and the partition $\{P,P^{\mathsf{c}}\}$ of the set $[M]$,
the following inequality holds:
\begin{equation}
\sum_{ i\in P } \sum_{ j\in P^{\mathsf{c}} } e_{i,j}(\psi,\{\mathcal{R}_n\})
\ge
l_P(\psi),
\end{equation}
where the segment entanglement rate $e_{i,j}$ is defined in Eq.~(\ref{eq:SER}),
and $\{\mathcal{R}_n\}_{n\in\mathbb{N}}$ is a sequence 
of QSR protocols $\mathcal{R}_n$ of $\psi^{\otimes n}$
with error $\varepsilon_n$ whose total entanglement rate is achievable.
\end{Lem}

A detailed description of the quantity $l_{P}(\psi)$ and the proof of Lemma~\ref{lem:LBOBER} are presented in Appendix~\ref{app:lem:lb}.
By using Lemma~\ref{lem:LBOBER}, we obtain the following theorem providing a lower bound on any achievable total entanglement rate of the QSR task.

\begin{Thm} \label{thm:lb}
Let $\ket{\psi}_{AE}$ be the initial state of the QSR task with $A = A_1A_2\cdots A_M$.
Any achievable total entanglement rate $r$ of the QSR task
is lower bounded by
\begin{equation} \label{eq:defLK}
r
\ge
l_k(\psi)
\coloneqq
\frac{1}{2{M-2 \choose k-1}}
\sum_{P\in S_k}
l_{P}(\psi),
\end{equation}
where $1 \le k < M$,
$S_k$ is the set of subsets $P$ of $[M]$ whose sizes are $k$,
i.e., $|P|=k$,
and $l_{P}(\psi)$ is given in Eq.~(\ref{eq:lP}).
\end{Thm}

We refer the reader to Appendix~\ref{app:thm:lb} for the proof of Theorem~\ref{thm:lb}.
Theorem~\ref{thm:lb} implies that the OEC $e_{\mathrm{opt}}(\psi)$
of the initial state $\psi$
is lower bounded by
\begin{equation} \label{eq:LB}
e_{\mathrm{opt}}(\psi)
\ge
l(\psi)
\coloneqq
\max_{1\le k \le \lfloor \frac{M}{2} \rfloor}
l_k(\psi),
\end{equation}
where $\lfloor x \rfloor$ denotes the floor function
defined as $\max\{m\in\mathbb{Z}:m\le x\}$.

\begin{Rem} \label{rem:NNOEC}
The lower bounds $l_{i}(\psi)$ 
are non-negative,
since for each $i$,
\begin{eqnarray}
\sum_{P\in S_i}
l_{P}(\psi)
&\ge&
\sum_{P\in S_i}
\left[
 H\left(\sbigotimes_{j\in P} A_{j-1}\right)_{\psi}
-H\left(\sbigotimes_{j\in P} A_{j }\right)_{\psi}
\right] \\
&=&
\sum_{P\in S_i}
 H\left(\sbigotimes_{j\in P} A_{j}\right)_{\psi}
 -\sum_{P\in S_i}
 H\left(\sbigotimes_{j\in P} A_{j}\right)_{\psi}=0.
\end{eqnarray}
Thus, the OEC cannot be negative,
i.e.,
$e_{\mathrm{opt}}(\psi)\ge0$.
This means that
the total amount of entanglement gained from the QSR task
cannot exceed that of entanglement resources consumed in the task.

In this work,
while we analyze the OEC
as a figure of merit,
the case of zero OECs,
$e_{\mathrm{opt}}=0$,
does not necessarily mean that the related segment entanglement rates are zero,
i.e,
$e_{i,j}=0$ for each $i\neq j$, as shown in Eqs.~(\ref{eq:ExamNSER}) and~(\ref{eq:ExamPSER}). 
In general,
entanglement resources for the QSR task
may be consumed by some pair of users while distilled by another pair, as in the example of Remark~\ref{rem:whytotal}.
\end{Rem}

\begin{Rem} \label{rem:SElb}
One of our contributions is to generalize
results of the QSE task~\cite{OW08}
to the general cases including more than two users.
To be specific,
Remark~\ref{rem:NNOEC} implies
the non-negativity of the OEC
for the QSE task.
For $M=2$,
the lower bound $l(\psi)$ in Eq.~(\ref{eq:LB}) becomes
\begin{eqnarray}
l(\psi)
&=&\max_{U}
\left\{
 H\left(A_{1}V\right)_{U\ket{\psi}}
-H\left(A_{2}V\right)_{U\ket{\psi}}
\right\} \\
&=& 
\max_{\mathcal{N}} 
\left\{ H(A_{1 }V)_{\mathcal{N}(\psi)}
-H(A_{2}V)_{\mathcal{N}(\psi)} \right\},
\end{eqnarray}
where $V$ is any quantum system,
$\mathcal{N}$ is any quantum channel from $\mathcal{L}(E)$ to $\mathcal{L}(V)$,
and $\mathcal{N}(\psi)$ is an abbreviation for $(\mathds{1}_{\mathcal{L}(A)}\otimes\mathcal{N})(\psi)$.
The first equality comes from $l_{\{1\}}(\psi)=l_{\{2\}}(\psi)$.
The second equality holds, since there is a one-to-one correspondence between isometries $U$ and quantum channels $\mathcal{N}$.
That is, any isometry $U\colon E\to V\otimes W$ combined with the partial trace over the quantum system $W$ becomes a quantum channel $\mathcal{N}\colon E\to V$, and for any quantum channel $\mathcal{N}$, we can find its isometric extension $U$~\cite{W13}.
The above quantity is the lower bound on the OEC for the QSE task presented in Ref.~\cite{OW08}.
\end{Rem}

\section{Achievable upper bound} \label{sec:AUB}

In this section,
we present an achievable upper bound on the OEC of the QSR task by considering a specific strategy.

The QSR task can be carried out 
by using an \emph{$M$-partite merge-and-send} strategy.
We can obtain this strategy by generalizing the merge-and-send strategy presented in Ref.~\cite{OW08}.
The idea of the $M$-partite merge-and-send strategy is as follows:
(i)
The $1^{\mathrm{st}}$ user and the $2^{\mathrm{nd}}$ user of the QSR task
merge the part $A_1$ to the $2^{\mathrm{nd}}$ user
by using quantum state merging~\cite{HOW05,HOW06}.
In this case,
the part $A_2$ of the $2^{\mathrm{nd}}$ user acts as the quantum side information.
After finishing merging $A_1$,
the $2^{\mathrm{nd}}$ user considers his/her part $A_1$ as a part of the environment system.
Then the $2^{\mathrm{nd}}$ user and the $3^{\mathrm{rd}}$ user can make use of the quantum state merging protocol again,
in order to merge $A_2$.
In this way,
the part $A_i$ is sequentially merged from the $i^{\mathrm{th}}$ user to the $(i+1)^{\mathrm{th}}$ user 
except for the last part $A_M$.
(ii) Finally,
the last user and the $1^{\mathrm{st}}$ user perform
Schumacher compression~\cite{S95,W13}
together with quantum teleportation~\cite{BBCJPW93}
in order to transfer the part $A_M$ to the $1^{\mathrm{st}}$ user.
Through this strategy,
the $M$ users can rotate any initial state of the QSR task.
Note that instead of using quantum state merging~\cite{HOW05,HOW06},
the $M$ users can apply
quantum state redistribution~\cite{DY08,YD09}
with quantum teleportation~\cite{BBCJPW93}
in order to perform the QSR task.
In this case,
the total amount of entanglement is identical to 
that of the $M$-partite merge-and-send strategy,
while the amounts of classical communication can be different.

When the users adopt the above $M$-partite merge-and-send strategy,
for each $i\in[M-1]$,
the entanglement cost of merging $A_i$ 
is represented as $H(A_i|A_{i+1})_\psi$,
and the entanglement cost for transferring $A_M$ is $H(A_M)_\psi$.
In other words,
these entanglement costs can be represented
in terms of the segment entanglement rates as follows.

\begin{Lem} \label{lem:MPMM}
For any initial state $\ket{\psi}_{AE}$ of the QSR task with $A = A_1A_2\cdots A_M$,
there is a sequence $\{\mathcal{R}_n\}_{n\in\mathbb{N}}$
of QSR protocols $\mathcal{R}_n$ of $\psi^{\otimes n}$
with error $\varepsilon_n$ such that $\lim_{n\to\infty} \varepsilon_n = 0$,
\begin{eqnarray}
e_{i,j}(\psi,\{\mathcal{R}_n\})
&=&
\begin{cases}
H(A_i|A_{i+1})_\psi &\text{if $i\in[M-1]$ and $j=i+1$} \\
H(A_M)_\psi &\text{if $i=M$ and $j=1$} \\
0 &\text{otherwise},
\end{cases} \\
e_{\mathrm{tot}}(\psi,\{\mathcal{R}_n\})
&=&
H(A_M)_\psi+\sum_{i=1}^{M-1}H(A_i|A_{i+1})_\psi.
\end{eqnarray}
\end{Lem}

In order to prove Lemma~\ref{lem:MPMM},
we apply a technique presented in Ref.~\cite{LTYAL19},
which is used to show the existence of the merge-and-merge protocol therein,
and we refer the reader
to Appendix~\ref{app:PoLMPMM} for the proof of Lemma~\ref{lem:MPMM}.

By using Lemma~\ref{lem:MPMM},
we obtain the following theorem, which provides an achievable upper bound on the OEC
of the QSR task.

\begin{Thm} \label{thm:UB}
Let $\ket{\psi}_{AE}$ be the initial state for the QSR task
with $A = A_1A_2\cdots A_M$.
The OEC $e_{\mathrm{opt}}(\psi)$ is upper bounded by
\begin{equation}
u(\psi)
\coloneqq
\sum_{i=1}^M H(A_i|A_{i+1})_\psi
+\min_{1\le i \le M} I(A_i;A_{i+1})_\psi.
\end{equation}
\end{Thm}

\begin{proof}
For each $i\in[M]$,
we consider an $M$-partite merge-and-send strategy
in which the part $A_i$ is firstly merged from the $i^{\mathrm{th}}$ user to the $(i+1)^{\mathrm{th}}$ user
and the part $A_{i-1}$ is lastly sent from the $(i-1)^{\mathrm{th}}$ user to the $i^{\mathrm{th}}$ user.
From Lemma~\ref{lem:MPMM},
the achievable total entanglement rate $u_i(\psi)$ for this strategy
is given by
\begin{eqnarray}
u_i(\psi)
&=&
H(A_{i-1})_\psi+\sum_{j\in[M]\setminus\{i-1\}}H(A_j|A_{j+1})_\psi \\
&=&
I(A_{i};A_{i-1})_\psi+\sum_{j\in[M]}H(A_j|A_{j+1})_\psi.
\end{eqnarray}
It follows that $e_{\mathrm{opt}}(\psi)
\le \min_{1\le i\le M}u_i(\psi)$,
from optimizing the choice of the $1^{\mathrm{st}}$ user starting the merge-and-send strategy.
\end{proof}

\begin{Rem} \label{rem:exam1}
By using the lower bound in Eq.~(\ref{eq:LB}) and Theorem~\ref{thm:UB},
we can exactly evaluate the OECs for some initial states.
For example,
let us consider an initial state
\begin{equation}
\ket{\phi_2}_{AE} = \sbigotimes_{i=1}^M \ket{\varphi_i}_{A_iE_i},
\end{equation}
where $E = E_1E_2\cdots E_M$,
and $\ket{\varphi_i}$ is any pure bipartite entangled state on the quantum systems $A_iE_i$.
Then, from the lower bound in Eq.~(\ref{eq:LB}), the OEC $e_{\mathrm{opt}}(\phi_2)$ is lower bounded by
\begin{equation} \label{eq:l1}
l_1(\phi_2)=
\frac{1}{2}
\sum_{i=1}^{M}\max_{U} 
\left\{ H(A_{i-1}V)_{U\ket{\phi_2}}-H(A_{i }V)_{U\ket{\phi_2}} \right\},
\end{equation}
where $l_1$ is defined in Theorem~\ref{thm:lb},
and $U\ket{\phi_2}$ is an abbreviation for $\mathds{1}_{A}\otimes U\ket{\phi_2}$.
So, $e_{\mathrm{opt}}(\phi_2)$ is lower bounded by
\begin{equation}
e_{\mathrm{opt}}(\phi_2)
\ge
\frac{1}{2} \sum_{i=1}^{M} \left[ H(A_{i-1}E_i)_{\phi_2}-H(A_{i}E_i)_{\phi_2} \right]
=
\sum_{i=1}^{M} H(A_i)_{\varphi_i},
\end{equation}
if we consider isometries $U_i\colon E\to E_i\otimes (E\setminus E_i)$ with $i\in[M]$
such that $\Tr_{E\setminus E_i}U_i\phi_2U_i^\dagger=\Tr_{E\setminus E_i}\phi_2$,
where $E\setminus E_i=E_1\cdots E_{i-1}E_{i+1}\cdots E_M$.
Moreover,
from Theorem~\ref{thm:UB},
we have
\begin{equation}
e_{\mathrm{opt}}(\phi_2)
\le
u(\phi_2)
=
\sum_{i=1}^{M} H(A_i)_{\varphi_i}.
\end{equation}
Hence, $e_{\mathrm{opt}}(\phi_2) = \sum_{i=1}^{M} H(A_i)_{\varphi_i}$.
\end{Rem}

\begin{Rem} \label{rem:Notoptimal}
In general,
the $M$-partite merge-and-send strategy is not necessarily the optimal strategy,
although 
we have used it in order to find the OEC
for the specific initial state in Remark~\ref{rem:exam1}.
As a counterexample of the optimality,
let us consider an initial state
\begin{equation} \label{eq:exmCEA}
\ket{\phi_3}_A
=
\ket{\varphi_1}_{A_1A_3}
\otimes
\ket{\varphi_2}_{A_2}
\otimes
\ket{\varphi_3}_{A_4},
\end{equation}
where $\ket{\varphi_1}$ is any pure two-qubit entangled state,
and $\ket{\varphi_2}$ and $\ket{\varphi_3}$ are any pure quantum states.
Here,
$E$ is regarded as a one-dimensional system.
If we apply the $M$-partite merge-and-send strategy to the initial state $\ket{\phi_3}_A$,
then we obtain $u(\phi_3)=2H(A_1)_{\varphi_1}$ from Theorem~\ref{thm:UB}.

However,
using the strategy presented in Remark~\ref{rem:whytotal},
we obtain an achievable upper bound smaller than $u(\phi_3)$.
To be specific,
the $2^{\mathrm{nd}}$ user locally prepares the two-qubit entangled state $\varphi_1$,
and transfers one qubit of the state to the $4^{\mathrm{th}}$ user
by consuming as much entanglement as $H(A_1)_{\varphi_1}$.
The $1^{\mathrm{st}}$ user and the $3^{\mathrm{rd}}$ user then generate
the same amount of entanglement 
by distilling their state $\varphi_1$ on the quantum systems $A_1A_3$.
Finally,
the $1^{\mathrm{st}}$ user and the $3^{\mathrm{rd}}$ user locally prepare pure states $\varphi_3$ and $\varphi_2$,
respectively,
without consuming and gaining any entanglement resource.
This strategy can be represented as a sequence $\{\mathcal{R}_n\}$ of QSR
protocols of $\phi_3$
whose segment entanglement rates are zero except for 
$e_{1,3}(\phi_3,\{\mathcal{R}_n\})=-H(A_1)_{\varphi_1}$
and $e_{2,4}(\phi_3,\{\mathcal{R}_n\})=H(A_1)_{\varphi_1}$,
and $e_{\mathrm{tot}}(\phi_3,\{\mathcal{R}_n\})=0$.
It follows that $e_{\mathrm{tot}}(\phi_3,\{\mathcal{R}_n\})<u(\phi_3)$,
since $\varphi_1$ is entangled.
This shows that the $M$-partite merge-and-send strategy is not optimal in general.
In addition,
the non-negativity of the OEC implies $e_{\mathrm{opt}}(\phi_3)=0$ in this case.
\end{Rem}

\begin{figure}
\includegraphics[clip,width=\columnwidth]{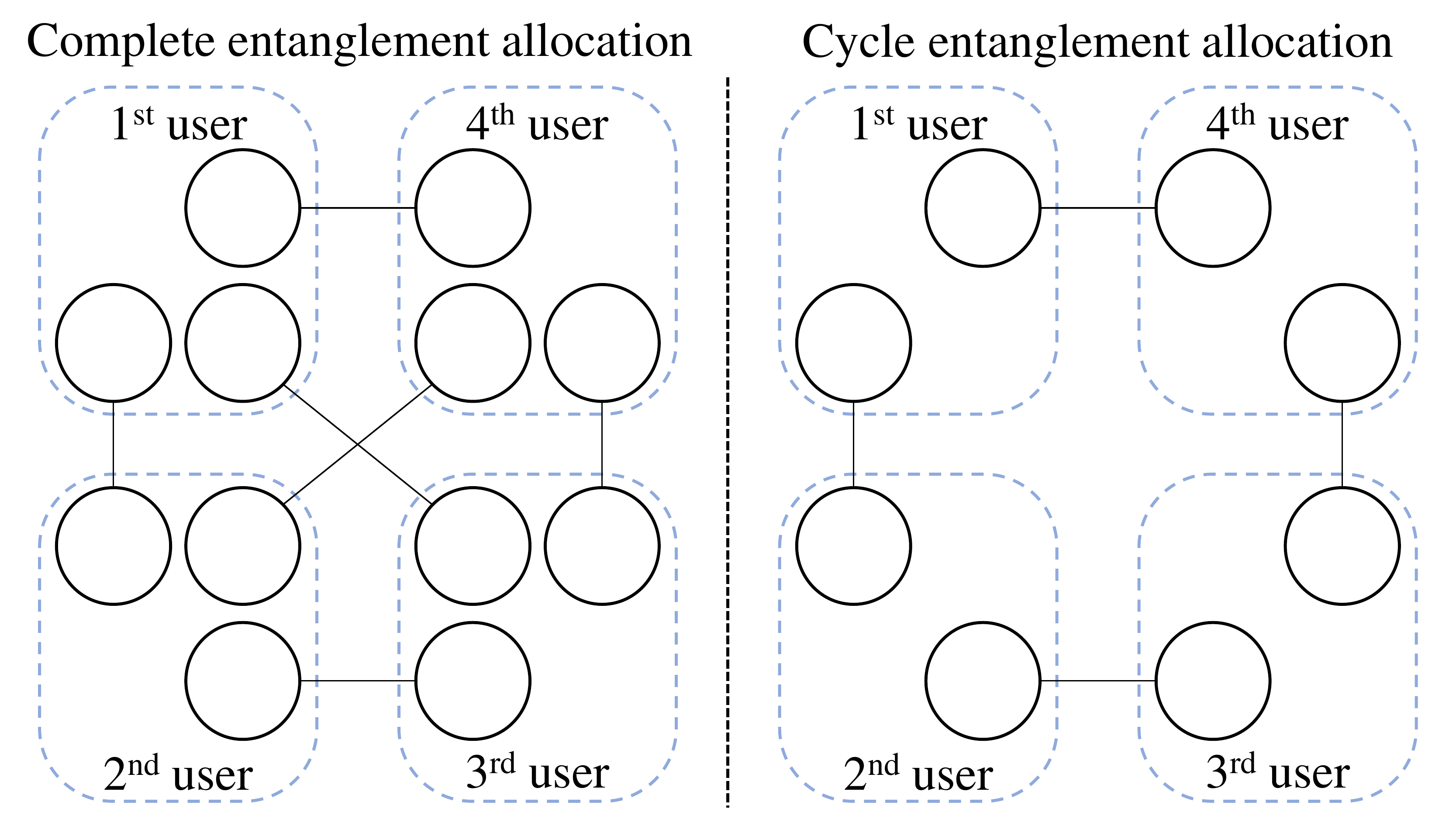}
\caption{
Illustrations of the complete entanglement allocation and the cycle entanglement allocation for four users:
An entanglement resource between two users is represented as two circles connected by a line.
Under the complete entanglement allocation,
every pair of four users can freely consume and generate entanglement resource.
However,
under the cycle entanglement allocation,
only the $i^{\mathrm{th}}$ user and the $(i+1)^{\mathrm{th}}$ user can manipulate entanglement resources,
and the $i^{\mathrm{th}}$ user and the $(i+2)^{\mathrm{th}}$ user are not allowed to deal with any entanglement resources.
}
\label{fig:ComplCircle}
\end{figure}

\begin{Rem} \label{rem:CGFR} 
Throughout this paper,
we have been assuming that the users of the QSR task make use of 
the complete entanglement allocation.
However,
one may think that it suffices to consider bipartite entanglement resources between the $i^{\mathrm{th}}$ user and the $(i+1)^{\mathrm{th}}$ user for each $i$,
since the $i^{\mathrm{th}}$ user transfers his/her quantum state to the $(i+1)^{\mathrm{th}}$ user
in the $M$-partite merge-and-send strategy.
Here,
we call such an allocation of entanglement resources the \emph{cycle entanglement allocation},
and we provide illustrations  explaining how four users share entanglement resources
according to the complete entanglement allocation and the cycle entanglement allocation in Fig.~\ref{fig:ComplCircle}.

The initial state $\phi_3$ in Eq.~(\ref{eq:exmCEA})
shows that, under the complete entanglement allocation setting,
the users can reduce the total amount of entanglement for the QSR task
compared to the case that the users use the cycle entanglement allocation for
rotating the same initial state.

To see this reduction,
we evaluate a lower bound on the OEC
for rotating $\phi_3$,
when the users use the cycle entanglement allocation.
This means
that the $1^{\mathrm{st}}$ ($2^{\mathrm{nd}}$) user and the $3^{\mathrm{rd}}$ ($4^{\mathrm{th}}$) user
cannot employ any entanglement resource between them, as depicted in Fig.~\ref{fig:ComplCircle}.
Let $\{\mathcal{C}_n\}_{n\in\mathbb{N}}$ be a sequence 
of such protocols $\mathcal{C}_n$ rotating $\phi_3^{\otimes n}$
with error $\varepsilon_n$,
where the users of each protocol use the cycle entanglement allocation.
While there is no need to consider
the segment entanglement rates $e_{1,3}(\phi_3,\{\mathcal{C}_n\})$
and $e_{2,4}(\phi_3,\{\mathcal{C}_n\})$ in this case,
we assume that $e_{1,3}(\phi_3,\{\mathcal{C}_n\})=e_{2,4}(\phi_3,\{\mathcal{C}_n\})=0$,
in order to regard the protocol $\mathcal{C}_n$ as the special case of the QSR protocol.
Recall that the state $\phi_3$ has no further environment system $E$, as shown in Eq.~(\ref{eq:exmCEA}).
Then, from Lemma~\ref{lem:LBOBER},
we obtain that
the inequality
\begin{equation}
e_{i,i-1}(\phi_3,\{\mathcal{C}_n\})+e_{i,i+1}(\phi_3,\{\mathcal{C}_n\})
\ge l_{\{i\}}(\phi_3)
\end{equation}
holds for each $1\le i\le 4$.
By using this inequality and the definition of $e_{\mathrm{tot}}$ in Eq.~(\ref{eq:tot}),
we obtain
\begin{eqnarray}
e_{\mathrm{tot}}(\phi_3,\{\mathcal{C}_n\})
&\ge& l_{\{1\}}(\phi_3) +l_{\{3\}}(\phi_3), \\
e_{\mathrm{tot}}(\phi_3,\{\mathcal{C}_n\})
&\ge& l_{\{2\}}(\phi_3) +l_{\{4\}}(\phi_3)= -(l_{\{1\}}(\phi_3) +l_{\{3\}}(\phi_3)). \nonumber \\
\end{eqnarray}
It follows that
\begin{eqnarray}
&&|l_{\{1\}}(\phi_3) +l_{\{3\}}(\phi_3)| \\
&&=|H(A_1)_{\phi_3}-H(A_2)_{\phi_3}+H(A_3)_{\phi_3}-H(A_4)_{\phi_3}| \\
&&=2H(A_1)_{\varphi_1}
\end{eqnarray}
is a non-zero lower bound on the OEC of rotating $\phi_3$
under the cycle entanglement allocation.
However, in the case of the complete entanglement allocation,
we obtain $e_{\mathrm{opt}}(\phi_3)=0$,
as explained in Remark~\ref{rem:Notoptimal}.

On this account,
the case of the initial state $\phi_3$
tells us that the use of the complete entanglement allocation can give
a smaller total entanglement rate than that of the cycle entanglement allocation.
This justifies that we consider the complete entanglement allocation
rather than cyclic entanglement allocation
in the definition of the QSR task.
\end{Rem}

\section{Conditions} \label{sec:CZATER}

In this section,
we present a sufficient condition on positive OECs
and a necessary condition on zero achievable total entanglement rates.

\subsection{Condition on positive optimal entanglement cost}

We provide a sufficient condition on positive OECs
of the QSR task.

When $M=2$, the QSR task is nothing but the QSE task,
and we can find out a condition
by using results on the QSE task
presented in Refs.~\cite{OW08,LYAL19}.
That is,
if $H(A_1)_\psi\neq H(A_2)_\psi$ for the initial state $\ket{\psi}_{A_1A_2E}$,
then the OEC of
the QSE task is positive,
i.e.,
$e_{\mathrm{opt}}(\psi)>0$.
So, one may naturally guess a generalized sufficient condition
with respect to
the initial state $\ket{\psi}_{AE}$ on $A = A_1A_2\cdots A_M$ as follows:
If there exist some $i,j\in[M]$ such that
\begin{equation}
H(A_i)_\psi\neq H(A_j)_\psi,
\end{equation}
then $e_{\mathrm{opt}}(\psi)>0$.

However,
this guess is not the case.
Let us consider the initial state $\phi_3$ in Eq.~(\ref{eq:exmCEA}).
Then, it is satisfied that $H(A_1)_{\phi_3}> 0=H(A_2)_{\phi_3}$,
but $e_{\mathrm{opt}}(\phi_3)=0$ as explained in Remark~\ref{rem:Notoptimal}.
Interestingly,
the above  condition can be corrected
in terms of quantum conditional entropies. 

\begin{Thm} \label{thm:SCOFC}
Let $\ket{\psi}_{AE}$ be the initial state for the QSR task with $A = A_1A_2\cdots A_M$.
If there exist some $i,j\in[M]$ such that
\begin{equation}
H(E|A_{i})_{\psi}\neq H(E|A_{j})_{\psi},
\end{equation}
then $e_{\mathrm{opt}}(\psi)>0$.
\end{Thm}

The proof of Theorem~\ref{thm:SCOFC} can be found in Appendix~\ref{app:PoTsc}.

\begin{Rem}
The converse of Theorem~\ref{thm:SCOFC} does not necessarily hold.
Consider the initial state
\begin{equation}
\ket{\phi_4}_{AE} = \sbigotimes_{i=1}^M \ket{\varphi}_{A_iE_i},
\end{equation}
where $E = E_1E_2\cdots E_M$ and $\ket{\varphi}$ is any pure bipartite entangled state.
Then we know that the OEC for rotating $\phi_4$ is positive
from the lower bound in Eq.~(\ref{eq:LB}),
but the condition in Theorem~\ref{thm:SCOFC} does not hold.
\end{Rem}

\subsection{Condition on zero achievable total entanglement rate}

We now present the following theorem providing a necessary condition on zero achievable total entanglement rates
for the QSR task.

\begin{Thm} \label{thm:necessary}
Let $\ket{\psi}_{AE}$ be the initial state of the QSR task with $A = A_1A_2\cdots A_M$,
and let $\{\mathcal{R}_n\}_{n\in\mathbb{N}}$ be a sequence of
QSR protocols $\mathcal{R}_n$ of $\psi^{\otimes n}$
with error $\varepsilon_n$ whose total entanglement rate $r$ is achievable.
If $r=0$,
then the segment entanglement rates $e_{i,j}(\psi,\{\mathcal{R}_n\})$ for $i\neq j$
are determined as

(i) If $M=3$,
$e_{i,j}(\psi,\{\mathcal{R}_n\})
=-l_{\{i,j\}}(\psi)$.

(ii) If $M=4$,
$e_{i,j}(\psi,\{\mathcal{R}_n\})=\frac{1}{2}
\left( l_{\{i\}}(\psi) + l_{\{j\}}(\psi) - l_{\{i,j\}}(\psi) \right)$.

(iii) If $M>4$, $e_{i,j}(\psi,\{\mathcal{R}_n\})$ is represented as
\begin{equation}
\frac{1}{\alpha_M}
\left(
\beta_M l_{\{i,j\}}(\psi)
+\gamma_M \sum_{\substack{s\in\{i,j\} \\ t\in[M]\setminus\{i,j\}}} l_{\{s,t\}}(\psi)
-2 \sum_{\substack{s,t\in[M]\setminus\{i,j\} \\ s<t}} l_{\{s,t\}}(\psi)
\right),
\end{equation}
where $l_{\{i\}}$ and $l_{\{i,j\}}$ are defined in Eq.~(\ref{eq:lP}),
$\alpha_M=2(M-2)(M-4)$, $\beta_M=2-(M-4)^2$, and $\gamma_M=M-4$.
\end{Thm}

The main idea of the proof for Theorem~\ref{thm:necessary}
is to construct a system of linear equations obtained by
regarding segment entanglement rates as its unknowns
and to solve it.
We refer the reader to Appendix~\ref{app:PoTnecessary} for the proof of Theorem~\ref{thm:necessary}.

\begin{Rem}
The meaning of Theorem~\ref{thm:necessary} is that 
if there exist two sequences $\{\mathcal{R}_n\}_{n\in\mathbb{N}}$
and $\{\mathcal{R}'_n\}_{n\in\mathbb{N}}$ of QSR protocols
for the same initial state $\ket{\psi}_{AE}$
whose achievable total entanglement rates are zero,
then their segment entanglement rates are the same,
i.e.,
\begin{equation}
e_{i,j}(\psi,\{\mathcal{R}_n\})
=e_{i,j}(\psi,\{\mathcal{R}'_n\})
\end{equation}
for each $i\neq j$.
This implies that, when an achievable total entanglement rate of the QSR task is zero,
for each $i\neq j$,
all possible segment entanglement rates $e_{i,j}$ between the $i^{\mathrm{th}}$ user and the $j^{\mathrm{th}}$ user are uniquely determined
as the same value, 
even though there may not be a unique optimal strategy for the QSR task.
\end{Rem}

\begin{Rem}
To evaluate the segment entanglement rates 
presented in Theorem~\ref{thm:necessary},
we need to evaluate two quantities $l_{\{i\}}(\psi)$ and $l_{\{i,j\}}(\psi)$.
In general,
it is difficult to calculate these quantities
with respect to the initial state $\ket{\psi}_{AE}$,
since they are optimized over all isometries $U$ from $E$ to $V\otimes W$, where $V$ and $W$ are any quantum systems.
However,
for initial states $\ket{\psi}_{A}$ without the environment system $E$,
it is possible to compute them as follows:
\begin{eqnarray}
l_{\{i\}}(\psi_A)
&=&
 H( A_{i-1})_{\psi_A}
-H( A_{i })_{\psi_A}, \\
l_{\{i,j\}}(\psi_A)
&=&
 H( A_{i-1}A_{j-1})_{\psi_A}
-H( A_{i }A_{j})_{\psi_A}.
\end{eqnarray}
We will see that
these computable quantities play a crucial role
in proving Proposition~\ref{prop:example} in the next section.
\end{Rem}

\section{Difference between quantum state rotation and quantum state exchange} \label{sec:DbRaE}

In this section,
we show that a property of the QSE task~\cite{OW08}
does not holds in the QSR task.
This shows the difference between the QSE task and the QSR task.

In the QSE task, 
the initial state $\ket{\psi}_{A_1A_2}$
without the environment system $E$ can be exactly
exchanged without consuming any entanglement resources
via local unitary operations.
So,
for any initial state $\ket{\psi}_{A_1A_2}$,
there exists a sequence of QSE protocols
for $\ket{\psi}_{A_1A_2}$
whose achievable (total) entanglement rate is zero.
Thus, the OEC for
the QSE task of $\ket{\psi}_{A_1A_2}$ is always zero.

How about the QSR task of $M$ users ($M\ge 3$)?
That is,
for any initial state $\ket{\psi}_A$ with $A = A_1A_2\cdots A_M$,
is there a sequence of QSR protocols
whose achievable total entanglement rate is zero?
In the cases of $M=3,4,5$,
we have not found answers to the above question.
However,
if $M\ge6$,
we can find some initial states that cannot be rotated at zero achievable total entanglement rate.

\begin{Prop} \label{prop:example}
For each $M\ge6$,
there exists an initial state $\ket{\psi}_A$ of the QSR task
whose achievable total entanglement rates $r$ cannot be zero,
i.e., $r>0$, where $A = A_1A_2\cdots A_M$.
\end{Prop}

Before proving Proposition~\ref{prop:example}, let us consider a three-user (TU) task different from the QSR task.
In the TU task, three users, Alice, Bob, and Charlie, share two Greenberger-Horne-Zeilinger (GHZ) states~\cite{GHZ89},
and they transform the GHZ states into three ebits symmetrically shared among the three users,
where the states are defined as
\begin{equation}  \label{eq:GHZebit}
\ket{\mathrm{GHZ}}=\frac{1}{\sqrt{2}}(\ket{000}+\ket{111})
\quad
\mathrm{and}
\quad
\ket{\mathrm{ebit}}=\frac{1}{\sqrt{2}}(\ket{00}+\ket{11}).
\end{equation}
It turns out that, by using LOCC, it is impossible to perform the TU task under the exact and asymptotic scenarios~\cite{BPRST00,LPSW05}.

\begin{figure}
\includegraphics[clip,width=\columnwidth]{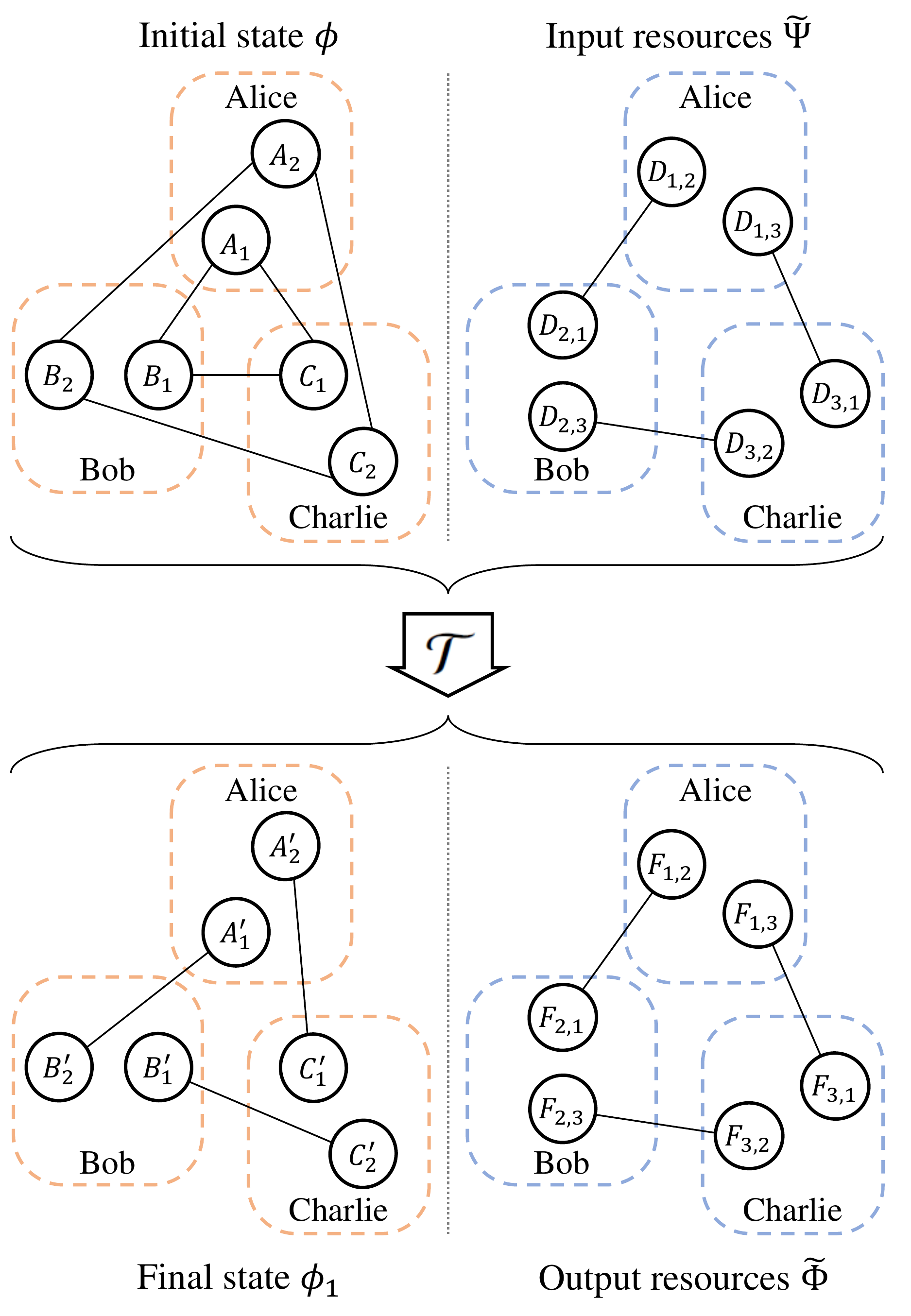}
\caption{
Illustration for the three-user task of Alice, Bob, and Charlie:
Circles indicate quantum systems for the task, and correlations among the quantum systems are represented 
by lines connecting them.
In this task,
the initial state $\phi$ and the final state $\phi_{\mathrm{f}}$ consist of 
two GHZ states and three ebits, respectively.
On the left side of the illustration,
a GHZ state (an ebit) is represented as three (two) circles connected by lines.
The aim of this task is to transform $\phi$ into $\phi_{\mathrm{f}}$.
To perform the task,
they apply LOCC $\mathcal{T}$
to the initial state $\phi$ and input entanglement resources $\tilde{\Psi}$.
After the task,
they can gain output entanglement resources $\tilde{\Phi}$ from the task.
}
\label{fig:TU_task}
\end{figure}

To prove Proposition~\ref{prop:example},
we further show that,
even considering the catalytic use of entanglement resources among them,
it is impossible to carry out the TU task under the asymptotic scenario.
To be specific,
assume that Alice, Bob, and Charlie of the TU task have quantum systems $A_iA'_i$, $B_iB'_i$, and $C_iC'_i$ with $i=1,2$,
respectively.
Let $\ket{\phi}$ and $\ket{\phi_{\mathrm{f}}}$ be the initial and final states of the TU task given by
\begin{eqnarray}
\ket{\phi}
&=&\ket{\mathrm{GHZ}}_{A_{1}B_{1}C_{1}}
\otimes
\ket{\mathrm{GHZ}}_{A_{2}B_{2}C_{2}}, \\
\ket{\phi_{\mathrm{f}}}
&=&\ket{\mathrm{ebit}}_{A'_{1}B'_{2}}
\otimes
\ket{\mathrm{ebit}}_{B'_{1}C'_{2}}
\otimes
\ket{\mathrm{ebit}}_{C'_{1}A'_{2}}.
\end{eqnarray}
Then a quantum channel
\begin{equation}
\mathcal{T}\colon
\mathcal{L}\left( \sbigotimes_{i=1}^2 A_i B_i C_i \otimes D \right)
\longrightarrow
\mathcal{L}\left( \sbigotimes_{i=1}^2 A'_i B'_i C'_i \otimes F \right)
\end{equation}
is called the TU protocol of the initial state $\phi$ with error $\varepsilon$,
if it is performed by LOCC among the three users and satisfies
$\|
\mathcal{T}
( \phi \otimes \tilde{\Psi} )
-
\phi_{\mathrm{f}} \otimes \tilde{\Phi}
\|_{1}
\le
\varepsilon$,
where $D$ and $F$ are multi-partite quantum systems with $D=D_{1,2}D_{1,3}D_{2,1}D_{2,3}D_{3,1}D_{3,2}$
and $F=F_{1,2}F_{1,3}F_{2,1}F_{2,3}F_{3,1}F_{3,2}$,
and $\tilde{\Psi}$ and $\tilde{\Phi}$ are entanglement resources on $D$ and $F$ for the complete entanglement allocation.
We present an illustration for the TU task in Fig.~\ref{fig:TU_task}.

We provide the following lemma whose proof is presented in Appendix~\ref{app:PoL}.

\begin{Lem} \label{lem:GHZtoSinglet}
Let $\ket{\phi}$ be the initial state of the TU task.
Then there is no sequence $\{\mathcal{T}_n\}_{n\in\mathbb{N}}$
of LOCC $\mathcal{T}_n$ of $\phi^{\otimes n}$ with error $\varepsilon_n $ such that
$e_{i,j}(\phi,\{\mathcal{T}_n\})=0$ for each $i\neq j$
and $\lim_{n\to\infty} \varepsilon_n=0$.
\end{Lem}

In Lemma~\ref{lem:GHZtoSinglet},
the catalytic use of entanglement resources is described in the following sense:
While three users are free to consume and gain the entanglement resources
in each protocol $\mathcal{T}_n$,
the amount of entanglement consumed by every pair of the users
is asymptotically equal to that of entanglement gained between them,
i.e.,
the segment entanglement rate $e_{i,j}(\phi,\{\mathcal{T}_n\})$ is zero,
for each $i\neq j$.

\begin{figure}
\subfloat[]{%
  \includegraphics[clip,width=.795\columnwidth]{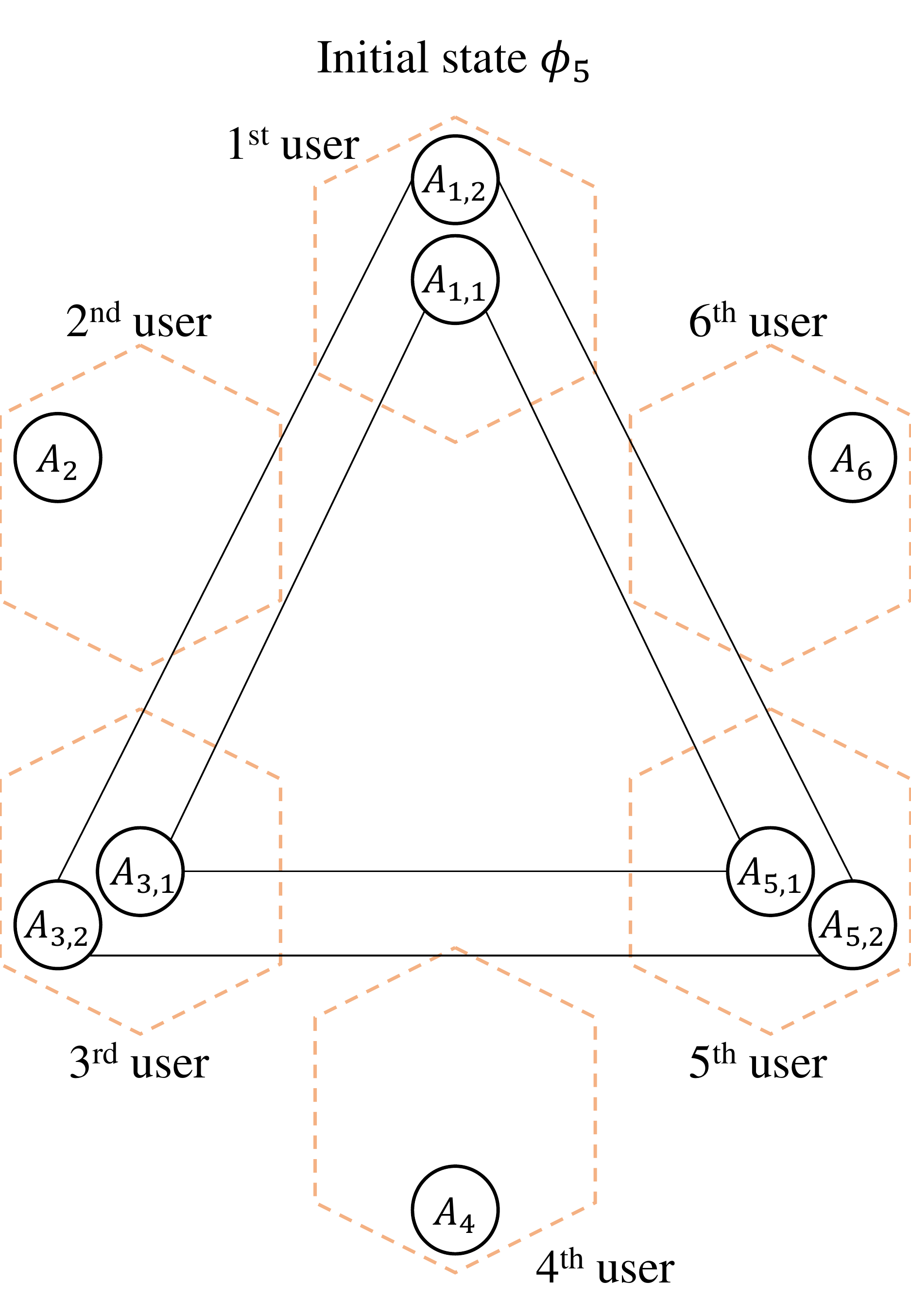}%
}

\subfloat[]{%
  \includegraphics[clip,width=.795\columnwidth]{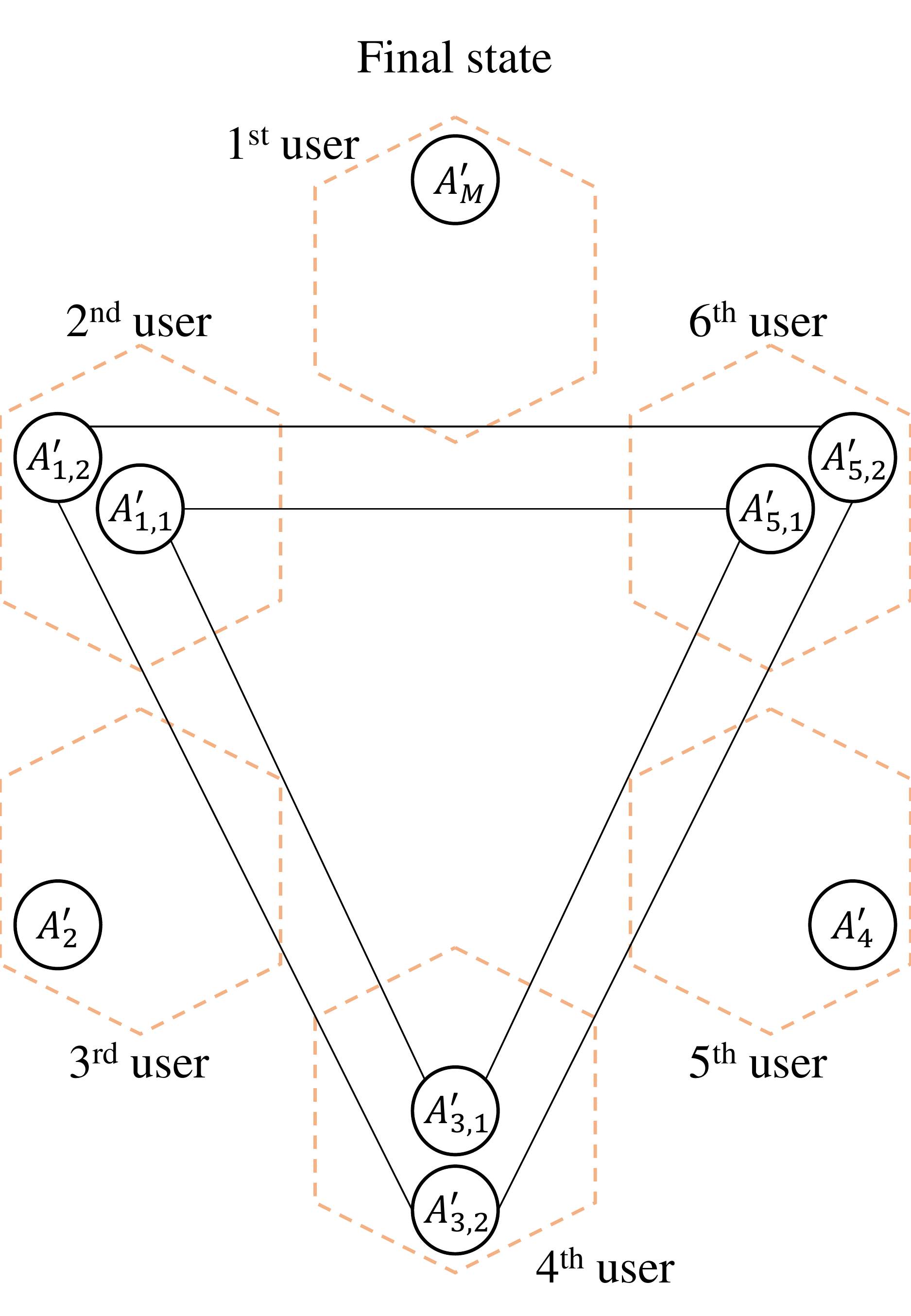}%
}

\caption{
(a) Initial state $\phi_5$ in Eq.~(\ref{eq:example1})
shared over $M$ users for $M\geq 6$;
(b) Final state obtained by rotating the initial state;
In each illustration,
circles indicate quantum systems for the quantum state rotation task,
and, for each $i\ge7$,
the $i^{\mathrm{th}}$ user and his/her quantum systems 
are not explicitly illustrated. 
A GHZ state is represented as three circles connected by lines.
}
\label{fig:twoGHZ}
\end{figure}

\begin{proof}[Proof of Proposition~\ref{prop:example}]
As described in Fig.~\ref{fig:twoGHZ}(a),
we construct an initial state $\ket{\phi_5}_A$ of the QSR task
on the system $A=A_1A_2\cdots A_M$ with $M\ge6$ as follows:
\begin{equation} \label{eq:example1}
\ket{\phi_5}_A
=\ket{\mathrm{GHZ}}_{A_{1,1}A_{3,1}A_{5,1}}
\otimes
\ket{\mathrm{GHZ}}_{A_{1,2}A_{3,2}A_{5,2}}
\otimes
\sbigotimes_{i\in[M]\setminus\{1,3,5\}} \ket{\varphi_i}_{A_i},
\end{equation}
where $A_i=A_{i,1}A_{i,2}$ for $i=1,3,5$,
and $\ket{\varphi_i}$ is any pure quantum state.
The final state corresponding to $\phi_5$ is also presented in Fig.~\ref{fig:twoGHZ}(b).

Suppose that there is a sequence $\{\mathcal{R}_n\}_{n\in\mathbb{N}}$
of QSR protocols $\mathcal{R}_n$ of $\phi_5^{\otimes n}$
whose achievable total entanglement rate is zero.
From Theorem~\ref{thm:necessary},
we obtain exact values of the segment entanglement rates
$e_{i,j}(\phi_5,\{\mathcal{R}_n\})$ for $i\neq j$ as follows:
\begin{equation} \label{eq:rates2}
e_{i,j}(\phi_5,\{\mathcal{R}_n\})=
\begin{cases}
1
&\text{if $i,j\in\{2,4,6\}$} \\
-1
&\text{if $i,j\in\{1,3,5\}$} \\
0
&\text{otherwise}.
\end{cases}
\end{equation}
Illustrations for quantum systems of entanglement resources giving non-zero segment entanglement rates
are provided in Fig.~\ref{fig:ConGen}.

\begin{figure}
\includegraphics[clip,width=.85\columnwidth]{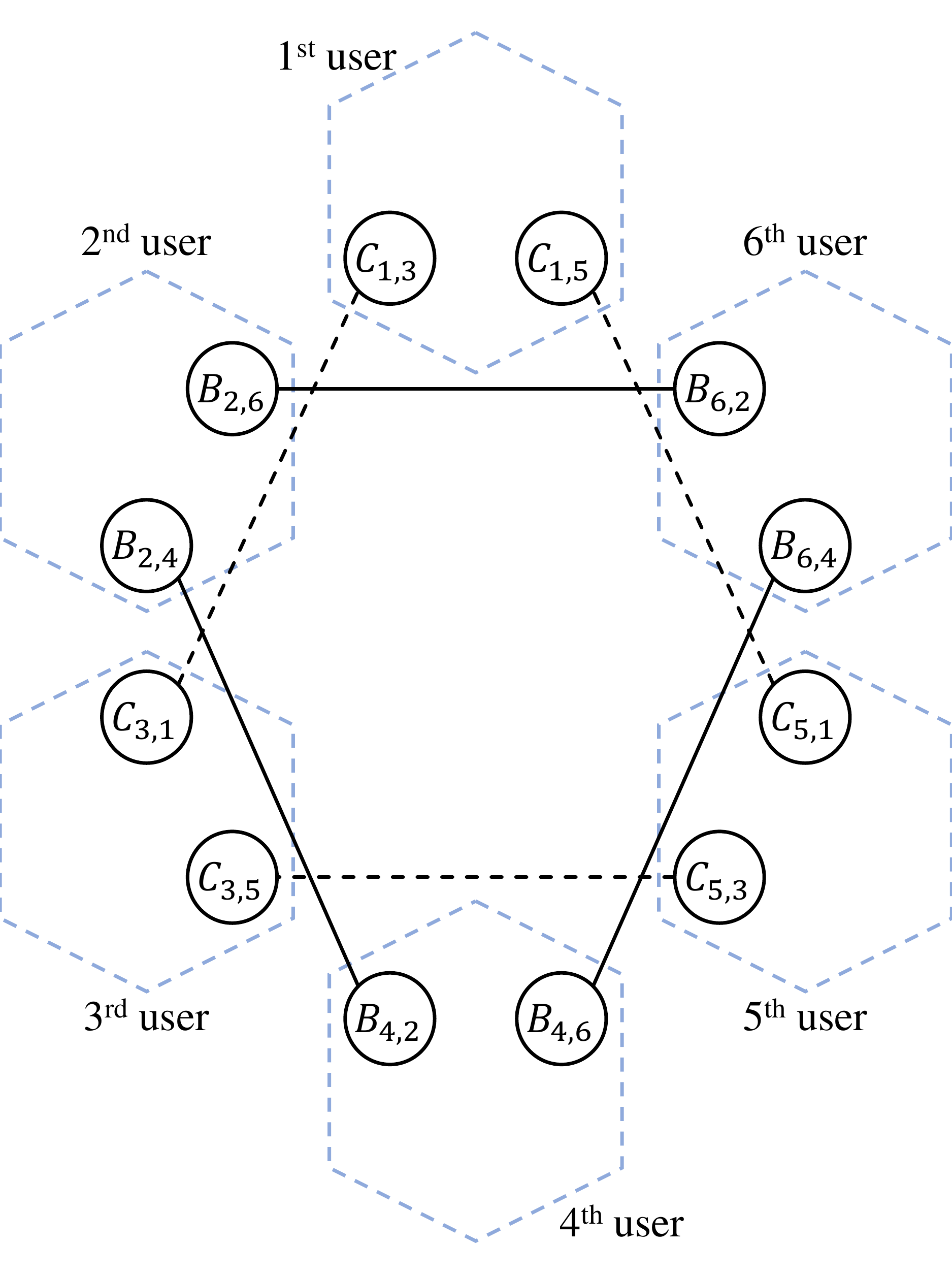}
\caption{
Systems of entanglement resources giving positive/negative segment entanglement rates:
Circles indicate quantum systems for entanglement resources,
and, for each $i\ge7$, the $i^{\mathrm{th}}$ user and his/her quantum systems are not explicitly illustrated.
If we assume that the initial state $\phi_5$ in Eq.~(\ref{eq:example1}) can be rotated
by a sequence $\{\mathcal{R}_n\}_{n\in\mathbb{N}}$ of QSR protocols whose achievable total entanglement rate is zero,
then Theorem~\ref{thm:necessary} implies that all segment entanglement rates have one of three values 1, 0, -1,
as shown in Eq.~(\ref{eq:rates2}).
In this illustration, consumed (generated) entanglement resources corresponding to positive (negative) segment entanglement rates are described as circles connected by straight (dashed) lines.
Entanglement resources for zero segment entanglement rates are not explicitly illustrated.
}
\label{fig:ConGen}
\end{figure}

The sequence $\{\mathcal{R}_n\}_{n\in\mathbb{N}}$ and its segment entanglement rates $e_{i,j}$ in Eq.~(\ref{eq:rates2}) imply that it is possible to carry out the TU task by means of LOCC assisted by the catalytic use of entanglement under the asymptotic scenario.
To be specific, recall that each $\mathcal{R}_n$ is LOCC protocol transforming the initial state $\phi_5^{\otimes n}$ and the input entanglement resources $\tilde{\Psi}_n$ into the final state $\phi_{\mathrm{f}}^{\otimes n}$ and the output entanglement resources $\tilde{\Phi}_n$ with error $\varepsilon_n$,
where $\phi_{\mathrm{f}}$ is the final state of the QSR task corresponding to the initial state $\phi_5$.
Note that, in Eq.~(\ref{eq:rates2}), the zero segment entanglement rate $e_{i,j}$ means that the $i^{\mathrm{th}}$ user and the $j^{\mathrm{th}}$ user catalytically use entanglement resources in the asymptotic scenario.
Thus, the sequence $\{\mathcal{R}_n\}_{n\in\mathbb{N}}$ of the QSR protocols can be considered as a sequence of LOCC protocols assisted by the catalytic use of entanglement, which asymptotically transforms an initial state
\begin{eqnarray}
\ket{\eta}&=&\ket{\mathrm{GHZ}}_{A_{1,1}A_{3,1}A_{5,1}}
\otimes
\ket{\mathrm{GHZ}}_{A_{1,2}A_{3,2}A_{5,2}}
\otimes
\sbigotimes_{i\in[M]\setminus\{1,3,5\}} \ket{\varphi_i}_{A_i} \nonumber \\
&&\otimes \ket{\mathrm{ebit}}_{B_{2,4}B_{4,2}}\otimes \ket{\mathrm{ebit}}_{B_{4,6}B_{6,4}}\otimes \ket{\mathrm{ebit}}_{B_{6,2}B_{2,6}}
\end{eqnarray}
into a final state
\begin{eqnarray}
\ket{\eta_\mathrm{f}}&=&\ket{\mathrm{GHZ}}_{A'_{1,1}A'_{3,1}A'_{5,1}}
\otimes
\ket{\mathrm{GHZ}}_{A'_{1,2}A'_{3,2}A'_{5,2}}
\otimes
\sbigotimes_{i\in[M]\setminus\{1,3,5\}} \ket{\varphi_i}_{A'_i} \nonumber \\
&&\otimes \ket{\mathrm{ebit}}_{C_{1,3}C_{3,1}}\otimes \ket{\mathrm{ebit}}_{C_{3,5}C_{5,3}}\otimes \ket{\mathrm{ebit}}_{C_{5,1}C_{1,5}}.
\end{eqnarray}
In this case, if the $1^{\mathrm{st}}$ user has all systems of the others
except for the $3^{\mathrm{rd}}$ user and the $5^{\mathrm{th}}$ user, and the $1^{\mathrm{st}}$ user play the roles of the rest except for the $3^{\mathrm{rd}}$ user and the $5^{\mathrm{th}}$ user,
then the $1^{\mathrm{st}}$ user can locally prepare the three ebits and the pure quantum states $\varphi_i$ with $i\in[M]\setminus\{1,3,5\}$ of the initial state $\eta$ and the two GHZ states of the final state $\eta_{\mathrm{f}}$,
and the $3^{\mathrm{rd}}$ user and the $5^{\mathrm{th}}$ user can locally prepare the pure quantum states $\varphi_2$ and $\varphi_4$, respectively.
It follows that there exists a sequence of LOCC protocols assisted by the catalytic use of entanglement, which asymptotically transforms a quantum state 
\begin{equation}
\ket{\mathrm{GHZ}}_{A_{1,1}A_{3,1}A_{5,1}}
\otimes
\ket{\mathrm{GHZ}}_{A_{1,2}A_{3,2}A_{5,2}}
\end{equation}
into a quantum state
\begin{equation}
\ket{\mathrm{ebit}}_{C_{1,3}C_{3,1}}\otimes \ket{\mathrm{ebit}}_{C_{3,5}C_{5,3}}\otimes \ket{\mathrm{ebit}}_{C_{5,1}C_{1,5}}.
\end{equation}
This means that two GHZ states shared by the $1^{\mathrm{st}}$ user, the $3^{\mathrm{rd}}$ user, and  the $5^{\mathrm{th}}$ user are transformed into
the three ebits symmetrically shared among the three users by means of LOCC
and the catalytic use of entanglement resources
under the asymptotic scenario.
However, this contradicts to Lemma~\ref{lem:GHZtoSinglet}.
Hence,
the achievable total entanglement rate $r$ is positive.
\end{proof}

We remark that
it is not sufficient to consider initial states
similar to the state $\phi_5$ in Eq.~(\ref{eq:example1})
in order to prove Proposition~\ref{prop:example}
with respect to $M=3,4,5$.
For example,
consider the initial state 
\begin{equation}
\ket{\phi_6}_A
=\ket{\mathrm{GHZ}}_{A_{1,1}A_{2,1}A_{3,1}}
\otimes
\ket{\mathrm{GHZ}}_{A_{1,2}A_{2,2}A_{3,2}}
\otimes \ket{\varphi}_{A_4},
\end{equation}
where $A_i=A_{i,1}A_{i,2}$ for $i=1,2,3$,
and $\ket{\varphi}$ is any pure quantum state.
If there exists a sequence $\{\mathcal{R}_n\}$ of QSR protocols
for $\phi_6$
whose achievable total entanglement rate is zero,
then Theorem~\ref{thm:necessary} tells us that
its segment entanglement rates are determined as
\begin{eqnarray}
&&e_{1,2}(\phi_6,\{\mathcal{R}_n\})
=e_{1,3}(\phi_6,\{\mathcal{R}_n\})
=-1, \\
&&e_{1,4}(\phi_6,\{\mathcal{R}_n\})
=e_{2,3}(\phi_6,\{\mathcal{R}_n\})
=0, \\
&&e_{2,4}(\phi_6,\{\mathcal{R}_n\})
=e_{3,4}(\phi_6,\{\mathcal{R}_n\})
=1.
\end{eqnarray}
To the best of our knowledge,
whether such a sequence exists or not is unknown.
On this account,
it is hard to prove Proposition~\ref{prop:example}
for $M=3,4,5$,
as long as we stick to initial states consisting of the two GHZ states.

\section{Examples} \label{sec:Exmpls}

\subsection{SWAP-invariant initial states}

\begin{figure}
\includegraphics[clip,width=\columnwidth]{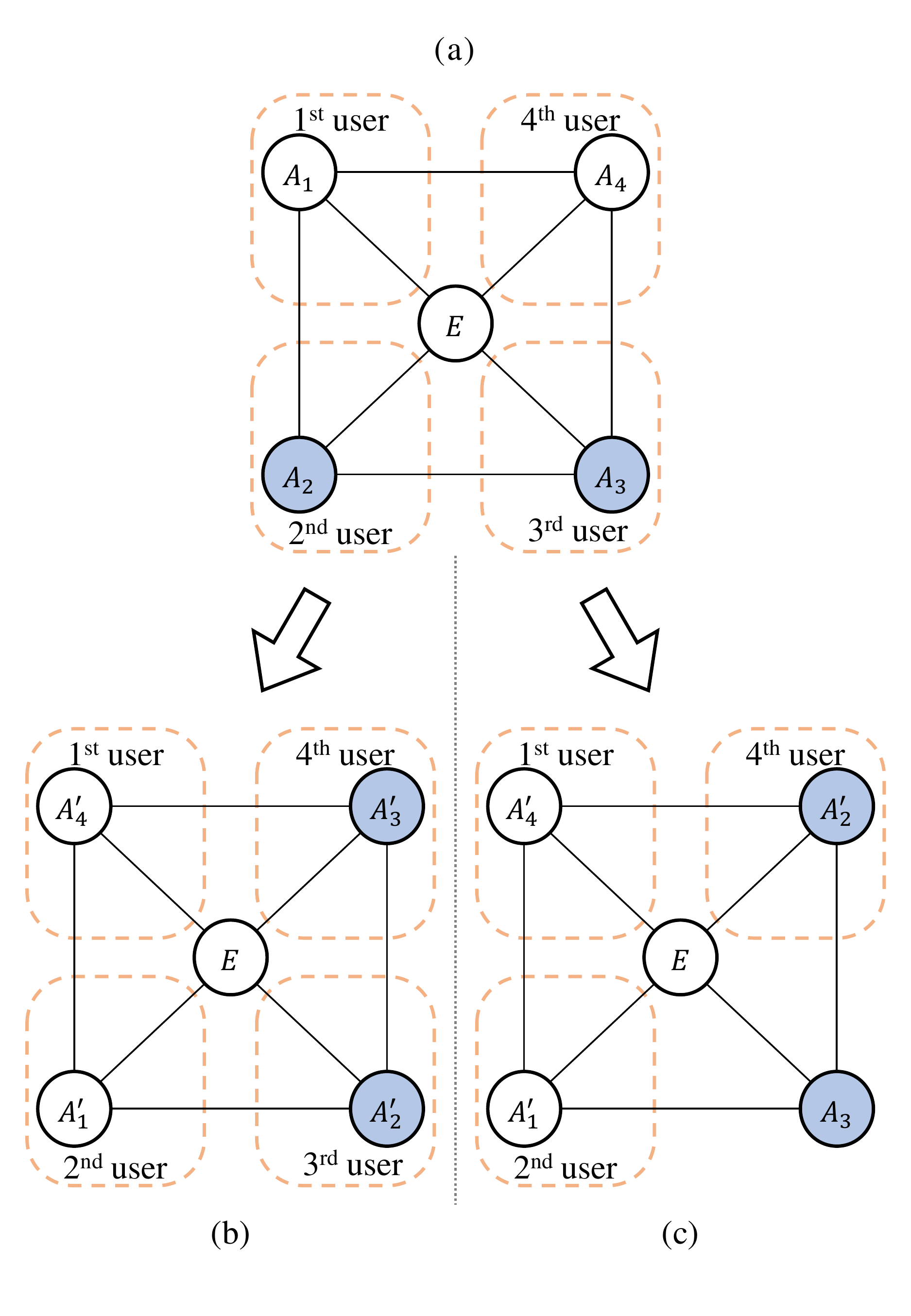}
\caption{
In each illustration,
circles indicate quantum systems for the initial and final states,
and two blue circles represent symmetric parts;
(a) Initial state $\ket{\psi}_{A_1A_2A_3A_4E}$ of the quantum state rotation task for four users:
The initial state is SWAP-invariant on systems $A_2$ and $A_3$;
(b) Final state rotated by all users;
(c) Final state rotated by the $1^{\mathrm{st}}$ user, the $2^{\mathrm{nd}}$ user, and the $4^{\mathrm{th}}$ user:
Here, the quantum system $A_3$ is considered as a part of the environment.
}
\label{fig:problem}
\end{figure}

In this section,
we see that reduction of the number of users in the QSR task
does not necessarily reduce the OEC of the task.

The initial state $\ket{\psi}_{A_1A_2A_3A_4E}$ of the QSR task
is said to be \emph{SWAP-invariant on systems $A_2$ and $A_3$},
if it satisfies
\begin{equation}
\left(\mathrm{SWAP}_{A_2\leftrightarrow A_3}\right)(\psi)=\psi,
\end{equation}
where $\mathrm{SWAP}_{X\leftrightarrow Y}$
is a quantum channel swapping quantum states in quantum systems $X$ and $Y$.
Let us consider the QSR task of the SWAP-invariant initial state $\ket{\psi}_{A_1A_2A_3A_4E}$.
We provide illustrations of the SWAP-invariant initial state and its final state in Fig.~\ref{fig:problem}(a)
and Fig.~\ref{fig:problem}(b),
respectively.
From the viewpoint of the $3^{\mathrm{rd}}$ user,
the part $A'_2$ of the final state $\psi_{\mathrm{f}}$
is identical to the part $A_3$ of the initial state $\psi$.
So, it is possible to exclude the $3^{\mathrm{rd}}$ user to carry out
the QSR task of the four users,
and so the $3^{\mathrm{rd}}$ user does nothing,
since this task can be done
by the $2^{\mathrm{nd}}$ user directly transmitting his/her quantum state
to the $4^{\mathrm{th}}$ user, as described in Fig.~\ref{fig:problem}(c).
In other words,
the original QSR task of the four users can be replaced by
the QSR task of the $1^{\mathrm{st}}$ user, the $2^{\mathrm{nd}}$ user, and the $4^{\mathrm{th}}$ user for the same initial state.

Let $e_{\mathrm{opt}}^{(3)}(\psi)$ and $e_{\mathrm{opt}}^{(4)}(\psi)$
be the OECs for the QSR tasks of the initial state $\ket{\psi}_{A_1A_2A_3A_4E}$ performed by the three users and the four users,
respectively.
In this case,
are two OECs
$e_{\mathrm{opt}}^{(3)}(\psi)$ and $e_{\mathrm{opt}}^{(4)}(\psi)$ equal?
One may guess that
$e_{\mathrm{opt}}^{(3)}(\psi)\le e_{\mathrm{opt}}^{(4)}(\psi)$ holds in general,
since the part $A_3$ does not need to be transmitted during the second.

However, this is not the case.
Consider the SWAP-invariant initial state
\begin{equation} \label{eq:exam}
\ket{\phi_7}_{A_1A_2A_3A_4} = \ket{\varphi_1}_{A_1}\otimes\ket{\mathrm{ebit}}_{A_2A_3}\otimes \ket{\varphi_2}_{A_4},
\end{equation}
where $E$ is regarded as a one-dimensional system, $\ket{\varphi_1}$ and $\ket{\varphi_2}$ are any pure quantum states,
and
$\ket{\mathrm{ebit}}$ is presented in Eq.~(\ref{eq:GHZebit}).

If the $1^{\mathrm{st}}$ user, the $2^{\mathrm{nd}}$ user, and the $4^{\mathrm{th}}$ user rotate the initial state $\phi_7$ without the $3^{\mathrm{rd}}$ user,
then this QSR task is nothing but Schumacher compression~\cite{S95,W13}
in which the part $A_2$ is transmitted to the $4^{\mathrm{th}}$ user by consuming ebits instead of qubit channels.
This is because the quantum states $\varphi_1$ and $\varphi_2$ can be locally prepared by the $2^{\mathrm{nd}}$ user and the $1^{\mathrm{st}}$ user,
respectively,
without consuming and gaining any entanglement resource.
It turns out that
the minimal amount of entanglement required for this Schumacher compression is $H(A_2)_{\phi_7}$,
and Theorem~\ref{thm:lb} implies that $H(A_2)_{\phi_7}$ is a lower bound on the OEC of this QSR task of three users.
Thus, we have
$e_{\mathrm{opt}}^{(3)}(\phi_7)=H(A_2)_{\phi_7}$.

On the other hand, in the QSR task of four users,
the $3^{\mathrm{rd}}$ user can locally prepare an ebit, and then the $3^{\mathrm{rd}}$ user can share the ebit with the $4^{\mathrm{th}}$ user
by using the Schumacher compression~\cite{S95,W13}
and the quantum teleportation~\cite{BBCJPW93}.
The amount of entanglement consumed in this transmission is $H(A_2)_{\phi_7}$.
The $2^{\mathrm{nd}}$ user and the $3^{\mathrm{rd}}$ user can gain the same amount of entanglement
by distilling the ebit on the systems $A_2$ and $A_3$.
Lastly,
without any entanglement resource,
the $1^{\mathrm{st}}$ user and the $2^{\mathrm{nd}}$ user locally prepare $\varphi_2$ and $\varphi_1$, respectively.
In this way, the initial state $\phi_7$ is rotated,
and the achievable total entanglement rate becomes zero in this case.
The non-negativity of the OEC implies $e_{\mathrm{opt}}^{(4)}(\phi_7)=0$.

Therefore,
we obtain that $e_{\mathrm{opt}}^{(3)}(\phi_7)> e_{\mathrm{opt}}^{(4)}(\phi_7)$
holds for the SWAP-invariant initial state $\phi_7$.
This means that
even though the $3^{\mathrm{rd}}$ user does not have to participate in the QSR task,
helping the rest users to achieve the task can reduce the OEC.

\subsection{Quantum state rotation with cooperation}

\begin{figure}
\subfloat[]{%
  \includegraphics[clip,width=\columnwidth]{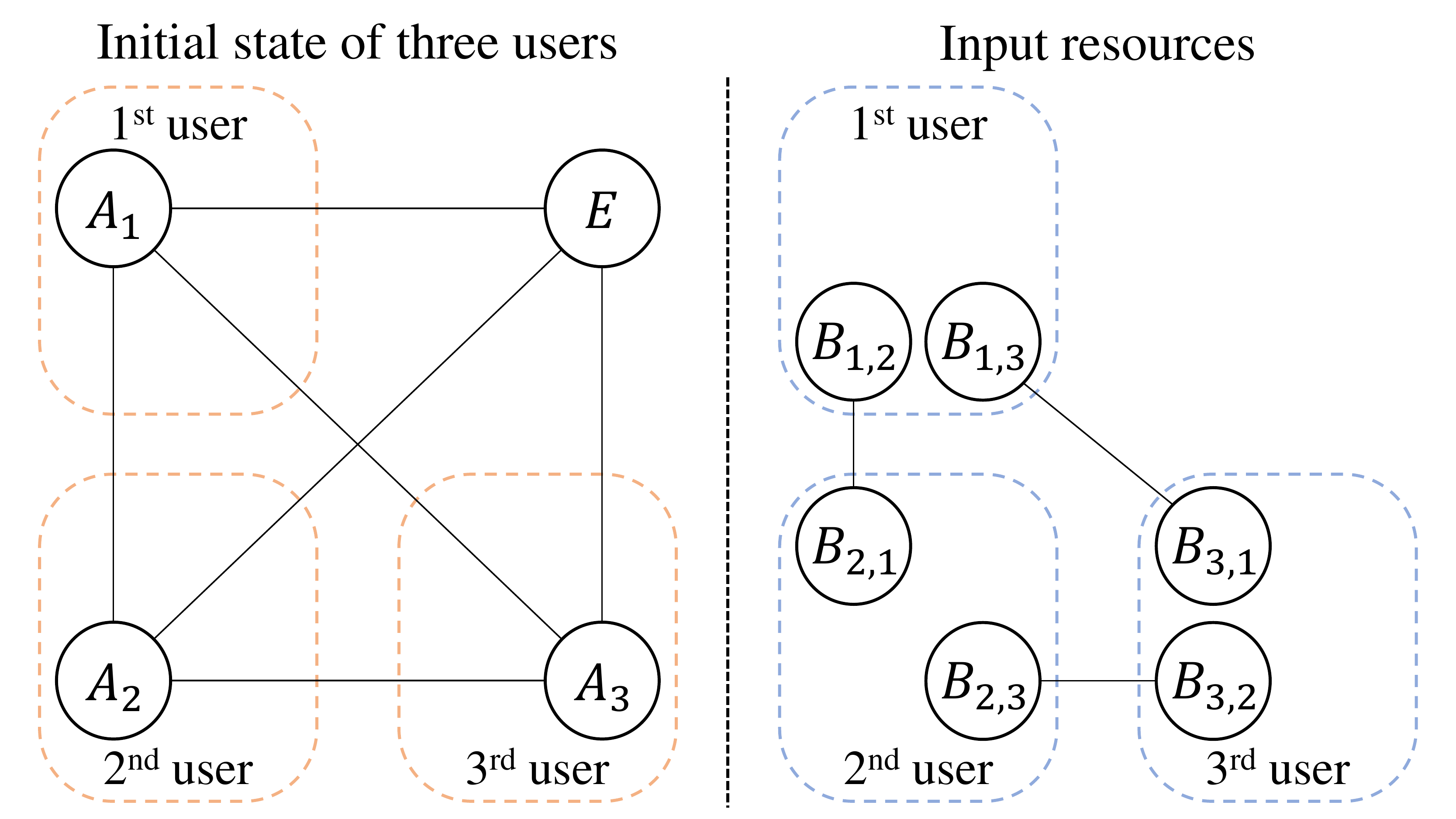}%
}

\subfloat[]{%
  \includegraphics[clip,width=\columnwidth]{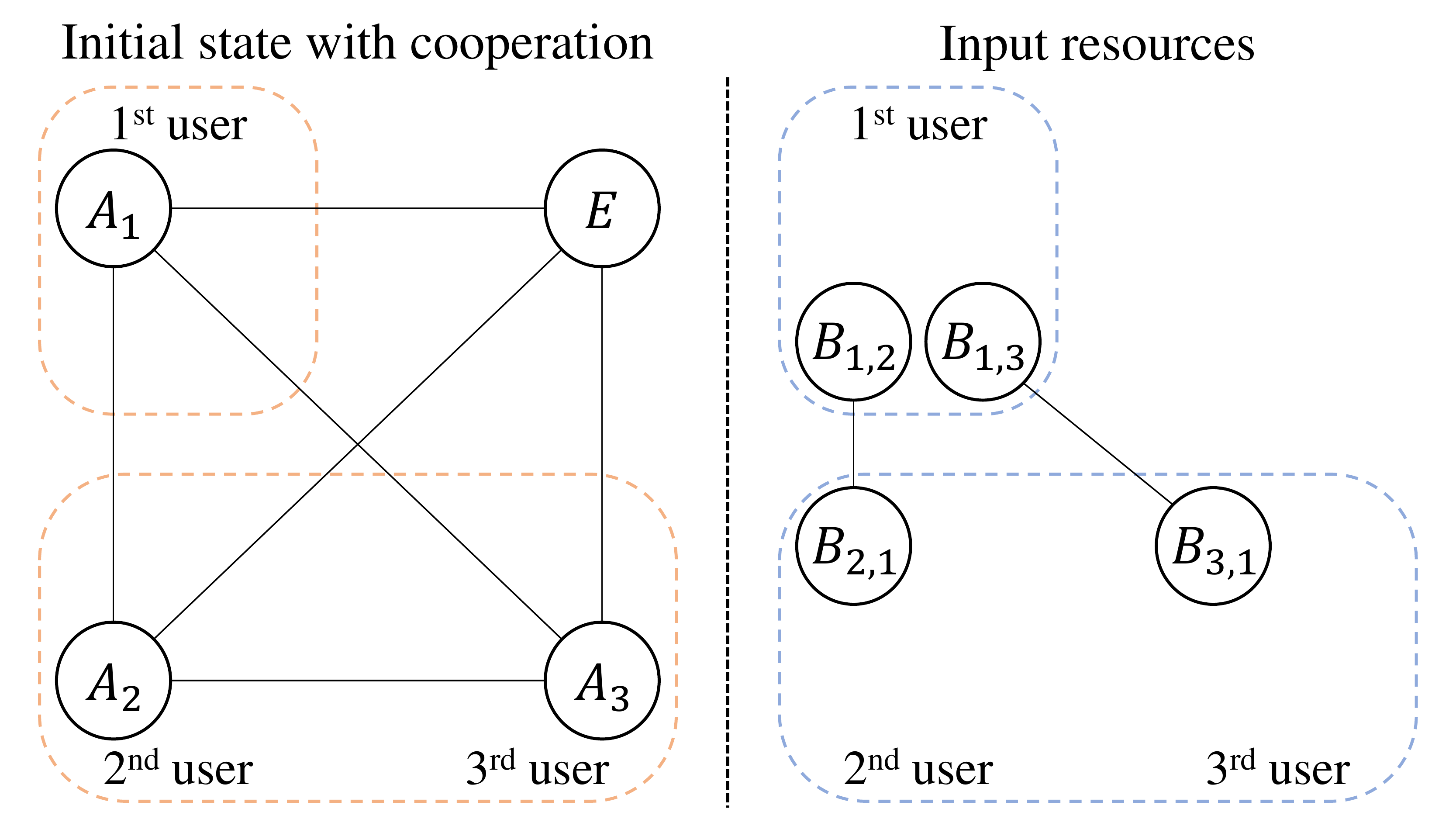}%
}

\caption{
(a) Initial state and input entanglement resources for the quantum state rotation task of three users;
(b) Initial state and input entanglement resources when the $2^{\mathrm{nd}}$ user and the $3^{\mathrm{rd}}$ user cooperate:
To cooperate, they gathered in the same laboratory,
and so an entangled state in the systems $B_{2,3}$ and $B_{3,2}$ is not considered as a non-local resource.
In the second illustration, these systems are not described.
}
\label{fig:problem2}
\end{figure}

In this section,
we answer the following question:
If some of the users are allowed not only LOCC but nonlocal (global) operations on their shared quantum systems,
can they perform the QSR task at a smaller OEC?

We consider the QSR task of $\ket{\psi}_{A_1A_2A_3E}$ performed by three users,
as shown in Fig.~\ref{fig:problem2}(a),
and we modify this QSR task
by assuming that the $2^{\mathrm{nd}}$ user and the $3^{\mathrm{rd}}$ user are in the same laboratory in order to cooperate,
as depicted in Fig.~\ref{fig:problem2}(b).
In this case,
the $2^{\mathrm{nd}}$ user and the $3^{\mathrm{rd}}$ user can apply
any quantum operations to their quantum states in the laboratory,
but pure maximally entangled states shared by the $2^{\mathrm{nd}}$ user and the $3^{\mathrm{rd}}$ user
are not considered as non-local resources in this modified task,
since the entangled states can be locally prepared in their laboratory.
On this account,
while the three users in Fig.~\ref{fig:problem2}(b) can make use of any QSR protocol of $\ket{\psi}_{A_1A_2A_3E}$ in Fig.~\ref{fig:problem2}(a),
it is hard to guess the minimal amount of entanglement consumed by two laboratories in the modified task.

Under this setting,
one may guess that the OEC of the modified task is less than or equal to that of the original one.
However,
the initial state
\begin{equation}
\ket{\phi_8}_{A_1A_2A_3} = \ket{\varphi_1}_{A_1}\otimes\ket{\varphi_2}_{A_2A_3}
\end{equation}
shows that such a guess is wrong,
where $E$ is regarded as a one-dimensional system, $\ket{\varphi_1}$ is any pure quantum state,
and $\ket{\varphi_2}$ is any pure entangled state.
In this case, the initial state $\phi_8$ can be rotated by three users with a zero achievable total entanglement rate as follows:
The pure quantum state $\ket{\varphi_1}$ is prepared by the $2^{\mathrm{nd}}$ user, and
the $1^{\mathrm{st}}$ user and the $3^{\mathrm{rd}}$ user can share the pure quantum state $\ket{\varphi_2}$ by using Schumacher compression~\cite{S95,W13}
together with the quantum teleportation~\cite{BBCJPW93}. The amount of entanglement consumed in this transmission is $H(A_2)_{\phi_8}$.
The $2^{\mathrm{nd}}$ user and the $3^{\mathrm{rd}}$ user can gain the same amount of entanglement by applying entanglement distillation~\cite{BBPS96,BDSW96,BBPSSW97}
to $\varphi_2$ on the systems $A_2A_3$.
Thus, the OEC is zero when the QSR task is performed without any cooperation.

On the other hand,
in the modified task,
the $1^{\mathrm{st}}$ user and the $2^{\mathrm{nd}}$ ($3^{\mathrm{rd}}$) user cannot share the quantum state $\varphi_2$
without consuming any entanglement resources between two laboratories,
since $\varphi_2$ is entangled.
This means that the OEC of the modified task is positive.
Therefore, from the initial state $\phi_8$, we know that
the OEC for the original task without any cooperation can be less
than that of the modified task in which some of the users cooperate.
This is because one does not take into account gain as well as consumption of entanglement resources between them when computing the OEC of the modified task with the cooperation of the $2^{\mathrm{nd}}$ user and the $3^{\mathrm{rd}}$ user.     

\section{Conclusion} \label{sec:conclusion}

In this work, we have introduced the QSR task in which the $M$ users circularly transfer their respective quantum states via entanglement-assisted LOCC.
We have considered the QSR as
a fundamental quantum communication task for $M$ users
and have investigated the minimal amount of entanglement consumed among the users
under the asymptotic scenario.
For this investigation,
we have formally formulated the QSR protocol,
the achievable total entanglement rate,
and the OEC.
We have derived lower and upper bounds on the OEC,
and have presented conditions on zero OECs
and zero achievable total entanglement rates. 

The QSR task includes the QSE task~\cite{OW08,LTYAL19,LYAL19} as a special case,
in which two users, Alice and Bob, exchange their respective quantum states
via entanglement-assisted LOCC.
However,
the QSR task is not a direct generalization of the QSE task.
That is,
we have shown that
there is a unique property of the QSR task
not appearing in QSE tasks for two users:
Not all initial states without the environment system can be rotated without consuming any entanglement,
while such states can be exchanged at zero entanglement cost via local unitary operations.
We have also considered two specific settings of QSR tasks.
In the first setting,
some users do not have to participate in the task.
In the second,
some of the users can cooperate by using non-local operations.
For some initial states,
we have shown that
the OEC for the original QSR task
can be smaller than those for each setting.

While the lower bound $l$ presented in Eq.~(\ref{eq:LB})
is helpful to evaluate the OEC,
it becomes zero for initial states without the environment system $E$.
This means that it is not straightforward
to determine whether the OECs
for such initial states are zero or not,
unless we can explicitly construct an optimal QSR protocol.
This is the main reason why we used the result of Theorem~\ref{thm:necessary} that
the segment entanglement rates are determined in terms of
the von Neumann entropies of the initial state,
in order to prove Proposition~\ref{prop:example} instead of the lower bound $l$.
On this account,
finding tighter lower bounds can be a meaningful future work.

As potential applications of our work,
the QSR task can serve as one of the fundamental sub-routines
in distributed quantum computing~\cite{CEHM99,BDLMSS04} and quantum networks~\cite{CZKM97,AML16},
since they usually involve more than two users.
In addition,
the QSR task can be used
as a sub-task of more general quantum communication tasks.
For example,
let $\sigma$ be a permutation on $[M]$,
then we can devise a new quantum communication task for $M$ user
in which the $i^{\mathrm{th}}$ user transmits his/her quantum state to
the $\sigma(i)^{\mathrm{th}}$ user
by means of entanglement-assisted LOCC.
We call this task \emph{quantum state permutation}.
It is a well-known fact that
any permutation on a finite set has a unique cycle decomposition,
i.e.,
the permutation is expressed as a product of disjoint cycles.
So, the quantum state permutation task with respect to $\sigma$ can be decomposed as QSR sub-tasks,
since the QSR tasks intuitively corresponds to disjoint cycles.
In this situation, our results for the QSR task can be useful tools to investigate the OEC
for the quantum state permutation task.

\section*{ACKNOWLEDGMENTS}
This research was supported by Basic Science Research Program
through the National Research Foundation of Korea(NRF)
funded by the Ministry of Education(Grant No. NRF-2019R1A6A3A01092426
and Grant No. NRF-2020R1I1A1A01058364).
H.\ Y.\ acknowledges CREST (Japan Science and Technology Agency) JPMJCR1671, Cross-ministerial Strategic Innovation Promotion Program (SIP) (Council for Science, Technology and Innovation (CSTI)), and JSPS Overseas Research Fellowships.
S.L. acknowledges support from the Basic Science Research Program through the National Research Foundation of Korea funded by the Ministry of Science and ICT (Grant No. NRF-2019R1A2C1006337) and the Ministry of Science and ICT, Korea, under the Information Technology Research Center support program (Grant No. IITP-2021-2018-0-01402) supervised by the Institute for Information and Communications Technology Promotion.

\bibliography{QSRbib}

\begin{widetext}
\appendix



\section{Proof of Lemma~\ref{lem:LBOBER}} \label{app:lem:lb}

In this appendix, we prove Lemma~\ref{lem:LBOBER} in the main text.
If a protocol $\mathcal{R}$ in Eq.~(\ref{eq:DomainRange}) satisfies
$(\mathcal{R} \otimes \mathrm{id}_{\mathcal{L}(E)})
(\psi \otimes \tilde{\Psi})
= \psi_{\mathrm{f}} \otimes \tilde{\Phi}$,
then $\mathcal{R}$ is said to be \emph{exact}. 
The case regarding only exact QSR protocols is called
an \emph{exact scenario}.

To express two sets of the $M$ users,
we make use of a partition of the set $[M]$.
Let $\{P,P^{\mathsf{c}}\}$ be a partition of the set $[M]$,
where $P$ is any non-empty proper subset of $[M]$,
and $P^{\mathsf{c}}$ is the complement of $P$,
i.e., $P^{\mathsf{c}}=[M]\setminus P$.
If we interpret an element $i$ of the set $[M]$
as the $i^{\mathrm{th}}$ user of the QSR task,
then we can divide the $M$ users into two disjoint subsets $P$ and $P^{\mathsf{c}}$
via the partition $\{P,P^{\mathsf{c}}\}$.

For an exact QSR protocol $\mathcal{R}$ of $\psi$,
we define the \emph{bipartite entanglement difference} $d_P(\psi,\mathcal{R})$
for a partition $\{P,P^{\mathsf{c}}\}$ as
\begin{equation}
d_P(\psi,\mathcal{R})
= \sum_{ i\in P } \sum_{ j\in P^{\mathsf{c}} }
\left[ \log d_{B_{i,j}} -\log d_{C_{i,j}} \right],
\end{equation}
where $d_{B_{i,j}}$ $(d_{C_{i,j}})$ indicate the Schmidt rank of the entanglement resource $\Psi_{i,j}$ $(\Phi_{i,j})$
shared by the $i^{\mathrm{th}}$ user and the $j^{\mathrm{th}}$ user before (after) performing the QSR protocol $\mathcal{R}$.
The following proposition provides a lower bound
on the bipartite entanglement difference 
for the QSR task of the initial state~$\psi$.

\begin{Prop} \label{prop:LBBED}
Let $\ket{\psi}_{AE}$
be the initial state of the QSR task.
The bipartite entanglement difference $d_P(\psi,\mathcal{R})$ for a partition $\{P,P^{\mathsf{c}}\}$ is lower bounded by
\begin{equation}
d_P(\psi,\mathcal{R})
\ge
l_P(\psi)
=\max_{U}
\left\{
 H\left(\sbigotimes_{i\in P} A_{i-1}V\right)_{U\ket{\psi}}
-H\left(\sbigotimes_{i\in P} A_{i }V\right)_{U\ket{\psi}}
\right\},
\end{equation}
where the maximum is taken over all isometries $U$ from $E$ to $V\otimes W$, $V$ and $W$ are any quantum systems,
and $U\ket{\psi}$ is an abbreviation for $\mathds{1}_{A}\otimes U\ket{\psi}$.
\end{Prop}

\begin{proof}
Let us consider an \emph{$R$-assisted QSR} task
whose idea comes
from Refs.~\cite{OW08,LTYAL19,LYAL19}.
While the environment system $E$ of the initial state $\ket{\psi}_{AE}$ is not owned by any users of the original QSR task,
in the $R$-assisted QSR task, we additionally consider a referee who has the environment system $E$.
In this task,
the referee can assist $M$ users as follows:
The referee divides his part $E$ of the initial state $\ket{\psi}_{AE}$
into two parts $V$ and $W$.
To be specific, the referee locally applies an isometry $U\colon E\to V\otimes W$~\cite{W13} to his quantum state on the quantum system $E$,
and so the initial state $\ket{\psi}_{AE}$ becomes a quantum state $\ket{\xi}_{AVW}$ satisfying $\Tr_{E}\psi=\Tr_{VW}\xi$.
The referee now transfers his quantum state on the system $V$ ($W$)
to one of the users belonging to the set $P$ ($P^{\mathsf{c}}$),
so that the $M$ users can share the quantum state $\ket{\xi}_{AVW}$.

After finishing the referee's assistance,
$M$ users rotate the quantum state $\ket{\xi}_{AVW}$
via entanglement-assisted LOCC,
as in the original QSR.
To be specific, the quantum systems $V$ and $W$ of the users are not rotated during the QSR task, while the user can use them as quantum side information, as in other quantum communication tasks~\cite{HOW05,HOW06,DY08,YD09,ADHW09,LTYAL19,LYAL19}.
In the following, we call such a protocol an exact $R$-assisted QSR protocol of the state $\ket{\xi}_{AVW}$,
and it is denoted by $\mathcal{A}$.
Since $\mathcal{A}$ is LOCC among the $M$ users,
it is also LOCC between two disjoint subsets $P$ and $P^{\mathsf{c}}$
of the users.
By using the fact that the amount of entanglement
between two sets $P$ and $P^{\mathsf{c}}$ of the users
cannot increase on average via LOCC~\cite{BDSW96},
we obtain the inequality
\begin{equation} \label{eq:IE_ERLOCC}
H\left(\sbigotimes_{i\in P} A_{i}VB_i\right)_{\xi \otimes \tilde{\Psi}}
\ge
H\left(\sbigotimes_{i\in P} A'_{i-1}VC_i\right)_{\xi_{\mathrm{f}} \otimes \tilde{\Phi}}.
\end{equation}
By using the additivity of the von Neumann entropy~\cite{W13},
we obtain
\begin{equation}
H\left(\sbigotimes_{i\in P} A_{i}VB_i\right)_{\xi \otimes \tilde{\Psi}}
=H\left(\sbigotimes_{i\in P} A_{i}V\right)_{\xi}
+H\left(\sbigotimes_{i\in P} B_i\right)_{\tilde{\Psi}}.
\end{equation}
Recall that the quantum state $\tilde{\Psi}$ is defined as the tensor product of bipartite maximally entangled states as in Eq.~(\ref{Eq:ERs}), and the systems $B_i$ are defined as in Eq.~(\ref{eq:ERSystems}).
The additivity of the von Neumann entropy~\cite{W13} implies
\begin{equation} \label{eq:MoreSpeci}
H\left(\sbigotimes_{i\in P} B_i\right)_{\tilde{\Psi}}
=\sum_{\substack{ i,j\in P\\ i<j}} H(B_{i,j}B_{j,i})_{\Psi_{i,j}}
+\sum_{ i\in P } \sum_{ j\in P^{\mathsf{c}} } H(B_{i,j})_{\Psi_{i,j}}
=\sum_{ i\in P } \sum_{ j\in P^{\mathsf{c}} } \log d_{B_{i,j}}.
\end{equation}
Since $\Psi_{i,j}$ is a pure bipartite maximally entangled state on quantum systems $B_{i,j}B_{j,i}$
whose Schmidt rank is $d_{B_{i,j}}$, 
$H(B_{i,j}B_{j,i})_{\Psi_{i,j}}=0$ and $H(B_{i,j})_{\Psi_{i,j}}=\log d_{B_{i,j}}$ hold for each $i\neq j$.
The second equality in Eq.~(\ref{eq:MoreSpeci}) comes from this fact.
So we obtain
\begin{equation}
H\left(\sbigotimes_{i\in P} A_{i}B_iV\right)_{\xi \otimes \tilde{\Psi}}
=H\left(\sbigotimes_{i\in P} A_{i}V\right)_{\xi}
+\sum_{ i\in P } \sum_{ j\in P^{\mathsf{c}} } \log d_{B_{i,j}}.
\end{equation}
By using the same method, we obtain
\begin{equation}
H\left(\sbigotimes_{i\in P} A'_{i-1}C_iV\right)_{\xi_{\mathrm{f}} \otimes \tilde{\Phi}}
=H\left(\sbigotimes_{i\in P} A'_{i-1}V\right)_{\xi_{\mathrm{f}}}
+\sum_{ i\in P } \sum_{ j\in P^{\mathsf{c}} } \log d_{C_{i,j}},
\end{equation}
where $d_{C_{i,j}}$ is the Schmidt rank of the entanglement resource $\Phi_{i,j}$ on quantum systems $C_{i,j}C_{j,i}$.
Consequently,
the inequality in Eq.~(\ref{eq:IE_ERLOCC}) is rewritten as
\begin{equation}
\sum_{ i\in P } \sum_{ j\in P^{\mathsf{c}} }
\left[ \log d_{B_{i,j}} -\log d_{C_{i,j}} \right]
\ge
H\left(\sbigotimes_{i\in P} A'_{i-1}V\right)_{\xi_{\mathrm{f}}}
-H\left(\sbigotimes_{i\in P} A_{i}V\right)_{\xi}
=
H\left(\sbigotimes_{i\in P} A'_{i-1}V\right)_{U\ket{\psi_{\mathrm{f}}}}
-H\left(\sbigotimes_{i\in P} A_{i}V\right)_{U\ket{\psi}},
\end{equation}
where $U\ket{\psi_{\mathrm{f}}}$ and $U\ket{\psi}$ are abbreviations for $\mathds{1}_{A}\otimes U\ket{\psi_{\mathrm{f}}}$ and $\mathds{1}_{A}\otimes U\ket{\psi}$,
respectively.
By the definition of the final state $\psi_{\mathrm{f}}$,
we obtain that
\begin{equation} \label{eq:inifinsame}
H\left(\sbigotimes_{i\in P} A'_{i-1}V\right)_{U\ket{\psi_{\mathrm{f}}}}
=H\left(\sbigotimes_{i\in P} A_{i-1}V\right)_{U\ket{\psi}}
\end{equation}
holds.
It follows that
\begin{equation}
d_P(\psi,\mathcal{A})
\ge
H\left(\sbigotimes_{i\in P} A_{i-1}V\right)_{U\ket{\psi}}
-
H\left(\sbigotimes_{i\in P} A_{i}V\right)_{U\ket{\psi}}.
\end{equation}
Note that the above inequality holds for any quantum systems $V$ and $W$ and any isometry $U\colon E\to V\otimes W$.
We further note that any exact QSR protocol of $\ket{\psi}_{AE}$ is the special case of the exact $R$-assisted QSR protocol in which the referee does not assist the users.
It follows that $d_P(\psi,\mathcal{R}) \ge l_P(\psi)$ holds.
\end{proof}

Similarly to the bipartite entanglement difference,
we define the \emph{bipartite entanglement rate} $e_P(\psi,\{\mathcal{R}_n\})$ with respect to the partition $\{P,P^{\mathsf{c}}\}$ of the set $[M]$ and the sequence $\{\mathcal{R}_n\}_{n\in\mathbb{N}}$
whose total entanglement rate is achievable
as follows:
\begin{equation} \label{eq:defBER}
e_P(\psi,\{\mathcal{R}_n\})
= \sum_{ i\in P } \sum_{ j\in P^{\mathsf{c}} } e_{i,j}(\psi,\{\mathcal{R}_n\}),
\end{equation}
where the segment entanglement rate $e_{i,j}$ is defined in Eq.~(\ref{eq:SER}).
To prove Lemma~\ref{lem:LBOBER}, we use the following lemma telling the continuity of the von Neumann entropy~\cite{F73,A07,W13}.

\begin{Lem}[Fannes--Audenaert Inequality~\cite{W13}] \label{lem:FAI}
Let $\rho$ and $\sigma$ be density operators in $\mathcal{D}(X)$, where $X$ is a quantum system, and suppose that $\varepsilon\coloneqq\frac{1}{2}\left\| \rho-\sigma \right\|_{1}$.
Then the inequality $|H(\rho)-H(\sigma)|\le \varepsilon\log[d_X-1] + h(\varepsilon)$ holds, where $h(\cdot)$ is the binary entropy.
\end{Lem}

\begin{proof}[Proof of Lemma~\ref{lem:LBOBER}]
We consider the $R$-assisted QSR task explained in the proof of Proposition~\ref{prop:LBBED}.
To be specific, for each $n$, we consider an $R$-assisted QSR protocol
\begin{equation}
\mathcal{A}_n\colon\mathcal{L}\left(\sbigotimes_{i=1}^M A_i^{\otimes n} V^{\otimes n} W^{\otimes n} B^{(n)}_i  \right)
\longrightarrow
\mathcal{L}\left(\sbigotimes_{i=1}^M (A'_i)^{\otimes n} V^{\otimes n} W^{\otimes n} C^{(n)}_i  \right)
\end{equation}
of the quantum state $\ket{\xi}_{AVW}^{\otimes n}$
with error $\varepsilon_n$ satisfying
\begin{equation} \label{eq:RAerror}
\left\|
\mathcal{A}_n
\left( \xi^{\otimes n} \otimes \tilde{\Psi}_n \right)
-
\xi_{\mathrm{f}}^{\otimes n} \otimes \tilde{\Phi}_n
\right\|_{1}
\le
\varepsilon_n,
\end{equation}
where quantum systems $B^{(n)}_i$ and $C^{(n)}_i$ are defined by
\begin{equation}
B^{(n)}_i = \sbigotimes_{j\in[M]\setminus\{i\}} B^{(n)}_{i,j}
\quad \mathrm{and} \quad
C^{(n)}_i = \sbigotimes_{j\in[M]\setminus\{i\}} C^{(n)}_{i,j},
\end{equation}
and $\tilde{\Psi}_n$ and $\tilde{\Phi}_n$ are explained in Sec.~\ref{subsec:QEC}.
For each $n$, let $T_n^{\mathrm{bef}}$ and $T_n^{\mathrm{aft}}$ be total amounts of entanglement
between two sets $P$ and $P^{\mathsf{c}}$ of the users before and after performing the protocol $\mathcal{A}_n$, respectively.
Since the amount of entanglement between the two sets of the users cannot increase on average via LOCC~\cite{BDSW96}, we obtain that $T_n^{\mathrm{bef}}\ge T_n^{\mathrm{aft}}$ holds for each $n$.
Note that the amounts of entanglement are represented as
\begin{equation} \label{eq:Tbefn}
T_n^{\mathrm{bef}}
=H\left(\sbigotimes_{i\in P} A_{i}^{\otimes n} V^{\otimes n} B^{(n)}_{i}\right)_{\xi^{\otimes n} \otimes \tilde{\Psi}_n}
= nH\left(\sbigotimes_{i\in P} A_{i}V\right)_{\xi}
+\sum_{ i\in P } \sum_{ j\in P^{\mathsf{c}} } \log d_{B^{(n)}_{i,j}}
\quad \mathrm{and} \quad
T_n^{\mathrm{aft}}
= H\left(\sbigotimes_{i\in P} (A'_{i-1})^{\otimes n} V^{\otimes n} C^{(n)}_{i}\right)_{\mathcal{A}_n(\xi^{\otimes n} \otimes \tilde{\Psi}_n)},
\end{equation}
where $T_n^{\mathrm{bef}}$ is obtained by using the additivity of the von Neumann entropy~\cite{W13}.
By applying the monotonicity of the trace distance~\cite{W13} to the inequality in Eq.~(\ref{eq:RAerror}), we have
\begin{equation}
\varepsilon'_n\coloneqq \frac{1}{2}\left\|
\Tr_{\otimes_{i\in P^{\mathsf{c}}}(A'_{i-1})^{\otimes n}C^{(n)}_{i}W^{\otimes n}}
\left[
\mathcal{A}_n
\left( \xi^{\otimes n} \otimes \tilde{\Psi}_n \right)
\right]-
\Tr_{\otimes_{i\in P^{\mathsf{c}}}(A'_{i-1})^{\otimes n}C^{(n)}_{i}W^{\otimes n}}
\left[
\xi_{\mathrm{f}}^{\otimes n} \otimes \tilde{\Phi}_n
\right]
\right\|_{1}
\le
\varepsilon_n.
\end{equation}
By applying Lemma~\ref{lem:FAI} to the above inequality,
we obtain the following inequalities:
\begin{eqnarray}
\left|
T_n^{\mathrm{aft}}
-
H\left(\sbigotimes_{i\in P} (A'_{i-1})^{\otimes n}V^{\otimes n}C^{(n)}_{i}\right)_{\xi_{\mathrm{f}}^{\otimes n} \otimes \tilde{\Phi}_n}
\right|
&\le&
\varepsilon'_n
\log \left(d_{\sbigotimes_{i\in P} (A'_{i-1})^{\otimes n}C^{(n)}_{i}V^{\otimes n}} -1\right)+h(\varepsilon'_n)
\le
\varepsilon'_n
\log \left(d_{\sbigotimes_{i\in P} (A'_{i-1})^{\otimes n}C^{(n)}_{i}V^{\otimes n}}\right)+h(\varepsilon'_n) \\
&\le&
\varepsilon'_n \left(n\sum_{i\in P}\log d_{A'_{i-1}} +\sum_{i\in P}\log d_{C^{(n)}_i}+ n\log d_{V} \right)
+ h(\varepsilon'_n).
\end{eqnarray}
The additivity of the von Neumann entropy~\cite{W13} implies
\begin{equation}
H\left(\sbigotimes_{i\in P} (A'_{i-1})^{\otimes n}V^{\otimes n}C^{(n)}_{i}\right)_{\xi_{\mathrm{f}}^{\otimes n} \otimes \tilde{\Phi}}
=
nH\left(\sbigotimes_{i\in P} A'_{i-1}V\right)_{\xi_{\mathrm{f}}}
+\sum_{ i\in P } \sum_{ j\in P^{\mathsf{c}} } \log d_{C^{(n)}_{i,j}}.
\end{equation}
Consequently, $T_n^{\mathrm{bef}}\ge T_n^{\mathrm{aft}}$ becomes
\begin{equation}
nH\left(\sbigotimes_{i\in P} A_{i}V\right)_{\xi}
+\sum_{ i\in P } \sum_{ j\in P^{\mathsf{c}} } \log d_{B^{(n)}_{i,j}}
\ge nH\left(\sbigotimes_{i\in P} A'_{i-1}V\right)_{\xi_{\mathrm{f}}}
+\sum_{ i\in P } \sum_{ j\in P^{\mathsf{c}} } \log d_{C^{(n)}_{i,j}}
-\varepsilon'_n \left(n\sum_{i\in P}\log d_{A'_{i-1}} +\sum_{i\in P}\log d_{C^{(n)}_i}+ n\log d_{V} \right)
- h(\varepsilon'_n).
\end{equation}
This implies that
\begin{equation}
\sum_{ i\in P } \sum_{ j\in P^{\mathsf{c}} } \frac{1}{n}\left(
\log d_{B^{(n)}_{i,j}} - \log d_{C^{(n)}_{i,j}}
\right)
\ge
H\left(\sbigotimes_{i\in P} A'_{i-1}V\right)_{\xi_{\mathrm{f}}}
-H\left(\sbigotimes_{i\in P} A_{i}V\right)_{\xi}
-\varepsilon'_n \left(\sum_{i\in P}\log d_{A'_{i-1}} +\frac{1}{n}\sum_{i\in P}\log d_{C^{(n)}_i}+ \log d_{V} \right)
- \frac{h(\varepsilon'_n)}{n},
\end{equation}
which holds for each $n$, and so we obtain that
\begin{equation}
e_P(\psi,\{\mathcal{A}_n\})
\ge
H\left(\sbigotimes_{i\in P} A'_{i-1}V\right)_{\xi_{\mathrm{f}}}
-H\left(\sbigotimes_{i\in P} A_{i}V\right)_{\xi}
=
H\left(\sbigotimes_{i\in P} A'_{i-1}V\right)_{U\ket{\psi_{\mathrm{f}}}}
-H\left(\sbigotimes_{i\in P} A_{i}V\right)_{U\ket{\psi}}
=
H\left(\sbigotimes_{i\in P} A_{i-1}V\right)_{U\ket{\psi}}
-H\left(\sbigotimes_{i\in P} A_{i}V\right)_{U\ket{\psi}}.
\end{equation}
Here, $U\ket{\psi_{\mathrm{f}}}$ and $U\ket{\psi}$ are abbreviations for $\mathds{1}_{A}\otimes U\ket{\psi_{\mathrm{f}}}$ and $\mathds{1}_{A}\otimes U\ket{\psi}$,
respectively, and the last equality comes from Eq.~(\ref{eq:inifinsame}).
Thus, we have $e_P(\psi,\{\mathcal{R}_n\})\ge l_P(\psi)$,
since the quantum system $V,W$ and the isometry $U$ are arbitrary, and any sequence of QSR protocols is also a sequence of $R$-assisted QSR protocols.
\end{proof}

From Proposition~\ref{prop:LBBED} and Lemma~\ref{lem:LBOBER},
we know that the lower bound $l_P$ of the exact scenario is also a lower bound of the asymptotic scenario.
In other words,
we can easily obtain a lower bound of the bipartite entanglement rate
by merely finding that of the bipartite entanglement difference in the exact scenario.
Note that
it is possible to apply this technique
to other quantum communication tasks,
such as the generalized quantum Slepian-Wolf~\cite{AJW18} and
the multi-party state merging~\cite{DH10},
in which users perform the tasks via entanglement-assisted LOCC
in the asymptotic scenario.

We remark that while the lower bound
in Proposition~\ref{prop:LBBED} is presented in terms of the von Neumann entropy,
this lower bound can be generalized
by replacing the von Neumann entropy with the R\'enyi entropies~\cite{DH02} under the exact scenario,
as in the one-shot quantum state exchange~\cite{LYAL19}.

\section{Proof of Theorem~\ref{thm:lb}} \label{app:thm:lb}

Let $r$ be any achievable total entanglement rate for the initial state $\psi$.
Then there is a sequence $\{\mathcal{R}_n\}_{n\in\mathbb{N}}$
of QSR protocols $\mathcal{R}_n$ of $\psi^{\otimes n}$
with error $\varepsilon_n$ such that
$e_{i,j}(\psi,\{\mathcal{R}_n\})$ converges for any $i,j$,
$e_{\mathrm{tot}}(\psi,\{\mathcal{R}_n\})=r$,
and $\lim_{n\to\infty} \varepsilon_n = 0$.
Since $l_{P_k}=l_{P_{M-k}}$ holds for any $k\in[M-1]$,
we have $l_k(\psi)=l_{M-k}(\psi)$.
So we will prove in the following that $l_k(\psi)$ is a lower bound on the OEC
for $1\le k \le \lfloor M/2 \rfloor$.

For a non-empty proper subset $P$ of the set $[M]$,
we defined a function $f_P\colon[M]\times [M]\to\{0,1\}$ as follows:
\begin{equation}
f_P(i,j)=
\begin{cases}
1 &\text{if ($i\in P$, $j\in P^{\mathsf{c}}$) or ($j\in P$, $i\in P^{\mathsf{c}}$)} \\
0 &\text{otherwise}.
\end{cases}
\end{equation}
Note that $f_P(j,i)=f_P(i,j)$ holds for each $i,j$, and
the bipartite entanglement rate $e_P(\psi,\{\mathcal{R}_n\})$
is represented as
\begin{equation} \label{eq:newform}
e_P(\psi,\{\mathcal{R}_n\})
= \sum_{\substack{ i,j\in[M] \\ i< j }} f_P(i,j)e_{i,j}(\psi,\{\mathcal{R}_n\}).
\end{equation}
For given elements $i,j$,
let $S_k^{ij}$ be the subset of the set $S_k$
whose elements $P_k$ satisfy $f_{P_k}(i,j)=1$.
Then the size of the set $S_k^{ij}$ is $n_k\coloneqq2{M-2 \choose k-1}$.
Observe that
$|S_k^{ij}|=|S_k^{i'j'}|$
holds for any elements $i$, $j$, $i'$ and $j'$.
This means that
for a given segment entanglement rate $e_{i,j}$
there exist $n_k$ subsets $P_k$ of the set $[M]$
such that $f_{P_k}(i,j)=1$,
i.e.,
\begin{equation} \label{eq:DoublSum}
\sum_{P_k\in S_k} f_{P_k}(i,j) e_{i,j}(\psi,\{\mathcal{R}_n\})
=
n_k e_{i,j}(\psi,\{\mathcal{R}_n\}).
\end{equation}
From Eqs.~(\ref{eq:newform}) and~(\ref{eq:DoublSum}),
it follows that
\begin{eqnarray}
e_{\mathrm{tot}}(\psi,\{\mathcal{R}_n\})
&=&
\sum_{\substack{ i,j\in[M] \\ i< j }} e_{i,j}(\psi,\{\mathcal{R}_n\})
=
\frac{1}{n_k}
\sum_{\substack{ i,j\in[M] \\ i< j }} \sum_{P_k\in S_k} f_{P_k}(i,j) e_{i,j}(\psi,\{\mathcal{R}_n\}) \\
&=&
\frac{1}{n_k}
\sum_{P_k\in S_k}\sum_{\substack{ i,j\in[M] \\ i< j }}  f_{P_k}(i,j) e_{i,j}(\psi,\{\mathcal{R}_n\})
=
\frac{1}{n_k}
\sum_{P_k\in S_k}e_{P_k}(\psi,\{\mathcal{R}_n\})
\ge
l_k(\psi). \label{eq:Goodequation}
\end{eqnarray}
Here, the last inequality comes from Eq.~(\ref{eq:defBER}) and Lemma~\ref{lem:LBOBER}.
This shows that $r\ge l_k(\psi)$ holds for any achievable total entanglement rate $r$
and any $k$.

\section{Proof of Lemma~\ref{lem:MPMM}} \label{app:PoLMPMM}

Let $\psi_0=\psi$, and for each $i\in[M-1]$,
we define quantum states $\psi_i$ for the quantum state merging tasks
as 
\begin{equation}
\psi_i
=
\left(
\sbigotimes_{j=1}^{i} \mathrm{id}_{\mathcal{L}(A_j)\to \mathcal{L}(A'_j)}
\otimes
\sbigotimes_{j=i+1}^{M} \mathrm{id}_{\mathcal{L}(A_j)}
\otimes
\mathrm{id}_{\mathcal{L}(E)}
\right)
(\psi).
\end{equation}
Note that, for each $i\in[M-1]$, $\psi_i$ is a pure quantum state on the quantum systems
\begin{equation}
\sbigotimes_{j=1}^{i} A'_j \otimes \sbigotimes_{j=i+1}^{M} A_j \otimes E.
\end{equation}
For each $i\in[M-1]$, the $i^{\mathrm{th}}$ user and the $(i+1)^{\mathrm{th}}$ user transform the quantum state $\psi_{i-1}$ into the quantum state $\psi_i$, by means of LOCC and shared entanglement.
To be specific, the quantum state on the quantum system $A_i$ of the $i^{\mathrm{th}}$ user is asymptotically merged to the $(i+1)^{\mathrm{th}}$ user's quantum system $A'_i$ by using the $(i+1)^{\mathrm{th}}$ user's quantum system $A_{i+1}$ as quantum side information. So, in this case, the rest quantum systems of the quantum state $\psi_{i-1}$,
\begin{equation}
E_i
\coloneqq
\sbigotimes_{j\in[M]\setminus[i+1]} A_j \otimes \sbigotimes_{j\in[i-1]} A'_j \otimes E,
\end{equation}
are considered as the parts of the environment system.
From the definition of the OEC of the quantum state merging~\cite{HOW05,HOW06},
for each $i\in[M-1]$, there is a sequence $\{\mathcal{M}_n^{(i)}\}_{n\in\mathbb{N}}$
of LOCC 
\begin{equation}
\mathcal{M}_n^{(i)}
\colon
\mathcal{L}\left(
A_i^{\otimes n} \otimes B^{(n)}_{i,i+1} \otimes A_{i+1}^{\otimes n}\otimes B^{(n)}_{i+1,i} \right)
\longrightarrow
\mathcal{L}\left(
{A'}_{i}^{\otimes n} \otimes C^{(n)}_{i,i+1} \otimes A_{i+1}^{\otimes n} \otimes C^{(n)}_{i+1,i} \right)
\end{equation}
of $\psi_{i-1}^{\otimes n}$ with error $\varepsilon_n^{(i)}$
which merges the part $A_i$ from the $i^{\mathrm{th}}$ user to the $(i+1)^{\mathrm{th}}$ user
and satisfies $\lim_{n\to\infty} \varepsilon_n^{(i)} = 0$,
\begin{eqnarray}
\left\|
\left( \mathcal{M}_n^{(i)} \otimes \mathrm{id}_{\mathcal{L}(E_i^{\otimes n})} \right)
\left( \psi_{i-1}^{\otimes n} \otimes \Psi_{n}^{(i)} \right)
-
\psi_{i}^{\otimes n} \otimes \Phi_{n}^{(i)}
\right\|_{1}
&\le&
\varepsilon_n^{(i)}, \label{eq:errern} \\
\lim_{n\to\infty}
\frac{1}{n}
\left(\log d_{B^{(n)}_{i,i+1}} - \log d_{C^{(n)}_{i,i+1}}
\right)
&=&H(A_i|A_{i+1}), 
\end{eqnarray}
where $\Psi_{n}^{(i)}$ and $\Phi_{n}^{(i)}$ are
pure maximally entangled states
on quantum systems $B^{(n)}_{i,i+1}B^{(n)}_{i+1,i}$ and $C^{(n)}_{i,i+1}C^{(n)}_{i+1,i}$
shared by the $i^{\mathrm{th}}$ user and the $(i+1)^{\mathrm{th}}$ user
with Schmidt rank $d_{B^{(n)}_{i,i+1}}$ and $d_{C^{(n)}_{i,i+1}}$,
respectively.
In addition,
from the Schumacher compression~\cite{S95,W13} together with 
the quantum teleportation~\cite{BBCJPW93},
there exists a sequence $\{\mathcal{S}_n\}_{n\in\mathbb{N}}$
of LOCC 
\begin{equation}
\mathcal{S}_n
\colon
\mathcal{L} \left(
A_M^{\otimes n} \otimes B^{(n)}_{M,1} \otimes B^{(n)}_{1,M} \right)
\longrightarrow
\mathcal{L} \left(
{A'}_{M}^{\otimes n}\otimes C^{(n)}_{M,1} \otimes C^{(n)}_{1,M} \right)
\end{equation}
of $\psi_{M-1}^{\otimes n}$ with error $\varepsilon_n^{(M)}$,
which transfers the part $A_M$ from the $M^{\mathrm{th}}$ user to the $1^{\mathrm{st}}$ user
and satisfies $\lim_{n\to\infty} \varepsilon_n^{(M)} = 0$,
\begin{eqnarray}
\left\|
\left( \mathcal{S}_n \otimes \mathrm{id}_{\mathcal{L}(E_M^{\otimes n})} \right)
\left( \psi_{M-1}^{\otimes n} \otimes \Psi_n^{(M)} \right)
-
\psi_{\mathrm{f}}^{\otimes n} \otimes \Phi_n^{(M)}
\right\|_{1}
&\le&
\varepsilon_n^{(M)}, \label{eq:errerM} \\
\lim_{n\to\infty}
\frac{1}{n}
\left(\log d_{B^{(n)}_{M,1}} - \log d_{C^{(n)}_{M,1}}
\right)
&=&H(A_M),
\end{eqnarray}
where $E_M = \sbigotimes_{j=1}^{M-1} A'_j \otimes E$,
and
$\Psi_n^{(M)}$ and $\Phi_n^{(M)}$ are
pure maximally entangled states
on quantum systems $B^{(n)}_{M,1}B^{(n)}_{1,M}$ and $C^{(n)}_{M,1}C^{(n)}_{1,M}$
shared by the $1^{\mathrm{st}}$ user and the $M^{\mathrm{th}}$ user
with Schmidt rank $d_{B^{(n)}_{M,1}}$ and $d_{C^{(n)}_{M,1}}$,
respectively.
For each $n\in\mathbb{N}$,
we define LOCC 
$\mathcal{R}_n$ 
as
\begin{equation}
\mathcal{R}_n
=\mathcal{S}_n \circ \mathcal{M}_n^{(M-1)} \circ\mathcal{M}_n^{(M-2)} \circ
\cdots \circ \mathcal{M}_n^{(1)}.
\end{equation}
We also define quantum states $\tilde{\Psi}_n$, $\tilde{\Phi}_n$, and $\tilde{\Omega}_n^{(i)}$ for each $i\in[M]$ as
\begin{equation}
\tilde{\Psi}_n = \sbigotimes_{i\in[M]} \Psi_n^{(i)},
\quad
\tilde{\Phi}_n = \sbigotimes_{i\in[M]} \Phi_n^{(i)},
\quad \mathrm{and} \quad
\tilde{\Omega}_n^{(i)}
=\sbigotimes_{j=i+1}^{M} \Psi_n^{(j)}
\otimes \sbigotimes_{j=1}^{i-1} \Phi_n^{(j)}.
\end{equation}
Observe that, for $i=2,\ldots,M-1$, the inequalities
\begin{eqnarray} \label{eq:triangle}
&&
\left\|
\left( \mathcal{M}_n^{(i)} \circ \cdots \circ \mathcal{M}_n^{(1)} \right)
\left( \psi^{\otimes n} \otimes \tilde{\Psi}_n \right)
-
\psi_{i}^{\otimes n} \otimes \Psi_n^{(i+1)}\otimes\tilde{\Omega}_n^{(i+1)}
\right\|_{1} \nonumber \\
&&\le
\left\|
\left( \mathcal{M}_n^{(i)} \circ \cdots \circ \mathcal{M}_n^{(1)} \right)
\left( \psi^{\otimes n} \otimes \tilde{\Psi}_n \right)
-
\left( \mathcal{M}_n^{(i)} \otimes \mathrm{id}_{\mathcal{L}(E_{i}^{\otimes n})} \right)
\left( \psi_{i-1}^{\otimes n} \otimes \Psi_n^{(i)} \right)
\otimes\tilde{\Omega}_n^{(i)}
\right\|_{1} \\
&&\quad
+
\left\|
\left( \mathcal{M}_n^{(i)} \otimes \mathrm{id}_{\mathcal{L}(E_{i}^{\otimes n})} \right)
\left( \psi_{i-1}^{\otimes n} \otimes \Psi_n^{(i)} \right)
\otimes\tilde{\Omega}_n^{(i)}
-
\psi_{i}^{\otimes n} \otimes \Psi_n^{(i+1)}
\otimes\tilde{\Omega}_n^{(i+1)}
\right\|_{1} \nonumber \\
&&\le
\left\|
\left( \mathcal{M}_n^{(i-1)} \circ \cdots \circ \mathcal{M}_n^{(1)} \right)
\left( \psi^{\otimes n} \otimes \tilde{\Psi}_n \right)
-
\psi_{i-1}^{\otimes n} \otimes \Psi_n^{(i)}
\otimes\tilde{\Omega}_n^{(i)}
\right\|_{1}
+
\left\|
\left( \mathcal{M}_n^{(i)} \otimes \mathrm{id}_{\mathcal{L}(E_{i}^{\otimes n})} \right)
\left( \psi_{i-1}^{\otimes n} \otimes \Psi_n^{(i)} \right)
-
\psi_{i}^{\otimes n} \otimes \Phi_n^{(i)}
\right\|_{1}
\end{eqnarray}
hold,
where the first inequality and the second inequality come from the triangle property
and the monotonicity of the trace distance~\cite{W13},
and other identity maps $\mathrm{id}_{E^{\otimes n}}$,
and $\mathrm{id}_{E_i^{\otimes n}}$ are omitted for convenience.
Then we have
\begin{eqnarray}
&&
\left\|
\left( \mathcal{R}_n \otimes \mathrm{id}_{E^{\otimes n}} \right)
\left( \psi^{\otimes n} \otimes \tilde{\Psi}_n \right)
-
\psi_{\mathrm{f}}^{\otimes n} \otimes \tilde{\Phi}_n
\right\|_{1} \nonumber \\
&&\le
\left\|
\left( \mathcal{R}_n \otimes \mathrm{id}_{E^{\otimes n}} \right)
\left( \psi^{\otimes n} \otimes \tilde{\Psi}_n \right)
-
\left( \mathcal{S}_n \otimes \mathrm{id}_{\mathcal{L}(E_M^{\otimes n})} \right)
\left( \psi_{M-1}^{\otimes n} \otimes \Psi_n^{(M)} \right)
\otimes\tilde{\Omega}_n^{(M)}
\right\|_{1} \\
&&\quad+
\left\|
\left( \mathcal{S}_n \otimes \mathrm{id}_{\mathcal{L}(E_M^{\otimes n})} \right)
\left( \psi_{M-1}^{\otimes n} \otimes \Psi_n^{(M)} \right)
\otimes\tilde{\Omega}_n^{(M)}
-
\psi_{\mathrm{f}}^{\otimes n} \otimes \tilde{\Phi}_n
\right\|_{1} \nonumber \\
&&\le
\left\|
\left( \mathcal{M}_n^{(M-1)} \circ \cdots \circ \mathcal{M}_n^{(1)} \right)
\left( \psi^{\otimes n} \otimes \tilde{\Psi}_n \right)
-
\psi_{M-1}^{\otimes n} \otimes \Psi_n^{(M)}
\otimes\tilde{\Omega}_n^{(M)}
\right\|_{1}
+
\left\|
\left( \mathcal{S}_n \otimes \mathrm{id}_{\mathcal{L}(E_M^{\otimes n})} \right)
\left( \psi_{M-1}^{\otimes n} \otimes \Psi_n^{(M)} \right)
-
\psi_{\mathrm{f}}^{\otimes n} \otimes \Phi_n^{(M)}
\right\|_{1} \\
&&\le
\left\|
\mathcal{M}_n^{(1)}
\left( \psi^{\otimes n} \otimes \tilde{\Psi}_n \right)
-
\psi_1^{\otimes n} \otimes \Psi_n^{(2)}
\otimes\tilde{\Omega}_n^{(2)}
\right\|_{1}
+
\sum_{i=2}^{M-1}
\left\|
\left( \mathcal{M}_n^{(i)} \otimes \mathrm{id}_{\mathcal{L}(E_{i}^{\otimes n})} \right)
\left( \psi_{i-1}^{\otimes n} \otimes \Psi_n^{(i)} \right)
-
\psi_{i}^{\otimes n} \otimes \Phi_n^{(i)}
\right\|_{1} \\
&&\quad
+
\left\|
\left( \mathcal{S}_n \otimes \mathrm{id}_{\mathcal{L}(E_M^{\otimes n})} \right)
\left( \psi_{M-1}^{\otimes n} \otimes \Psi_n^{(M)} \right)
-
\psi_{\mathrm{f}}^{\otimes n} \otimes \Phi_n^{(M)}
\right\|_{1} \nonumber \\
&&=
\sum_{i=1}^{M-1}
\left\|
\left( \mathcal{M}_n^{(i)} \otimes \mathrm{id}_{\mathcal{L}(E_{i}^{\otimes n})} \right)
\left( \psi_{i-1}^{\otimes n} \otimes \Psi_n^{(i)} \right)
-
\psi_{i}^{\otimes n} \otimes \Phi_n^{(i)}
\right\|_{1}
+
\left\|
\left( \mathcal{S}_n \otimes \mathrm{id}_{\mathcal{L}(E_M^{\otimes n})} \right)
\left( \psi_{M-1}^{\otimes n} \otimes \Psi_n^{(M)} \right)
-
\psi_{\mathrm{f}}^{\otimes n} \otimes \Phi_n^{(M)}
\right\|_{1}
\le
\sum_{i=1}^{M} \varepsilon_n^{(i)}.
\end{eqnarray}
Here, the first inequality and the second inequality hold
from the triangle property and the monotonicity of the trace distance again.
The third inequality is obtained 
by repeatedly applying the inequality in Eq.~(\ref{eq:triangle}).
Since $\psi=\psi_0$,
the last equality holds.
The last inequality comes from Eqs.~(\ref{eq:errern}) and~(\ref{eq:errerM}).
Set $\varepsilon_n=\sum_{i=1}^{M} \varepsilon_n^{(i)}$.
Then $\lim_{n\to\infty} \varepsilon_n = 0$,
since $\lim_{n\to\infty} \varepsilon_n^{(i)} = 0$ holds for each $i\in[M]$.
It follows that
there is a sequence $\{\mathcal{R}_n\}_{n\in\mathbb{N}}$
of QSR protocols $\mathcal{R}_n$ of $\ket{\psi}^{\otimes n}$
with error $\varepsilon_n$ such that $\lim_{n\to\infty} \varepsilon_n = 0$,
\begin{eqnarray}
e_{i,j}(\psi,\{\mathcal{R}_n\})
&=&
\begin{cases}
H(A_i|A_{i+1}) &\text{if $i\in[M-1]$ and $j=i+1$} \\
H(A_M) &\text{if $i=M$ and $j=1$} \\
0 &\text{otherwise},
\end{cases} \\
e_{\mathrm{tot}}(\psi,\{\mathcal{R}_n\})
&=&
H(A_M)+\sum_{i=1}^{M-1}H(A_i|A_{i+1}).
\end{eqnarray} 

\section{Proof of Theorem~\ref{thm:SCOFC}} \label{app:PoTsc}

To prove Theorem~\ref{thm:SCOFC},
we use the following lemma.

\begin{Lem} \label{lem:represent}
The lower bound $l_1(\psi)$ shown in Theorem~\ref{thm:lb} is lower bounded by
\begin{equation}
l_1(\psi)
\ge
\frac{1}{2} \max_{D\subseteq[M]}
\left|
\sum_{i_j\in D} (-1)^j H(E|A_{i_j})_{\psi}
\right|,
\end{equation}
where $D$ denotes a subset $\{{i_1},{i_2},\ldots,{i_{2k}}\}$
of the set $[M]$
with $k = 1,\ldots,\lfloor M/2 \rfloor$ and $i_1<i_2<\cdots<i_{2k}$,
and the maximum is taken over all possible subsets $D$ whose sizes are even.
\end{Lem}

\begin{proof}
It is easy to check that
$l_1(\psi)$ is lower bounded by
\begin{equation}
\frac{1}{2}
\sum_{i=1}^{M}
\max 
\left\{
H(A_{i-1})_{\psi}-H(A_{i })_{\psi},
H(A_{i-1}E)_{\psi}-H(A_{i }E)_{\psi}
\right\},
\end{equation}
by using the definition of the lower bound $l_i(\psi)$ in Eq.~(\ref{eq:defLK}).
So it suffices to show the equality
$\mathrm{LHS}=\mathrm{RHS}$,
where $\mathrm{LHS}$ and $\mathrm{RHS}$ are defined as
\begin{equation}
\mathrm{LHS}
=
\sum_{i=1}^{M}
\max 
\left\{
\alpha_i,
\beta_i
\right\}
\quad \mathrm{and} \quad
\mathrm{RHS}
=
\max_{D\subseteq[M]}
\left|
\sum_{i_j\in D}(-1)^j H(E|A_{i_j})_{\psi}
\right|,
\end{equation}
with $\alpha_i=H(A_{i-1})_{\psi}-H(A_{i })_{\psi}$
and $\beta_i=H(A_{i-1}E)_{\psi}-H(A_{i }E)_{\psi}$.

(i) To show $\mathrm{LHS} \le \mathrm{RHS}$,
we use functions $f_i\colon\{0,1\}\to\mathbb{R}$ defined as
$f_i(x) = (1-x) \alpha_i + x \beta_i$.
Let $\bold{b}$ be an $M$-bit string $\bold{b}=b_1b_2\cdots b_M$ such that 
$b_i\in\{0,1\}$ for each $i$.
Then LHS is represented as
\begin{equation}
\mathrm{LHS}
=\max_{\bold{b}}\sum_{i=1}^M f_i(b_i),
\end{equation}
where the maximum is taken over all $M$-bit strings.
In addition,
we observe that the equalities
\begin{eqnarray} \label{eq:relation}
\sum_{i=1}^M f_i(b_i)
&=&
\sum_{i=1}^M
\left[ (1-b_i)\alpha_i + b_i\beta_i \right]
=
\sum_{i=1}^M
b_i \left(\beta_i -\alpha_i\right)
=
\sum_{i=1}^M (b_i-1)\beta_i - \sum_{i=1}^M b_i\alpha_i
=
-\sum_{i=1}^M
\left[
(1-(1-b_i))\alpha_i + (1-b_i)\beta_i
\right] \\
&=&
-\sum_{i=1}^M f_i(1-b_i)
\end{eqnarray}
hold
for any $M$-bit string $\bold{b}$,
where
the second equality and the third equality come from
equalities
$\sum_{i=1}^M \alpha_i=\sum_{i=1}^M \beta_i=0$.
This implies
\begin{equation}
\mathrm{LHS}
=\max_{\bold{b}}
\left| \sum_{i=1}^M f_i(b_i) \right|,
\end{equation}
where the maximum is taken over all $M$-bit strings
having $k$ zero bits with $1\le k \le \lfloor M/2 \rfloor$.
For any $M$-bit string $\bold{b}=b_1b_2\cdots b_M$ with $k$ bits in state zero,
we can express $k$ zero bits and the other bits in state one 
using two functions $g_\bold{b}\colon[k]\to[M]$ and $h_\bold{b}\colon[M-k]\to[M]$
satisfying $b_{g_\bold{b}(i)}=0$ and $b_{h_\bold{b}(i)}=1$,
respectively.
Observe that
\begin{eqnarray}
\sum_{i=1}^M f_i(b_i)
&=&
\sum_{i=1}^k f_{g_\bold{b}(i)}(0) + \sum_{i=1}^{M-k} f_{h_\bold{b}(i)}(1)
=
\sum_{i=1}^k f_{g_\bold{b}(i)}(0) - \sum_{i=1}^k f_{g_\bold{b}(i)}(1)
=
\sum_{i=1}^k \left[
H(E|A_{g_\bold{b}(i) })_{\psi}-H(E|A_{g_\bold{b}(i)-1})_{\psi}
\right] \\
&=&
\sum_{i\in X\setminus Y} H(E|A_{i })_{\psi}
-\sum_{i\in Y\setminus X} H(E|A_{i})_{\psi},
\end{eqnarray}
where $X=\{g_\bold{b}(i):i\in[k]\}$ and $Y=\{g_\bold{b}(i)-1:i\in[k]\}$.
The second equality comes from the simple fact
\begin{equation}
\sum_{i=1}^k f_{g_\bold{b}(i)}(1) + \sum_{i=1}^{M-k} f_{h_\bold{b}(i)}(1)
=
\sum_{i=1}^M f_i(1)
=0.
\end{equation}
Since $1\le k \le \lfloor M/2 \rfloor$, the set $X$ is non-empty.
Let $l_X$ be the largest element of the set $X$.
Then $l_X\notin Y$, by the definition of the set $Y$, and so $X\setminus Y$ is non-empty.
Assume that $|X|=|Y|=s>0$ and $|X\setminus Y|=|X\setminus Y|=t>0$ for some natural numbers $s$ and $t$ with $t\le s$.
Then we can represent the sets $X$, $Y$, $X\setminus Y$, and $X\setminus Y$ as
\begin{equation}
X=\{x_1,x_2,\ldots,x_s\},\quad
Y=\{y_1,y_2,\ldots,y_s\},\quad
X\setminus Y= \{a_1,a_2,\ldots,a_t \},\quad
Y\setminus X= \{b_1,b_2,\ldots,b_t \},
\end{equation}
where $x_i<x_j$ and $y_i<y_j$ for each $i,j\in[s]$ with $i<j$,
and $a_k<a_l$ and $b_k<b_l$ for each $k,l\in[t]$ with $k<l$.

For each $i\in[s-1]$, we consider two consecutive elements $x_i$ and $x_{i+1}$ of the set $X$.
By the definition of the set $Y$, $x_i-1\in Y$ and $x_{i+1}-1\in Y$.
If $x_i+1=x_{i+1}$, then $x_i=x_{i+1}-1\in Y$, and so $x_i\notin X\setminus Y$.
Conversely, if $x_i\notin X\setminus Y$, then $x_i\in Y$.
By the definition of the set $Y$, $x_i+1\in X$.
Since $x_i<x_{i+1}$, we have $x_{i+1}=x_i+1$.
Thus, we obtain that, for each $i\in[s-1]$,
\begin{equation} \label{eq:EQUIV1}
x_i+1=x_{i+1}
\quad
\text{if and only if}
\quad
x_i\notin X\setminus Y.
\end{equation}
Similarly to the above equivalence, we also obtain that, for each $i\in[s-1]$,
\begin{equation} \label{eq:EQUIV2}
y_i+1=y_{i+1}
\quad
\text{if and only if}
\quad
y_{i+1}\notin Y\setminus X.
\end{equation}
Note that $x_1\le a_1$ holds in general, and equalities $x_1-1=y_1=b_1$ also hold, by the definition of the set $Y$.
Thus, $b_1<a_1$.

For the case that $x_1=a_1$, we have $a_1=b_1+1$. If $b_2=b_1+1$, then $a_1+1=b_2+1\in X$, by the definition of the set $Y$.
From Eq.~(\ref{eq:EQUIV1}), $a_1\in X$ and $a_1+1\in X$ means $a_1\notin X\setminus Y$, which contradicts to $a_1\in X\setminus Y$.
Thus, $a_1=b_1+1<b_2$.
For the case that $x_1<a_1$, we have $x_1\notin X\setminus Y$, since $a_1$ is the smallest element of $X\setminus Y$.
From Eq.~(\ref{eq:EQUIV1}), this means that $x_1+1$ is an element of the set $X$. 
If $x_1+1<a_1$, then $x_1\notin X\setminus Y$, since $a_1$ is the smallest element of $X\setminus Y$, and Eq.~(\ref{eq:EQUIV1}) implies that $x_1+2$ is an element of the set $X$.
In this way, we find a subset $\{x_1,x_1+1,\ldots,a_1\}$ of the set $X$, and so a set $\{x_1-1,x_1,\ldots,a_1-1\}$ is a subset of the set $Y$, by the definition of the set $Y$.
From Eq.~(\ref{eq:EQUIV2}), we obtain $a_1\le b_2$.
In addition, since $a_1$ is the element of the set $X\setminus Y$, $b_2$ can not be equal to $a_1$.
Thus, $a_1<b_2$.

If $a_2\le b_2$, then $a_2-1\in Y$, by the definition of the set $Y$.
Since $b_2$ is the second smallest element of the set $Y\setminus X$, $a_2-1\notin Y\setminus X$.
From Eq.~(\ref{eq:EQUIV2}), $a_2-1\notin Y\setminus X$ implies $a_2-2\in Y$.
If $b_1<a_2-2$, then $a_2-2\notin Y\setminus X$, since $b_2$ is the second smallest element of the set $Y\setminus X$, and so $a_2-3\in Y$ from Eq.~(\ref{eq:EQUIV2}).
In this way, we find a subset $\{b_1,b_1+1,\ldots,a_2-1\}$ of the set $Y$, and so we obtain that  $a_1\in\{b_1+1,b_1+2,\ldots,a_2\}\subset X$, by the definition of the set $Y$.
From Eq.~(\ref{eq:EQUIV1}), $a_1\notin X\setminus Y$.
In addition, since $a_1$ is the element of the set $X\setminus Y$, $b_2$ can not be equal to $a_1$.
Thus, $a_1<b_2$, which is a contradiction. Thus, $b_2<a_2$.

Consequently, we have shown that $b_1<a_1<b_2<a_2$. By repeatedly applying the above process,
we obtain that $b_i<a_i<b_{i+1}<a_{i+1}$ for each $i\in[t-1]$.
This shows that 
there is a subset $D = \{{i_1},{i_2},\ldots,{i_{2k}}\}$ of $[M]$
with $k \in \{1,\ldots,\lfloor M/2 \rfloor\}$
such that
$i_1<i_2<\cdots<i_{2k}$,
for each $j\in[k]$,
$i_{2j-1}\in Y\setminus X$ and $i_{2j}\in X\setminus Y$,
\begin{equation}
\sum_{i=1}^M f_i(b_i)
=\sum_{i_j\in D}(-1)^j H(E|A_{i_j})_{\psi}.
\end{equation}
Thus, $\mathrm{LHS} \le \mathrm{RHS}$ holds,
since the $M$-bit string $\bold{b}$ with $k$ zero bits is arbitrary.

(ii) We show $\mathrm{LHS} \ge \mathrm{RHS}$.
Let $D = \{{i_1},{i_2},\ldots,{i_{2k}}\}$ be a subset of $[M]$ with $k \in \{1,\ldots,\lfloor M/2 \rfloor\}$ and $i_1<i_2<\cdots<i_{2k}$
satisfying
$
\mathrm{RHS}
=
\left|
\sum_{i_j\in D}(-1)^j H(E|A_{i_j})_{\psi}
\right|$.
Set an $M$-bit string $\bold{b}$ as follows:
\begin{equation} \label{eq:Defbj}
b_j=
\begin{cases}
1 &\text{if $j\in\{ i_1,i_3,\ldots,i_{2k-1}\}$} \\
0 &\text{if $j\in\{ i_2,i_4,\ldots,i_{2k  }\}$} \\
b_{j+1} &\text{otherwise},
\end{cases}
\end{equation}
where $b_M$ is defined as $b_1$ when $M\notin D$.
We obtain the following equalities:
\begin{eqnarray}
\sum_{j=1}^M f_j(b_j)
&=&
 \sum_{j=1}^{i_1} f_{j}(1)
+\sum_{j=1}^k \sum_{l=i_{(2j-1)}+1}^{i_{2j}} f_{l}(0)
+\sum_{j=1}^{k-1} \sum_{l=i_{2j}+1}^{i_{(2j+1)}} f_{l}(1)
+\sum_{j=i_{2k}+1}^{M} f_{j}(1) \\
&=&
 \sum_{j=1}^{i_1} \beta_{j}
+\sum_{j=1}^k \sum_{l=i_{(2j-1)}+1}^{i_{2j}} \alpha_{l}
+\sum_{j=1}^{k-1} \sum_{l=i_{2j}+1}^{i_{(2j+1)}} \beta_{l}
+\sum_{j=i_{2k}+1}^{M} \beta_{j} \\
&=&
\left(
 \sum_{j=i_{2k}+1}^{M} \beta_{j}
+\sum_{j=1}^{i_1} \beta_{j} \right)
+\sum_{j=1}^k \sum_{l=i_{(2j-1)}+1}^{i_{2j}} \alpha_{l}
+\sum_{j=1}^{k-1} \sum_{l=i_{2j}+1}^{i_{(2j+1)}} \beta_{l} \\
&=&
\left(
 H(A_{i_{2k}}E)_\psi-H(A_{i_1}E)_\psi \right)
+\sum_{j=1}^k \left[ H(A_{i_{(2j-1)}})_\psi-H(A_{i_{2j}})_\psi \right]
+\sum_{j=1}^{k-1} \left[ H(A_{i_{2j}}E)_\psi-H(A_{i_{(2j+1)}}E)_\psi \right] \\
&=&
 H(A_{i_{2k}}E)_\psi-H(A_{i_1}E)_\psi
+H(A_{i_{(2k-1)}})_\psi-H(A_{i_{2k}})_\psi \\
&& +\sum_{j=1}^{k-1} \left[ H(A_{i_{(2j-1)}})_\psi-H(A_{i_{2j}})_\psi \right]
+\sum_{j=1}^{k-1} \left[ H(A_{i_{2j}}E)_\psi-H(A_{i_{(2j+1)}}E)_\psi \right] \nonumber \\
&=&
 H(E|A_{i_{2k}})_\psi-H(A_{i_1}E)_\psi
+H(A_{i_{(2k-1)}})_\psi +\sum_{j=1}^{k-1} H(A_{i_{(2j-1)}})_\psi
+\sum_{j=1}^{k-1} H(E|A_{i_{2j}})_\psi
-\sum_{j=1}^{k-1} H(A_{i_{(2j+1)}}E)_\psi  \\
&=&
\left( H(E|A_{i_{2k}})_\psi +\sum_{j=1}^{k-1} H(E|A_{i_{2j}})_\psi \right)
+\left(
H(A_{i_{(2k-1)}})_\psi +\sum_{j=1}^{k-1} H(A_{i_{(2j-1)}})_\psi
-H(A_{i_1}E)_\psi-\sum_{j=2}^{k} H(A_{i_{(2j-1)}}E)_\psi
\right) \\
&=&
\sum_{j=1}^{k} H(E|A_{i_{2j}})_\psi-\sum_{j=1}^{k} H(E|A_{i_{(2j-1)}})_\psi \\
&=&\sum_{j = 1}^{2k}(-1)^j H(E|A_{i_j})_{\psi},
\end{eqnarray}
where the fourth equality comes from the fact that
\begin{equation}
\sum_{i=n}^m \beta_i=H(A_{n-1}E)_\psi-H(A_{m}E)_\psi,
\quad \sum_{i=n}^m \alpha_i=H(A_{n-1})_\psi-H(A_{m})_\psi.
\end{equation}
In the case that the sum $\sum_{i_j\in D}(-1)^j H(E|A_{i_j})_{\psi}$ is negative,
we can find another $M$-bit string $\bold{b}'$ satisfying
$b'_i=1-b_i$, 
where $b_i$ is the $i^{\mathrm{th}}$ bit of the $M$-bit string $\bold{b}$
defined in Eq.~(\ref{eq:Defbj}).
By using the relation in Eq.~(\ref{eq:relation}),
we obtain
\begin{equation}
\sum_{j=1}^M f_j(b'_j)
=
-\sum_{i_j\in D}(-1)^j H(E|A_{i_j})_{\psi}.
\end{equation}
It follows that $\mathrm{LHS} \ge \mathrm{RHS}$.
\end{proof}

\begin{proof}[Proof of Theorem~\ref{thm:SCOFC}]
We prove the contrapositive of Theorem~\ref{thm:SCOFC}.
Suppose that $e_{\mathrm{opt}}(\psi)=0$.
Then $l_1(\psi)=0$,
since the lower bound $l_i$ on the OEC is non-negative.
So Lemma~\ref{lem:represent} implies 
\begin{equation}
\max_{D\subseteq[M]}
\left|
\sum_{i_j\in D}(-1)^j H(E|A_{i_j})_{\psi}
\right|=0,
\end{equation}
where $D$ is a subset $\{{i_1},{i_2},\ldots,{i_{2k}}\}$
of the set $[M]$
with $k = 1,\ldots,\lfloor M/2 \rfloor$ and $i_1<i_2<\cdots<i_{2k}$.
By choosing $D$ as a set $\{i,j\}$ with $i\neq j$,
we obtain
$H(E|A_{i})_{\psi}= H(E|A_{j})_{\psi}$ for any $i,j$.
\end{proof}

\section{Proof of Theorem~\ref{thm:necessary}} \label{app:PoTnecessary}

To prove Theorem~\ref{thm:necessary},
we use the following lemma.

\begin{Lem} \label{lem:equality}
Let $\ket{\psi}_{AE}$ be the initial state of the QSR task,
and let $\{\mathcal{R}_n\}_{n\in\mathbb{N}}$ be a sequence of
QSR protocols $\mathcal{R}_n$ of $\ket{\psi}^{\otimes n}$
with error $\varepsilon_n$ whose total entanglement rate $r$ is achievable.
If $r=0$,
then 
\begin{equation}
e_P(\psi,\{\mathcal{R}_n\})=l_P(\psi) 
\end{equation}
holds for any non-empty proper subset $P$ of $[M]$,
where $e_{P}$ and $l_P$ are defined in Eq.~(\ref{eq:defBER}) and Eq.~(\ref{eq:lP}), respectively.
\end{Lem}

\begin{proof}
Since $r=0$,
Theorem~\ref{thm:lb} and Remark~\ref{rem:NNOEC} imply
that the lower bound $l_i(\psi)$ in Eq.~(\ref{eq:defLK}) is zero for each $i\in[M]$.

Suppose that there exists a non-empty proper subset $Q$ of $[M]$
such that $e_{Q}(\psi,\{\mathcal{R}_n\})\neq l_{Q}(\psi)$.
Then Lemma~\ref{lem:LBOBER} implies
\begin{equation} \label{eq:suppose}
e_{Q}(\psi,\{\mathcal{R}_n\}) > l_Q(\psi).
\end{equation}
If the size of the set $Q$ is $k$,
then we consider the set $S_k$ of subsets $P_k$ of $[M]$ whose size is $k$,
so that $Q\in S_k$.
Then we can obtain
\begin{equation}
0
=e_{\mathrm{tot}}(\psi,\{\mathcal{R}_n\})
=\frac{1}{n_k}
\sum_{P_k\in S_k}e_{P_k}(\psi,\{\mathcal{R}_n\})
>
\frac{1}{n_k}
\sum_{P_k\in S_k}l_{P_k}(\psi)
=
l_k(\psi)=0,
\end{equation}
where $n_k=2{M-2 \choose k-1}$,
which is a contradiction.
Here,
the second equality and the inequality come from Eqs.~(\ref{eq:Goodequation})
and~(\ref{eq:suppose}),
respectively.
Consequently,
$e_P(\psi,\{\mathcal{R}_n\})=l_P(\psi)$
holds for any non-empty proper subset $P$ of $[M]$.
\end{proof}

\begin{proof}[Proof of Theorem~\ref{thm:necessary}]
Since $r=0$,
Lemma~\ref{lem:equality} implies that
\begin{equation} \label{eq:LEQIN}
\sum_{i\in P} \sum_{j\in P^{\mathsf{c}}} e_{i,j}(\psi,\{\mathcal{R}_n\})
= l_P(\psi)
\end{equation}
holds
for any non-empty proper subset $P$ of $[M]$.
This can be interpreted as the following linear equation,
if we consider the segment entanglement rates $e_{i,j}(\psi,\{\mathcal{R}_n\})$
and the lower bounds $l_P(\psi)$ as unknowns and coefficients:
\begin{equation}
\sum_{i=1}^M\sum_{j=1}^M  c_{i,j}(P) e_{i,j}(\psi,\{\mathcal{R}_n\})
=
l_P(\psi),
\end{equation}
where the coefficient $c_{i,j}(P)$ is defined as
\begin{equation}
c_{i,j}(P)=
\begin{cases}
\frac{1}{2} &\text{if ($i\in P$, $j\in P^{\mathsf{c}}$) or ($j\in P$, $i\in P^{\mathsf{c}}$)} \\
0 &\text{otherwise}.
\end{cases}
\end{equation}
Note that $e_{i,i}(\psi,\{\mathcal{R}_n\})=0$ and $e_{j,i}(\psi,\{\mathcal{R}_n\})=e_{i,j}(\psi,\{\mathcal{R}_n\})$ for each $i,j\in[M]$.
In this way,
we construct a system of linear equations for each case as follows.

(i) For $M=3$,
there exist three unknowns of $e_{i,j}(\psi,\{\mathcal{R}_n\})$.
Consider the sets $P\subseteq[3]$ whose sizes are one. Then, from Eq.~(\ref{eq:LEQIN}), we obtain that
\begin{equation}
 e_{i,i+1}(\psi,\{\mathcal{R}_n\})
+e_{i,i+2}(\psi,\{\mathcal{R}_n\})
=
l_{\{i\}}(\psi)
\end{equation}
for each $i$.
This can be expressed as a system of linear equations as follows:
\begin{equation} \label{eq:LEsPO}
\begin{pmatrix}
1 & 0 & 1 \\
1 & 1 & 0 \\
0 & 1 & 1 \\
\end{pmatrix}
\begin{pmatrix}
e_{1,2}(\psi,\{\mathcal{R}_n\}) \\
e_{2,3}(\psi,\{\mathcal{R}_n\}) \\
e_{1,3}(\psi,\{\mathcal{R}_n\}) \\
\end{pmatrix}
=
\begin{pmatrix}
l_{\{1\}}(\psi) \\
l_{\{2\}}(\psi) \\
l_{\{3\}}(\psi) \\
\end{pmatrix}
.
\end{equation}
Note that if we consider other sets $P$ whose sizes are $k>1$, then we can have a different representation of the linear equations in Eq.~(\ref{eq:LEsPO}).
By simply solving this system,
we obtain
\begin{equation}
e_{i,j}(\psi,\{\mathcal{R}_n\})
=\frac{1}{2}\left(l_{\{i\}}(\psi)+l_{\{j\}}(\psi)-l_{\{k\}}(\psi)\right)\,
\end{equation}
where $\{i,j,k\}=[3]$,
which becomes
$e_{i,j}(\psi,\{\mathcal{R}_n\})=-l_{\{k\}}(\psi)=-l_{\{i,j\}}(\psi)$,
since $l_1(\psi)=0$.

(ii)
Similarly,
the system of linear equations corresponding to the case of $M=4$
can be represented as

\begin{equation}
\begin{pmatrix}
1 & 0 & 0 & 1 & 1 & 0 \\
1 & 1 & 0 & 0 & 0 & 1 \\
0 & 1 & 1 & 0 & 1 & 0 \\
0 & 0 & 1 & 1 & 0 & 1 \\
0 & 1 & 0 & 1 & 1 & 1 \\
1 & 0 & 1 & 0 & 1 & 1 \\
\end{pmatrix}
\begin{pmatrix}
e_{1,2}(\psi,\{\mathcal{R}_n\}) \\
e_{2,3}(\psi,\{\mathcal{R}_n\}) \\
e_{3,4}(\psi,\{\mathcal{R}_n\}) \\
e_{1,4}(\psi,\{\mathcal{R}_n\}) \\
e_{1,3}(\psi,\{\mathcal{R}_n\}) \\
e_{2,4}(\psi,\{\mathcal{R}_n\}) \\
\end{pmatrix}
=
\begin{pmatrix}
l_{\{1\}}(\psi) \\
l_{\{2\}}(\psi) \\
l_{\{3\}}(\psi) \\
l_{\{4\}}(\psi) \\
l_{\{1,2\}}(\psi) \\
l_{\{2,3\}}(\psi) \\
\end{pmatrix}
,
\end{equation}
and its solution is given by
\begin{equation}
e_{i,j}(\psi,\{\mathcal{R}_n\})
=\frac{1}{2}
\left( l_{\{i\}}(\psi) + l_{\{j\}}(\psi) - l_{\{i,j\}}(\psi) \right).
\end{equation}
Recall that, for any partition $\{P,P^{\mathsf{c}}\}$ of the set $[M]$, $l_P(\psi)=l_{P^{\mathsf{c}}}(\psi)$ holds.
For example, $l_{\{1,2\}}(\psi)=l_{\{3,4\}}(\psi)$ when $M=4$.

(iii)
Set $N=M(M-1)/2$.
If $M>4$,
the number of unknowns of $e_{i,j}(\psi,\{\mathcal{R}_n\})$ is $N$.
In this case,
it suffices to consider subsets $P_2$ of $[M]$
whose size is two
in order to construct a system of linear equations.
To be specific,
there exist $N$ different linear equations
\begin{equation} \label{eq:SpecifiSLE}
\sum_{i=1}^M\sum_{j=1}^M  c_{i,j}(P_2) e_{i,j}(\psi,\{\mathcal{R}_n\})
=
l_{P_2}(\psi),
\end{equation}
so we have a system of $N$ linear equations
with $N$ unknowns.
This system of linear equations can be represented as
a matrix equation of the form
\begin{equation}
D_M \bold{x}_M =\bold{b}_M,
\end{equation}
where the matrix $D_M$ is $N$ by $N$,
and the matrices $\bold{x}_M$ and $\bold{b}_M$ are $N$ by $1$.
To describe entries of these matrices,
we use a bijective function $f_M\colon[N]\to T_M$,
where $T_M$ is the set of all two-element subsets $P_2$ of $[M]$.
Then the entries of the matrices $D_M$, $ \bold{x}_M$, and $\bold{b}_M$
are given by
\begin{eqnarray} \label{eq:DefDM}
\left[ D_M \right]_{s,t}
&=&
\begin{cases}
0 &\text{if $s=t$} \\
1 &\text{if $s\neq t$ and $f_M(s)\cap f_M(t)\neq \emptyset$} \\
0 &\text{if $s\neq t$ and $f_M(s)\cap f_M(t)   = \emptyset$},
\end{cases} \\
\left[ \bold{x}_M \right]_{s,1}
&=&
e_{f_M(s)}(\psi,\{\mathcal{R}_n\}
), \\
\left[ \bold{b}_M \right]_{s,1}
&=&
 l_{f_M(s)}(\psi),
\end{eqnarray}
where $\left[ D_M \right]_{s,t}$ is derived from
the coefficients $c_{i,j}(P_2)$ in Eq.~(\ref{eq:SpecifiSLE}),
and $e_{f_M(s)}(\psi,\{\mathcal{R}_n\})$
indicates
the segment entanglement rate $e_{i_sj_s}(\psi,\{\mathcal{R}_n\})$
if $f_M(s)=\{i_s,j_s\}\subset[M]$.

Now,
we show that the matrix $D_M$ is invertible.
Consider an $N$ by $N$ matrix $D_M^{-1}$ defined as
\begin{equation} \label{eq:DefDMInverse}
\left[ D_M^{-1} \right]_{s,t}
=
\begin{cases}
\beta_M/\alpha_M &\text{if $s=t$} \\
\gamma_M/\alpha_M &\text{if $s\neq t$ and $f_M(s)\cap f_M(t)\neq \emptyset$} \\
-2/\alpha_M &\text{if $s\neq t$ and $f_M(s)\cap f_M(t)   = \emptyset$},
\end{cases}
\end{equation}
where $\alpha_M=2(M-2)(M-4)$, $\beta_M=2-(M-4)^2$, and $\gamma_M=M-4$.
For each $s\neq t$, define subsets $T_{s,s}^{(1)}$ and $T_{s,t}^{(i)}$ of $[N]$  as follows:
\begin{eqnarray}
T_{s,s}^{(1)}
&=&
\left\{ k\in [N] : k\neq s, |f_M(k)\cap f_M(s)|=1 \right\}, \\
T_{s,t}^{(2)}
&=&
\left\{ k\in [N] : s=k, t\neq k, f_M(t)\cap f_M(k)\neq \emptyset \right\}, \\
T_{s,t}^{(3)}
&=&
\left\{ k\in [N] : s\neq k, f_M(s)\cap f_M(k)\neq\emptyset, t\neq k,
f_M(t)\cap f_M(k)\neq \emptyset \right\}, \\
T_{s,t}^{(4)}
&=&
\left\{ k\in [N] : s\neq k, f_M(s)\cap f_M(k)=\emptyset, t\neq k, 
f_M(t)\cap f_M(k)\neq \emptyset \right\}.
\end{eqnarray}
The sizes of these sets are given by
\begin{eqnarray}
\left| T_{s,s}^{(1)} \right|
&=&
2(M-2), \\
\left| T_{s,t}^{(2)} \right|
&=&
\begin{cases}
1 &\text{if $f_M(s)\cap f_M(t)\neq \emptyset$} \\
0 &\text{otherwise},
\end{cases} \\
\left| T_{s,t}^{(3)} \right|
&=&
\begin{cases}
M-2 &\text{if $f_M(s)\cap f_M(t)\neq \emptyset$} \\
4 &\text{otherwise},
\end{cases} \\
\left| T_{s,t}^{(4)} \right|
&=&
\begin{cases}
M-3 &\text{if $f_M(s)\cap f_M(t)\neq \emptyset$} \\
2(M-4) &\text{otherwise}.
\end{cases}
\end{eqnarray}
We obtain that
the diagonal entries of the matrix $D_M^{-1}D_M$ are
\begin{equation}
\left[ D_M^{-1}D_M \right]_{s,s}
= \sum_{k=1}^{N} \left[ D_M^{-1} \right]_{s,k} \left[ D_M \right]_{s,k}
= \frac{\gamma_M}{\alpha_M}
\left| T_{s,s}^{(1)} \right|=\frac{\gamma_M}{\alpha_M}2(M-2) =1.
\end{equation}
Since the matrix $D_M$ is symmetric,
the first equality holds,
and by directly comparing Eqs.~(\ref{eq:DefDM}) and~(\ref{eq:DefDMInverse})
we obtain the second equality.
On the other hand,
observe that the equality
\begin{equation}
\left[ D_M^{-1} \right]_{s,k} \left[ D_M \right]_{k,t}
=
\begin{cases}
\beta_M/\alpha_M &\text{if $k\in T_{s,t}^{(2)}$} \\
\gamma_M/\alpha_M &\text{if $k\in T_{s,t}^{(3)}$} \\
-2/\alpha_M &\text{if $k\in T_{s,t}^{(4)}$} \\
0 &\text{otherwise}
\end{cases}
\end{equation}
holds for any $s,t,k\in[N]$ with $s\neq t$.
From the above equation,
the off-diagonal entries of the matrix $D_M^{-1}D_M$ are calculated as
\begin{equation}
\left[ D_M^{-1}D_M \right]_{s,t}
= \sum_{k=1}^{N} \left[ D_M^{-1} \right]_{s,k} \left[ D_M \right]_{k,t}
= \frac{1}{\alpha_M}
\left(
\beta_M\left| T_{s,t}^{(2)} \right|+\gamma_M\left| T_{s,t}^{(3)} \right|-2\left| T_{s,t}^{(4)} \right|
\right)
=0.
\end{equation}

This shows that the matrix $D_M^{-1}$ is the inverse of the matrix $D_M$,
and so $\bold{x}_M =D_M^{-1}\bold{b}_M$.
\end{proof}

\section{Proof of Lemma~\ref{lem:GHZtoSinglet}} \label{app:PoL}

To prove Lemma~\ref{lem:GHZtoSinglet},
we use the relative entropy of entanglement~\cite{DH99} between the $2^{\mathrm{nd}}$ user and the $3^{\mathrm{rd}}$ user
instead of the entanglement entropy between each user and the other two users,
since the entanglement entropies for the initial and final states are the same.

Suppose that
there exists a sequence $\{\mathcal{T}_n\}_{n\in\mathbb{N}}$
of LOCC $\mathcal{T}_n$ of $\phi^{\otimes n}$
with error $\varepsilon_n $ such that
$e_{i,j}(\phi,\{\mathcal{T}_n\})=0$ for each $i,j$
and $\lim_{n\to\infty} \varepsilon_n=0$,
where the segment entanglement rate $e_{i,j}$ is defined in Eq.~(\ref{eq:SER}).
From the monotonicity of the trace distance~\cite{W13},
we obtain that
\begin{equation}
\delta_n=\left\|
 \Tr_{{A'_1}^{\otimes n}{A'_2}^{\otimes n}F^{(n)}_{1,2}F^{(n)}_{1,3}}
\left[\phi_{\mathrm{f}}^{\otimes n} \otimes \tilde{\Phi}_n\right]
-\Tr_{{A'_1}^{\otimes n}{A'_2}^{\otimes n}F^{(n)}_{1,2}F^{(n)}_{1,3}}
\left[\mathcal{T}_n \left( \phi^{\otimes n} \otimes \tilde{\Psi}_n \right)\right]
\right\|_{1}
\le
\varepsilon_n.
\end{equation}
Let $D(\varrho\|\tau)$ be the quantum relative entropy
between two mixed states $\varrho$ and $\tau$,
i.e., $D(\varrho\|\tau)=\Tr[\varrho(\log \varrho-\log \tau)]$.
Then the relative entropy of entanglement of $\varrho_{XY}$ is defined by
\begin{equation}
E_R(X;Y)_{\varrho}
=\min_{\tau_{XY}\in \mathrm{SEP}(X;Y)}
D(\varrho_{XY}\|\tau_{XY}),
\end{equation}
where $\mathrm{SEP}(X;Y)$ is the set of all separable states on the system $XY$.
From the continuity of the relative entropy of entanglement~\cite{DH99},
if $\delta_n\le 1/3$,
then we have
\begin{eqnarray}
&&2
\left(
\delta_n (n \log d_{B'_1B'_2C'_1C'_2}+\log d_{F^{(n)}_{2,1}F^{(n)}_{2,3}F^{(n)}_{3,1}F^{(n)}_{3,2}})
-\delta_n\log \delta_n
\right)
+
4\delta_n \\
&&\ge
\left|
 E_R({B'_1}^{\otimes n}{B'_2}^{\otimes n}F^{(n)}_{2,1}F^{(n)}_{2,3};{C'_1}^{\otimes n}{C'_2}^{\otimes n}F^{(n)}_{3,1}F^{(n)}_{3,2})_{\Tr_{{A'_1}^{\otimes n}{A'_2}^{\otimes n}F^{(n)}_{1,2}F^{(n)}_{1,3}}
\left[\phi_{\mathrm{f}}^{\otimes n} \otimes \tilde{\Phi}_n\right]} \right. \\
&&\quad
\left.
-E_R({B'_1}^{\otimes n}{B'_2}^{\otimes n}F^{(n)}_{2,1}F^{(n)}_{2,3};{C'_1}^{\otimes n}{C'_2}^{\otimes n}F^{(n)}_{3,1}F^{(n)}_{3,2})_{\Tr_{{A'_1}^{\otimes n}{A'_2}^{\otimes n}F^{(n)}_{1,2}F^{(n)}_{1,3}}
\left[\mathcal{T}_n \left( \phi^{\otimes n} \otimes \tilde{\Psi}_n \right)\right]}
\right| \nonumber \\
&&\ge
 E_R({B'_1}^{\otimes n}{B'_2}^{\otimes n}F^{(n)}_{2,1}F^{(n)}_{2,3};{C'_1}^{\otimes n}{C'_2}^{\otimes n}F^{(n)}_{3,1}F^{(n)}_{3,2})_{\Tr_{{A'_1}^{\otimes n}{A'_2}^{\otimes n}F^{(n)}_{1,2}F^{(n)}_{1,3}}
\left[\phi_{\mathrm{f}}^{\otimes n} \otimes \tilde{\Phi}_n\right]} \label{eq:b1} \\
&&\quad
-E_R({B_1}^{\otimes n}{B_2}^{\otimes n}D^{(n)}_{2,1}D^{(n)}_{2,3};{C_1}^{\otimes n}{C_2}^{\otimes n}D^{(n)}_{3,1}D^{(n)}_{3,2})_{\Tr_{A_1^{\otimes n}A_2^{\otimes n}D^{(n)}_{1,2}D^{(n)}_{1,3}}
\left[\phi^{\otimes n} \otimes \tilde{\Psi}_n\right]}, \nonumber
\end{eqnarray}
where the second inequality comes from the fact that
the relative entropy of entanglement cannot increase under LOCC~\cite{W13}.
It is easy to check that two equalities
\begin{eqnarray}
\Tr_{A_1^{\otimes n}A_2^{\otimes n}D^{(n)}_{1,2}D^{(n)}_{1,3}}
\left[\phi^{\otimes n} \otimes \tilde{\Psi}_n\right]
&=& J_{B_1C_1}^{\otimes n}
\otimes
J_{B_2C_2}^{\otimes n}
\otimes I_n(2,1)_{D^{(n)}_{2,1}}\otimes \Psi^{(n)}_{2,3}
\otimes I_n(3,1)_{D^{(n)}_{3,1}}, \\
\Tr_{{A'_1}^{\otimes n}{A'_2}^{\otimes n}F^{(n)}_{1,2}F^{(n)}_{1,3}}
\left[\phi_{\mathrm{f}}^{\otimes n} \otimes \tilde{\Phi}_n\right]
&=& I_{B'_2}^{\otimes n}
\otimes
\left(\ketbra{\mathrm{ebit}}{\mathrm{ebit}}\right)_{B'_1C'_2}^{\otimes n}
\otimes
I_{C'_1}^{\otimes n}
\otimes I'_n(2,1)_{F^{(n)}_{2,1}}\otimes \Phi^{(n)}_{2,3}
\otimes I'_n(3,1)_{F^{(n)}_{3,1}}
\end{eqnarray}
hold.
Here, the mixed states $J$, $I_n(i,j)$, $I'_n(i,j)$, and $I$ are
\begin{equation}
J = \frac{1}{2}\left( \ketbra{00}{00}+\ketbra{11}{11} \right),
\quad
I_n(i,j)
= \frac{1}{d_{D^{(n)}_{i,j}}}\sum_{j=0}^{d_{D^{(n)}_{i,j}}-1} \ketbra{j}{j},
\quad
I'_n(i,j)
= \frac{1}{d_{F^{(n)}_{i,j}}}\sum_{j=0}^{d_{F^{(n)}_{i,j}}-1} \ketbra{j}{j},
\quad
I = \frac{1}{2}\left( \ketbra{0}{0}+\ketbra{1}{1} \right),
\end{equation}
where $d_{D^{(n)}_{i,j}}$ ($d_{F^{(n)}_{i,j}}$) is the Schmidt rank of the entanglement resource $\Psi^{(n)}_{i,j}$ $(\Phi^{(n)}_{i,j})$ on the quantum systems $D^{(n)}_{i,j}D^{(n)}_{j,i}$ ($F^{(n)}_{i,j}F^{(n)}_{j,i}$)
shared by the $i^{\mathrm{th}}$ user and the $j^{\mathrm{th}}$ user before (after) performing the QSR protocol $\mathcal{T}_n$.
So we obtain 
\begin{eqnarray} \label{eq:b2}
&&E_R({B_1}^{\otimes n}{B_2}^{\otimes n}D^{(n)}_{2,1}D^{(n)}_{2,3};{C_1}^{\otimes n}{C_2}^{\otimes n}D^{(n)}_{3,1}D^{(n)}_{3,2})_{\Tr_{A_1^{\otimes n}A_2^{\otimes n}D^{(n)}_{1,2}D^{(n)}_{1,3}}
\left[\phi^{\otimes n} \otimes \tilde{\Psi}_n\right]} \\
&&\le
n E_R(B_1B_2;C_1C_2)_{\Tr_{A_1A_2}\left[\phi\right]}
+
E_R(D^{(n)}_{2,1}D^{(n)}_{2,3};D^{(n)}_{3,1}D^{(n)}_{3,2})_{\Tr_{D^{(n)}_{1,2}D^{(n)}_{1,3}}
\left[\tilde{\Psi}_n\right]} \\
&&=E_R(D^{(n)}_{2,1}D^{(n)}_{2,3};D^{(n)}_{3,1}D^{(n)}_{3,2})_{\Tr_{D^{(n)}_{1,2}D^{(n)}_{1,3}}
\left[\tilde{\Psi}_n\right]}.
\end{eqnarray}
In the above,
the first inequality comes from the subadditivity~\cite{VW01} of the relative entropy of entanglement.
The last equality holds,
since $\Tr_{A_1A_2}\left[\phi\right]$ is separable.

By discarding systems ${B'_2}^{\otimes n}F^{(n)}_{2,1}$
and ${C'_1}^{\otimes n}F^{(n)}_{3,1}$,
we have
\begin{equation}
E_R({B'_1}^{\otimes n}{B'_2}^{\otimes n}F^{(n)}_{2,1}F^{(n)}_{2,3};{C'_1}^{\otimes n}{C'_2}^{\otimes n}F^{(n)}_{3,1}F^{(n)}_{3,2})_{\Tr_{{A'_1}^{\otimes n}{A'_2}^{\otimes n}F^{(n)}_{1,2}F^{(n)}_{1,3}}
\left[\phi_{\mathrm{f}}^{\otimes n} \otimes \tilde{\Phi}_n\right]}
\ge
E_R({B'_1}^{\otimes n}F^{(n)}_{2,3};{C'_2}^{\otimes n}F^{(n)}_{3,2})_{\ket{\mathrm{ebit}}_{B'_1C'_2}^{\otimes n}\otimes\Phi^{(n)}_{2,3}}.
\end{equation}
In addition,
Bob and Charlie can locally prepare the quantum states $I_{B'_2}^{\otimes n} \otimes I'_n(2,1)_{F^{(n)}_{2,1}}$
and $I_{C'_1}^{\otimes n} \otimes I'_n(3,1)_{F^{(n)}_{3,1}}$,
respectively.
It follows that
\begin{equation}
E_R({B'_1}^{\otimes n}F^{(n)}_{2,3};{C'_2}^{\otimes n}F^{(n)}_{3,2})_{\ket{\mathrm{ebit}}_{B'_1C'_2}^{\otimes n}\otimes\Phi^{(n)}_{2,3}}
\ge
E_R({B'_1}^{\otimes n}{B'_2}^{\otimes n}F^{(n)}_{2,1}F^{(n)}_{2,3};{C'_1}^{\otimes n}{C'_2}^{\otimes n}F^{(n)}_{3,1}F^{(n)}_{3,2})_{\Tr_{{A'_1}^{\otimes n}{A'_2}^{\otimes n}F^{(n)}_{1,2}F^{(n)}_{1,3}}
\left[\phi_{\mathrm{f}}^{\otimes n} \otimes \tilde{\Phi}_n\right]}.
\end{equation}
From the fact that $E_R(X;Y)_{\varrho}=H(X)_{\varrho}$ holds for any pure state $\varrho_{XY}$,
we have
\begin{equation} \label{eq:b3}
E_R({B'_1}^{\otimes n}F^{(n)}_{2,3};{C'_2}^{\otimes n}F^{(n)}_{3,2})_{\ket{\mathrm{ebit}}_{B'_1C'_2}^{\otimes n}\otimes\Phi^{(n)}_{2,3}}
=H({B'_1}^{\otimes n}F^{(n)}_{2,3})_{I_{{B'_1}^{\otimes n}}^{\otimes n}\otimes{I'_n(2,3)}_{F_{2,3}}}
=n+\log d_{F^{(n)}_{2,3}}.
\end{equation}
Similarly,
we have
\begin{equation} \label{eq:b4}
E_R(D^{(n)}_{2,1}D^{(n)}_{2,3};D^{(n)}_{3,1}D^{(n)}_{3,2})_{\Tr_{D^{(n)}_{1,2}D^{(n)}_{1,3}}
\left[\tilde{\Psi}\right]}
=\log d_{D^{(n)}_{2,3}}.
\end{equation}
By using Eqs.~(\ref{eq:b2}),~(\ref{eq:b3}), and~(\ref{eq:b4}),
Eq.~(\ref{eq:b1}) becomes
\begin{equation}
2\delta_n \log d_{B'_1B'_2C'_1C'_2}
+
\frac{\delta_n}{n}\left(2\log d_{F^{(n)}_{2,1}F^{(n)}_{2,3}F^{(n)}_{3,1}F^{(n)}_{3,2}}-2\log \delta_n + 4 \right)
+
\frac{1}{n}\left(\log d_{D^{(n)}_{2,3}}-\log d_{F^{(n)}_{2,3}} \right)
\ge 1.
\end{equation}
As $n\rightarrow\infty$,
this inequality becomes $0=e_{2,3}(\phi,\{\mathcal{T}_n\})\ge1$,
which is a contradiction.
Therefore,
it is impossible to transform
two GHZ states shared by Alice, Bob, and Charlie
into three ebits symmetrically shared among them via LOCC,
even under the catalytic use of entanglement resource.

If we consider a non-asymptotic scenario in which users begin this transformation with finite copies of the initial states, whether the transformation is possible or not under the non-asymptotic scenario with the catalytic use of entanglement is unknown, to the best of our knowledge.

\vfill \clearpage
\end{widetext}

\end{document}